\documentclass[a4paper]{amsart}

\usepackage{amssymb,pstricks,amscd,epsfig}
\usepackage{amssymb,amscd,epsfig,color}
\usepackage[totalwidth=17cm,totalheight=24cm]{geometry}
\usepackage{graphicx}
\begin{document}
\newtheorem{cor}{Corollary}[section]
\newtheorem{theorem}[cor]{Theorem}
\newtheorem{prop}[cor]{Proposition}
\newtheorem{lemma}[cor]{Lemma}
\newtheorem{sublemma}[cor]{Sublemma}
\theoremstyle{definition}
\newtheorem{defi}[cor]{Definition}
\theoremstyle{remark}
\newtheorem{remark}[cor]{Remark}
\newtheorem{example}[cor]{Example}

\newcommand{\cD}{{\mathcal D}}
\newcommand{\cF}{{\mathcal F}}
\newcommand{\cM}{{\mathcal M}}
\newcommand{\cN}{{\mathcal N}}
\newcommand{\cT}{{\mathcal T}}
\newcommand{\cCP}{{\mathcal C\mathcal P}}
\newcommand{\cML}{{\mathcal M\mathcal L}}
\newcommand{\cFML}{{\mathcal F\mathcal M\mathcal L}}
\newcommand{\cGH}{{\mathcal G\mathcal H}}
\newcommand{\cQF}{{\mathcal Q\mathcal F}}
\newcommand{\C}{{\mathbb C}}
\newcommand{\N}{{\mathbb N}}
\newcommand{\R}{{\mathbb R}}
\newcommand{\Z}{{\mathbb Z}}
\newcommand{\At}{\tilde{A}}
\newcommand{\Kt}{\tilde{K}}
\newcommand{\Mt}{\tilde{M}}
\newcommand{\dr}{{\partial}}
\newcommand{\betab}{\overline{\beta}}
\newcommand{\kappab}{\overline{\kappa}}
\newcommand{\pib}{\overline{\pi}}
\newcommand{\taub}{\overline{\tau}}
\newcommand{\ub}{\overline{u}}
\newcommand{\Sigmab}{\overline{\Sigma}}
\newcommand{\gd}{\dot{g}}
\newcommand{\diff}{\mbox{Diff}}
\newcommand{\dev}{\mbox{dev}}
\newcommand{\devb}{\overline{\mbox{dev}}}
\newcommand{\devt}{\tilde{\mbox{dev}}}
\newcommand{\vol}{\mbox{Vol}}
\newcommand{\hess}{\mathrm{Hess}}
\newcommand{\cb}{\overline{c}}
\newcommand{\db}{\overline{\partial}}
\newcommand{\Sigmat}{\tilde{\Sigma}}

\newcommand{\cunc}{{\mathcal C}^\infty_c}
\newcommand{\cun}{{\mathcal C}^\infty}
\newcommand{\dd}{d_D}
\newcommand{\dmin}{d_{\mathrm{min}}}
\newcommand{\dmax}{d_{\mathrm{max}}}
\newcommand{\Dom}{\mathrm{Dom}}
\newcommand{\dn}{d_\nabla}
\newcommand{\ded}{\delta_D}
\newcommand{\delmin}{\delta_{\mathrm{min}}}
\newcommand{\delmax}{\delta_{\mathrm{max}}}
\newcommand{\hyp}{\mathbb{H}}
\newcommand{\hmin}{H_{\mathrm{min}}}
\newcommand{\maxi}{\mathrm{max}}
\newcommand{\oL}{\overline{L}}
\newcommand{\oP}{{\overline{P}}}
\newcommand{\xb}{{\overline{x}}}
\newcommand{\yb}{{\overline{y}}}
\newcommand{\Ran}{\mathrm{Ran}}
\newcommand{\tgamma}{\tilde{\gamma}}
\newcommand{\cotan}{\mbox{cotan}}
\newcommand{\area}{\mbox{Area}}
\newcommand{\lambdat}{\tilde\lambda}
\newcommand{\xt}{\tilde x}
\newcommand{\Ct}{\tilde C}
\newcommand{\St}{\tilde S}
\newcommand{\tr}{\mbox{\rm tr}}
\newcommand{\tgh}{\mbox{th}}
\newcommand{\sh}{\mathrm{sinh}\,}
\newcommand{\ch}{\mathrm{cosh}\,}
\newcommand{\grad}{\mathrm{grad}}
\newcommand{\Iso}{\mathrm{Isom}}
\newcommand{\hc}{H\C^3}
\newcommand{\cp}{\C P^1}
\newcommand{\rp}{\R P^1}
\newcommand{\ret}{\mathbf r}
\newcommand{\fut}{\mathrm{I}^+}
\newcommand{\past}{\mathrm{I}^-}

\newcommand{\II}{I\hspace{-0.1cm}I}
\newcommand{\III}{I\hspace{-0.1cm}I\hspace{-0.1cm}I}
\newcommand{\note}[1]{{\small {\color[rgb]{1,0,0} #1}}}

\newcommand{\ZZ}{\mathbb{Z}}
\newcommand{\QQ}{\mathbb{Q}}
\newcommand{\RR}{\mathbb{R}}
\newcommand{\CC}{\mathbb{C}}
\newcommand{\MM}{\mathbb{M}}
\newcommand{\NN}{\mathbb{N}}

\newcommand{\skp}[2]{\langle {{#1}} , {{#2}}\rangle}

%\title{Determining the initial singularity in flat space-times}
\title{Recovering the geometry of a flat spacetime from background radiation}

\author{Francesco Bonsante}
\address{FB: Universit\`a degli Studi di Pavia\\
Via Ferrata, 1\\
27100 Pavia, Italy}
\email{bonsante@sns.it}
\thanks{F.B. is partially supported by the A.N.R. through project Geodycos.}
\author{Catherine Meusburger}
\address{CM: Department Mathematik,  Friedrich-Alexander Universit\"at  Erlangen-N\"urnberg, Cauerstr.~11, 91058 Erlangen, Germany}
\email{catherine.meusburger@math.uni-erlangen.de}
\thanks{C.M. was supported by the DFG Emmy-Noether fellowship ME 3425/1-1.}
\author{Jean-Marc Schlenker}
\thanks{J.-M. S. was partially supported by the A.N.R. through projects
GeomEinstein, ANR-09-BLAN-0116-01 and ETTT, ANR-09-BLAN-0116-01, 2009-13.}
\address{JMS: University of Luxembourg, Campus Kirchberg,
Mathematics Research Unit, BLG 6, rue Richard Coudenhove-Kalergi,
L-1359 Luxembourg}
\email{jean-marc.schlenker@uni.lu}
\date{v2, June 2013}

\begin{abstract}
We consider  globally hyperbolic flat spacetimes in 2+1 and 3+1 dimensions, in which a uniform light signal is
emitted on the $r$-level surface of the cosmological time for $r\to 0$.
We show that the frequency of this signal, as perceived by a fixed observer,
is a well-defined, bounded function which is
generally not continuous.  This defines a  model with anisotropic background radiation that contains information about initial singularity of the spacetime. 
In dimension $2+1$, we show that this observed frequency
function is stable under suitable perturbations of the spacetime, and that,
under certain conditions, it contains sufficient information to recover
its geometry and topology. We compute an approximation
of this frequency function for  a few simple examples.
\end{abstract}

\maketitle

\tableofcontents

\section{Introduction}

%:bms1.tex

\subsection{Motivation}

There is considerable interest in  the cosmic  background radiation
as an indicator of the history and structure of the universe.
Its anisotropy is explained by quantum fluctuations
early in the history of the universe, whose classical remnants  became visible
 when the universe became transparent to electromagnetic radiation after decoupling. Background radiation has also been used  to determine the topology of the universe
  \cite{Gott:1986uz, cornish, gott, luminet:2003}. 
  
  In the framework of general relativity, such measurements of cosmic background radiation are described by a spacetime  in which light signals are emitted near the initial singularity and received by an observer.   
  The question to what degree  the geometry of a  spacetime can be reconstructed from such  light signals %emitted near its initial singularity 
  is of interest both from the mathematics and the physics perspective. However, it has not been investigated systematically yet,  even for simple
  examples such as constant curvature spacetimes  or lower-dimensional models.

A further motivation to investigate the  properties of such light signals are possible applications  in 2+1 (and higher-dimensional) quantum gravity. 
To give a physical interpretation to a quantum theory of gravity, it is essential to relate 
 the variables that describe the physical  phase space of the theory and serve as the basic building blocks in quantization to concrete measurements that could be performed by an observer.
 From a mathematics perspective, this amounts to a (partial) classification of spacetimes  in terms of light signals and other general relativistic quantities measured by observers or, equivalently, to reconstructing the geometry of the spacetime from such measurements.

In this article, we investigate this question  for a  class of simple examples, namely flat globally hyperbolic spacetimes
in  3+1 and 2+1 dimensions. We consider 
a uniform  light signal that is emitted from a hypersurface of constant cosmological
time $\epsilon$  where  $\epsilon \to 0$.  
The observer  who receives the signal at a cosmological time $T>\epsilon$ is modeled by a point % a point
$p\in M$ of cosmological time $T$  and a unit, future-oriented timelike vector $v$, which specifies his velocity.  
The frequency of the light signal measured by the observer then defines  a frequency function on, respectively,  $S^2$ and $S^1$, which depends on both the spacetime $M$ and the observer $(p,v)$.

If the spacetime $M$ is conformally static, i.~e.~characterized by  a linear 
 holonomy representation, then the associated frequency function is isotropic,
and  contains no relevant  information on $M$. However,  as soon as the holonomy
representation of $M$ has a non-trivial translation component, which corresponds to a universe whose geometry changes with the cosmological time, 
the frequency  function  contains essential  information about the underlying spacetime. Under certain
conditions, this information allows the observer to recover the geometry and topology of $M$ as well as his motion  relative to the initial singularity.

The models under consideration in this article  are  unrealistic insofar as they are purely classical --- we do not consider any quantum fluctuations near the initial singularity --- and as their initial singularities  are  not of the same type as those in cosmological models.  However,  they  allow one to investigate the mathematical properties of such light signals emitted near the initial singularity systematically for a large class of spacetimes.  Moreover, it  turns out that
 the resulting frequency functions have rich and subtle properties and contain interesting information about the underlying spacetimes.  In particular, they exhibit anisotropies, which are due  entirely to the classical geometry of the spacetime $M$ or, more specifically, the initial singularity of the universal cover
of $M$, which is a domain of dependence in Minkowski space. 

The model spaces considered here could actually be more relevant than they might appear at first sight.
Indeed, Mess \cite{mess} proved that {\it any} maximal globally hyperbolic Minkowski space of dimension 2+1
is of this form. Moreover, any globally hyperbolic flat space-like slice must embed isometrically into one of
those maximal globally hyperbolic Minkowski space-times. So, if we consider for instance a universe $U$
with a very strong curvature at small cosmological time, but suppose that this curvature decreases fast
enough so that it can be considered zero outside a neighborhood $\Omega$ of the initial singularity, then 
$U\setminus \Omega$ will embed in a globally hyperbolic Minkowski manifold as studied here. The presence
of the curvature might of course change the signal emitted close to the initial singularity but, if 
the curved region $\Omega$ is thin enough, it is conceivable that the model used here remains relevant.

An additional reason to investigate flat globally hyperbolic spacetimes  is  their role in 2+1 gravity. This theory plays an essential role as a toy model for quantum gravity in higher dimensions (see \cite{carlipbook} and references therein) because it allows one to investigate  important questions of quantum gravity in a fully quantized theory. 
The  classification result by Mess \cite{mess} implies  that any globally hyperbolic vacuum solution  Einstein's equations in 2+1 dimensions is a flat globally hyperbolic spacetime of the type considered in this article and can be characterized in terms of its holonomy representation. The holonomies associated with closed curves in $M$ are diffeomorphism invariant observables and serve as
the fundamental building blocks in the quantization of the theory.

Characterizing the holonomy representation of a spacetime in terms of light signals measured by an observer thus  allows one to give a physical interpretation to  these variables  and to model cosmological measurements. The characterization
of  observables of 2+1 gravity in terms  of light signals has been explored to some degree 
in \cite{meusburger:cosmological}, but the  measurements of background radiation considered in this article provide a richer and more realistic model.

\subsection{The frequency function}

After  recalling the relevant
 background material on flat maximally globally hyperbolic spacetimes and their description in terms of domains of dependence in   Section  \ref{sc:background}, we  introduce
 the  rescaled frequency function of a domain of dependence $\Mt$ in Section \ref{sc:light}. 
The rescaled frequency function is defined 
 in 
Section \ref{ssc:defi}. It is given  in terms of the limit  $\epsilon\to 0$ of a uniform  light signal emitted from the
hypersurface of cosmological time $T=\epsilon$ of $\Mt$ and received by a free-falling observer in the spacetime. Sections \ref{ssc:ex1} and \ref{ssc:ex2} contain an explicit description of the frequency function for the two basic examples, namely 
the future of a point and the future of a spacelike line.  

These two basic examples  are the  building blocks in the analysis of the frequency function for a general domain $\Mt$, whose properties are investigated in 
Section \ref{ssc:basic}. The central result is Proposition \ref{pr:conv-intensity}, which  asserts that the frequency function is well-defined and
locally bounded. It should be stressed that even at this  point, the mathematical analysis
of the frequency function is not as simple as it may appear at first sight, and some care is needed.
The situation is simpler for domains $\Mt$ whose initial singularity  is closed, which are investigated  in Section \ref{ssc:closedsing}.  In this case,  the associated frequency function is
continuous. Note, however, that many relevant examples are not of this type.

Section \ref{ssc:generic21} analyses the properties of the frequency functions for  generic domains in 2+1  dimensions. This case is more accessible than its 3+1-dimensional counterpart due to a simple description of
domains of dependence, discovered by Mess \cite{mess}, in terms of a measured lamination 
on the hyperbolic plane. Domains of dependence, which are the universal covers of globally
hyperbolic flat spacetimes, are obtained from  measured geodesic laminations on closed
hyperbolic surfaces.  By using these results we prove (Proposition \ref{pr:meagre}) that the frequency
function is lower semi-continuous, and that its discontinuity set is meagre.

\subsection{Stability}

In Section \ref{sc:stability} we investigate  the stability properties of the frequency function.
We analyze the variation of the frequency function under small deformations 
 of the domain of dependence $\Mt$  and changes of the observer. 
Note that stability of the frequency function, at least with respect to small changes of the observer, is a minimum requirement for  assigning any  physical meaning to it.  
 
Again, this question turns out to be more subtle than it appears and 
some care is needed in the analysis. This is illustrated in Section \ref{ssc:example}, where
we show in a very simple example that if a domain of dependence $\Mt$ is the limit (in the
Hausdorff sense) of a decreasing sequence of {\it finite domains} $\Mt_n$ (domains which are the 
intersection of the futures of a finite set of lightlike planes) the frequency function of $\Mt$ does not necessarily coincide with  the limit of the frequency functions of the domains $\Mt_n$. 

With this example in mind, we  introduce in Section \ref{ssc:flat} a notion of domain
of dependence with a {\it flat boundary}.  Important examples of this are 
 finite domains, which always have flat boundary, and universal covers of globally
hyperbolic flat spacetimes in 2+1 dimensions  (Proposition \ref{pr:2+1flat}).
We prove (Theorem \ref{stability:thrm}) that if a sequence of domains with flat boundary
%(e.g. finite domains) 
converges to a domain with flat boundary, then the limit of the frequency functions the frequency function  of the limit.

If a domain $D$ does not have flat boundary, and if $(D_n)_{n\in \N}$ is sequence of domains with
flat boundary converging to $D$, then the sequences of the frequency functions $\iota_n$ of $D_n$ always converges
to a limit frequency function $\iota_{lim}$. This limit frequency function is not necessarily equal to the frequency function $\iota$ of $D$, but it is
independent of the sequence $(D_n)_{n\in\N}$. We prove (Theorem \ref{tm:non-flat}) that frequency  function $\iota$ is
at least equal to the limit frequency function  $\iota_{lim}$, and at most equal to $d\,\iota_{lim}$, where $d$ is the dimension of the spacetime. In particular, this
 shows that the example of Section \ref{ssc:example}  exhibits the worst possible behavior with respect to this limit, as  the ratio of the frequency function of the domain to its limit frequency is the largest
possible in dimension $2+1$.

\subsection{Recovering the spacetime geometry and topology}

In Section \ref{sc:2+1} we turn to the question of reconstructing the geometry and topology 
of a globally hyperbolic flat spacetime $M$ from the frequency  function of the background
radiation as seen by an observer.  We investigate this question in 2+1 dimensions, and there 
 are two basic remarks regarding the general situation.
\begin{itemize}
\item  Reconstructing the geometry or topology of the spacetime from the observed frequency function is only possible for observers that receive the light signal  at a sufficiently  large cosmological
time. If the observer is too close to the initial singularity,  he might see
only a small part of $M$ %, for instance lying ``above'' a segment of the initial singularity, 
and could  infer little from the observed background radiation.
\item The observer can only determine parts of  the initial singularity of the universal cover $\Mt$ of $M$. Therefore, there is no way for the observer to be sure, 
at any given time, that what he observes is really the topology of $M$. It could happen that
$M$ is ``almost'' a finite cover of a globally hyperbolic flat spacetime $M'$, with only a
tiny difference in a part not ``seen'' by the observer. In this case the observer could
only conclude that the spacetime is either $M'$ or one of its finite covers.
\end{itemize}

To obtain results we make a (presumably) technical assumption, which simplifies the 
situation to some extend. We only
consider spacetimes obtained by ``grafting'' a hyperbolic surface along a {\it rational}
measured lamination, that is, a measured lamination with support on a finite set of 
simple closed curves. Under those hypothesis, we prove (Proposition \ref{pr:rational}) that
the observer can reconstruct the part of the lamination corresponding to the part
of the initial singularity that intersects his past. We also prove (Proposition 
\ref{pr:construct}) that the observer can  reconstruct the
whole geometry and topology of the spacetime in finite eigentime up to the above-mentioned problem with  finite covers. 

\subsection{Computations for examples}

In Section \ref{sc:2+1} we present explicit computations for the frequency function seen by an observer  for  three different globally hyperbolic flat 2+1-dimensional spacetimes. These spacetimes 
are chosen for their simplicity. Two are obtained by grafting a hyperbolic surface along
a rational measured lamination, the third by grafting along an irrational lamination. %The universal covers of second and third examples are shown in Figure \ref{fig:dd}. 
For each of those spacetimes, we provide pictures of the frequency function as seen by an observer
located at different points in the spacetime. This allows one to observe explicitly the variation of the frequency function depending on the cosmological time.

In dimension $3+1$, we only consider one example, described in Section \ref{sc:3+1}. This is due to the fact that 3+1-dimensional
globally hyperbolic flat spacetimes are much more difficult to construct than their $2+1$-dimensional counterparts. In both cases, they are associated to first-order deformations of the flat conformal
structure underlying a hyperbolic manifold. However,  hyperbolic manifolds are flexible in
dimension $2$, while they are rigid in dimension $3$. Consequently, it becomes more difficult  to find an adequate deformation
cocycle in dimension $3+1$. The example considered  in
Section \ref{ssc:apanasovex} is due to Apanasov \cite{apanasov:deformations}, and it has
the relatively rare property of admitting several distinct deformation cocycles.
We provide some pictures of the frequency function measured by an observer in a spacetime constructed
from this example.

\subsection{Possible extensions}

In this article, we consider only flat globally hyperbolic spacetimes. However,  it should be possible to perform a similar analysis for globally hyperbolic  de Sitter or anti-de Sitter spacetimes, 
which have a similar structure, at least with respect to the geometry of their initial singularity.

%:bms2.tex

\section{Globally hyperbolic Minkowski space-times}
\label{sc:background}

\subsection{Minkowski space and domain of dependences}
\label{ssc:minkowski}

In this section, we collect some  properties of Minkowski space and refer the reader to \cite{mess, bonsante}
for details.
Minkowski space in $n+1$ dimensions, denoted  $\mathbb R^{1,n}$ in the following, 
is the manifold $\mathbb R^{n+1}$ equipped with the flat Lorentzian form
   $ \eta=-dx_0^2+dx_1^2+\ldots+dx_n^2$,
often referred to as  Minkowski metric. 

\subsubsection*{Isometry group}

Isometries of Minkowski space are affine transformations of
  $\mathbb R^{n+1}$ whose linear part preserves the Minkowski
  metric. We denote by $O(1,n)$ the group of linear transformations
  of $\mathbb R^{n+1}$ which  preserve the Minkowski metric (Lorentz
  group in $n+1$ dimensions) and by $\Iso(n,1)$ the group of isometries of Minkowski space
  (Poincar\'e group in $n+1$ dimensions).  $O(n,1)$ is a ${n(n+1)}/{2}$-dimensional
  Lie group with four connected components, and  we denote by $SO^+(n,1)$
  its identity component, which  contains linear orthochronous
  transformations with positive determinant.  The
  dimension of $\Iso(n,1)$ is
  ${n(n+1)}/{2}+(n+1)={(n+1)(n+2)}/{2}$, and for $n\geq 3$ this group has four connected
  components. The identity component, denoted  $\Iso_0(1,n)$,
  contains the transformations that preserve both the orientation and the
  time orientation.

\subsubsection*{Flat spacetimes}

It is well-known that every flat spacetime is locally modeled
  on Minkowski space.  For globally hyperbolic flat spacetimes, a more
  precise result holds (see \cite{mess,mess-notes}).  For every flat spacetime  $M$ with a closed Cauchy surface,  there is a discrete group
  of isometries $\Gamma\subset \Iso_0(n,1)$ and a convex domain
  $ D\subset\mathbb R^{1,n}$ such that $ D$ is
  $\Gamma$-invariant and $M$ embeds into the quotient $ D/\Gamma$.
The domain $ D$ is a domain of dependence, in the sense that it
is the intersection of the futures of one or more lightlike planes.  Domains of
dependence play an essential role in  this paper, and will be described in more detail  
 below. The quotient space
$ D/\Gamma$ is called a maximal globally hyperbolic flat
space-time with compact Cauchy surfaces, for which we use the acronym
MGHFC.

\subsubsection*{Hyperbolic representations}

The unit timelike vectors in $\mathbb R^{1,n}$ form a smooth
  hypersurface, $H\subset \mathbb R^{1,n}$, which contains two connected
  components: the component  $H^+$ that contains future oriented unit vectors, and $H^-$
  that contains past oriented unit vectors.  Both $H^+$ and $H^-$ are
  achronal spacelike smooth surfaces. The Minkowski metric induces a Riemannian metric of constant 
  curvature $-1$ on $H^+$ and $H^-$. Equipped with this metric $H^+$ and $H^-$  are isometric to the  $n$-dimensional hyperbolic space $\mathbb H^n$.  
The group $SO^+(1,n)$ acts by isometries on $H^+$, and it is identified
with the identity component of the isometry group of $H^+$.
Every geodesic of $H^+$ is given as the
intersection of $H^+$ with a timelike linear $2$-dimensional plane.

\subsection{Domains of dependence}

A domain of dependence (called regular domain in \cite{bonsante}) $ D$ is
a convex domain of $\mathbb R^{1,n}$ that is given as the
intersection of the future (or the past) of a number of lightlike $n$-planes. We will
exclude two limit cases: the whole space and the future of a single
lightlike $n$-plane.  In other words,  we  require that $\mathbb
R^{1,n}\setminus D$ contains at least two non-parallel lightlike
$n$-planes.

Simple examples of domains of dependence are the future of a point, or
the future of a spacelike $(n-1)$-plane, whereas   the future of a
spacelike $n$-plane is not a domain of dependence.  Examples with interesting geometrical properties are
the universal covers of maximal globally hyperbolic flat manifolds with compact Cauchy
surface (MGHFC manifolds). Figure \ref{fig:dd} shows two
of those more complex examples, corresponding to the domains of dependence
described in sections \ref{sssc:torus1} and \ref{sssc:torus4}.

\begin{figure}[ht]
  \begin{center}
  \includegraphics[width=8cm]{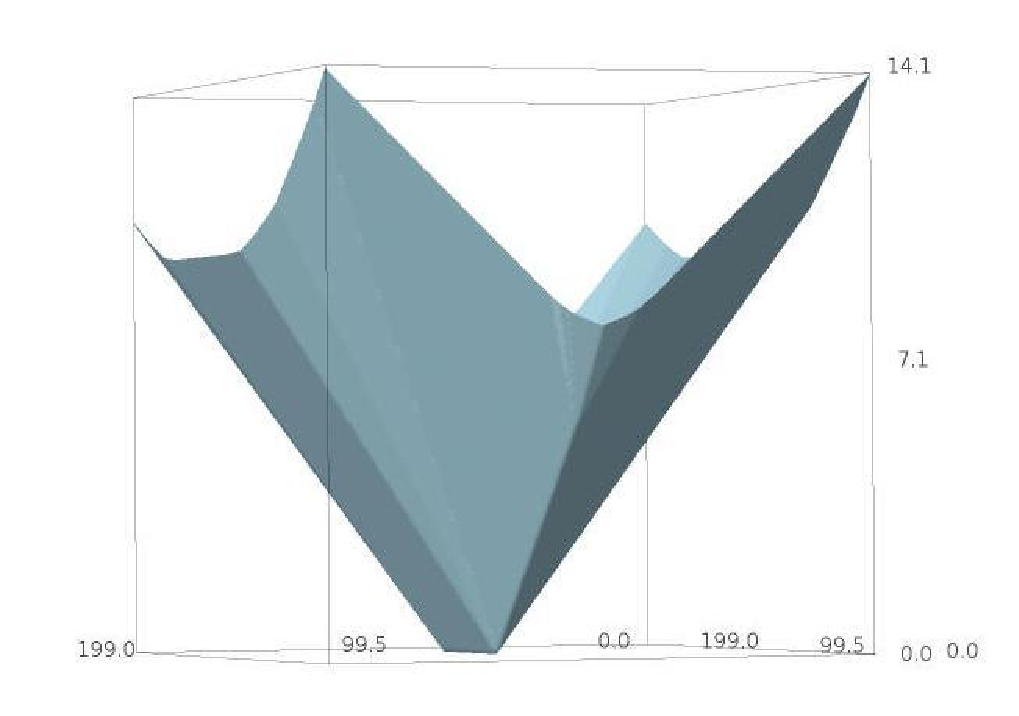}
  \includegraphics[width=8cm]{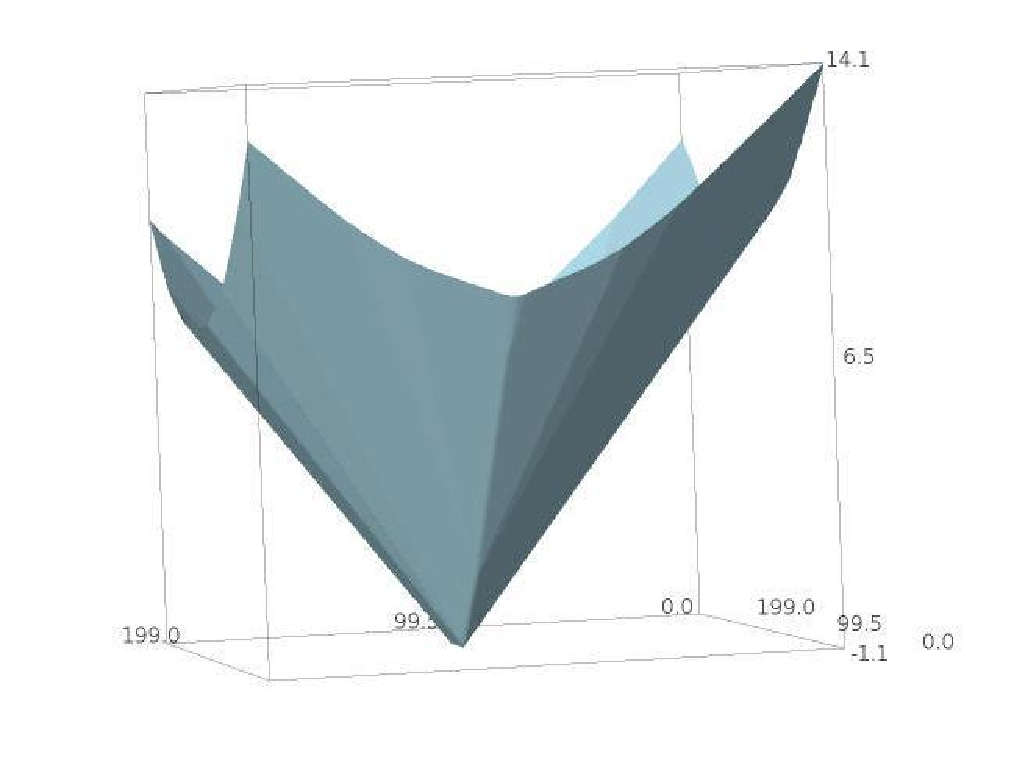}
  \caption{Two examples of domains of dependence}
  \label{fig:dd}
  \end{center}
\end{figure}

Let us recall that for any Lorentzian manifold $M$, the cosmological
time is a function $\tau:M\rightarrow(0,+\infty]$ whose value  at a point
$p\in M$ is the supremum of the length of causal curves in $M$
ending at $p$:
\[
   \tau(p)=\sup\{\ell(c)|c\textrm{ is a casual curve ending at }p\}\,.
\]
One of the main features of domains of dependence is that their
cosmological time is a regular function. This means that $\tau$ is
finite-valued and $\mathrm C^{1,1}$.  In fact, if $ D$ is a domain
of dependence and $p\in  D$, there is a unique point $r=\ret(p)\in\partial
 D\cap\past(p)$ such that $\tau(p)=|p-r|$.  Level surfaces
$H_a=\tau^{-1}(a)$ of the cosmological time are spacelike Cauchy surfaces, and their normal vector  at a point  $p\in  H_a$ is the vector $p-\ret(p)$.

\begin{example}\label{domainex} $\quad$

\begin{itemize}
\item
If $ D$ is the future of a point $r_0$, then $\ret(p)=r_0$ for all $p\in D$, and
the cosmological time $\tau(p)$ coincides with the distance of $p$  from $r_0$.
In this case, the cosmological time is smooth (real analytic in fact), 
and the induced metric on the level surface  $ H_a$ has constant curvature $-1/a^2$.

\item
If $ D$ is the future of a spacelike affine plane $l_0$ of dimension $k\leq n-1$
then $ D$ is a domain of dependence. For $p\in D$, 
$\ret(p)$ is the intersection point of $l_0$ with the affine subspace
orthogonal to $l_0$  passing through $p$. 
Also in this case $\tau$ is smooth. The level surface
$H_a$ are isometric to $\mathbb R^k\times\mathbb H^{n-k}$.
If $n=2$ and $k=1$, this implies that the metric is flat.

\item
If $ D\subset\mathbb R^{n,1}$ is the future of a spacelike segment
$[p_0,p_1]$, then $ D$ is divided into three regions by two
timelike hyperplanes $P_0,P_1$ orthogonal to $[p_0,p_1]$ and passing, respectively, 
through $p_0$ and $p_1$. The first region is the half-space $ D_0$ bounded by
$P_0$ which does not contain $p_1$, the second is the half-space $D_1$ bounded by $P_1$ which does not 
contain $p_0$, and the third is the intersection $V$ of the other two half-spaces bounded
by $P_0$ and by $P_1$. 

For $p\in D_0$, one has
$\ret(p)=p_0$, for $p\in D_1$ $\ret(p)=p_1$ and for
$p\in V$, $\ret(p)$ is the intersection point of $[p_0,p_1]$
with the plane orthogonal to $[p_0,p_1]$ that passes through $p$.
In this case $\tau$ is smooth outside the boundaries of the regions $D_0,D_1$ and $V$
and is only $C^{1,1}$ on their boundaries. Level surfaces are divided
in three regions:  the regions $H_0(a)=H_a\cap D_0$ and
$H_1(a)=H_a\cap D_1$ are isometric to half-spaces of
constant curvature $-1/a^2$, while
$B_a=H_a\cap V$ is isometric to the product of the hyperbolic
space of dimension $n-1$ with an interval of length equal to $|p_1-p_0|$.
(For $n=2$, this is a flat strip of width  $|p_1-p_0|$). \end{itemize}
These examples are illustrated in Figure \ref{fig:domainex}.
\end{example}

\begin{figure}[ht]
  \begin{center}
  \includegraphics[width=6cm,height=5cm]{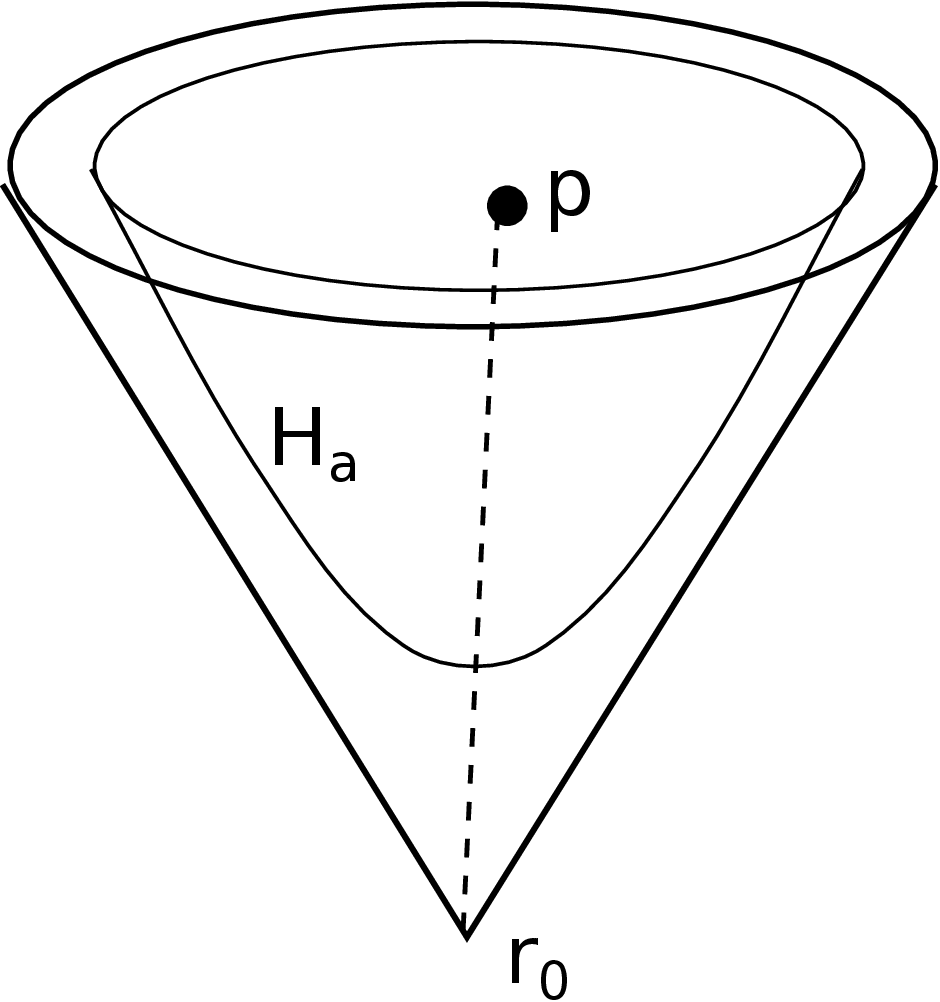}
  \includegraphics[width=8cm,height=5cm]{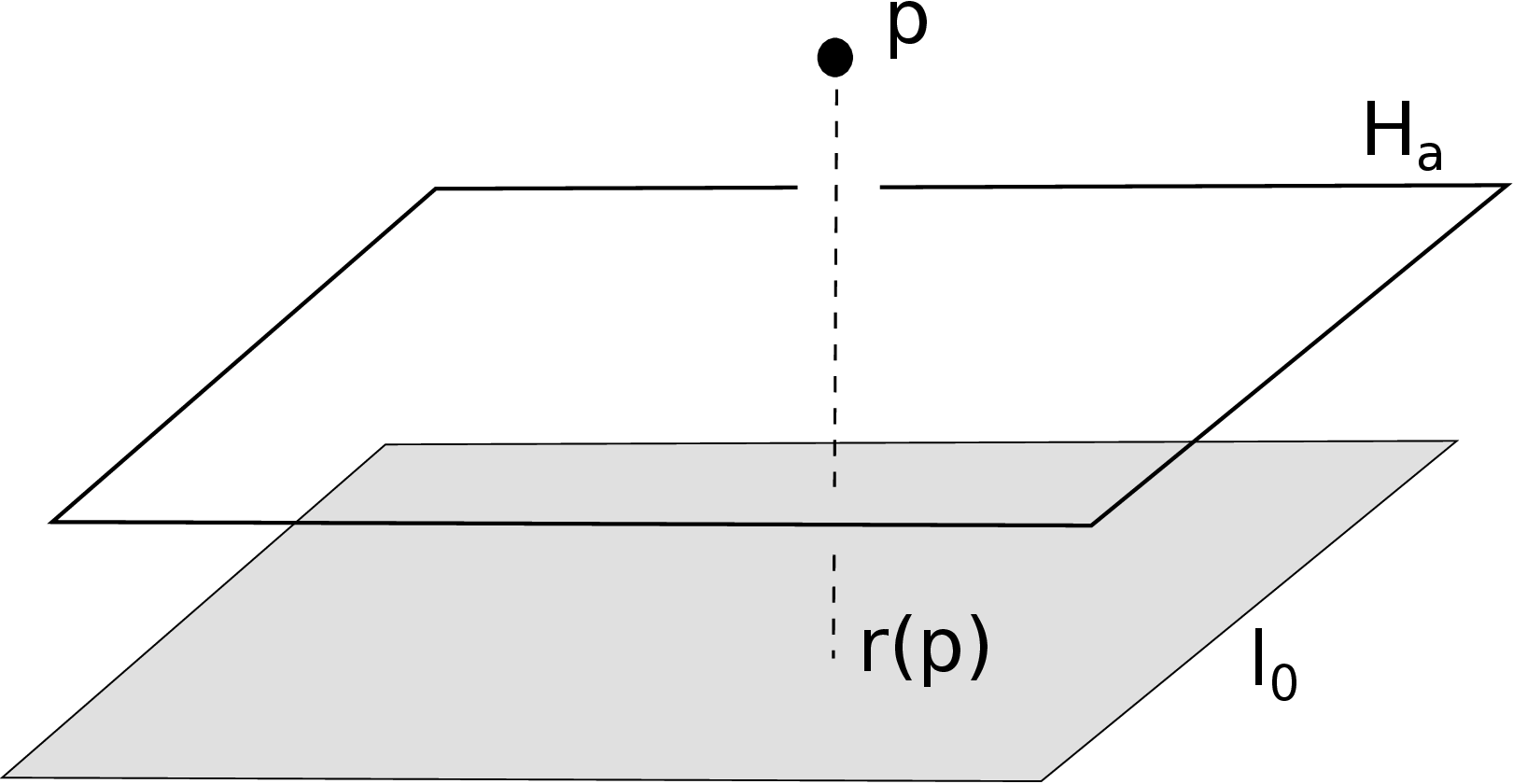}
  \includegraphics[width=8cm,height=6cm]{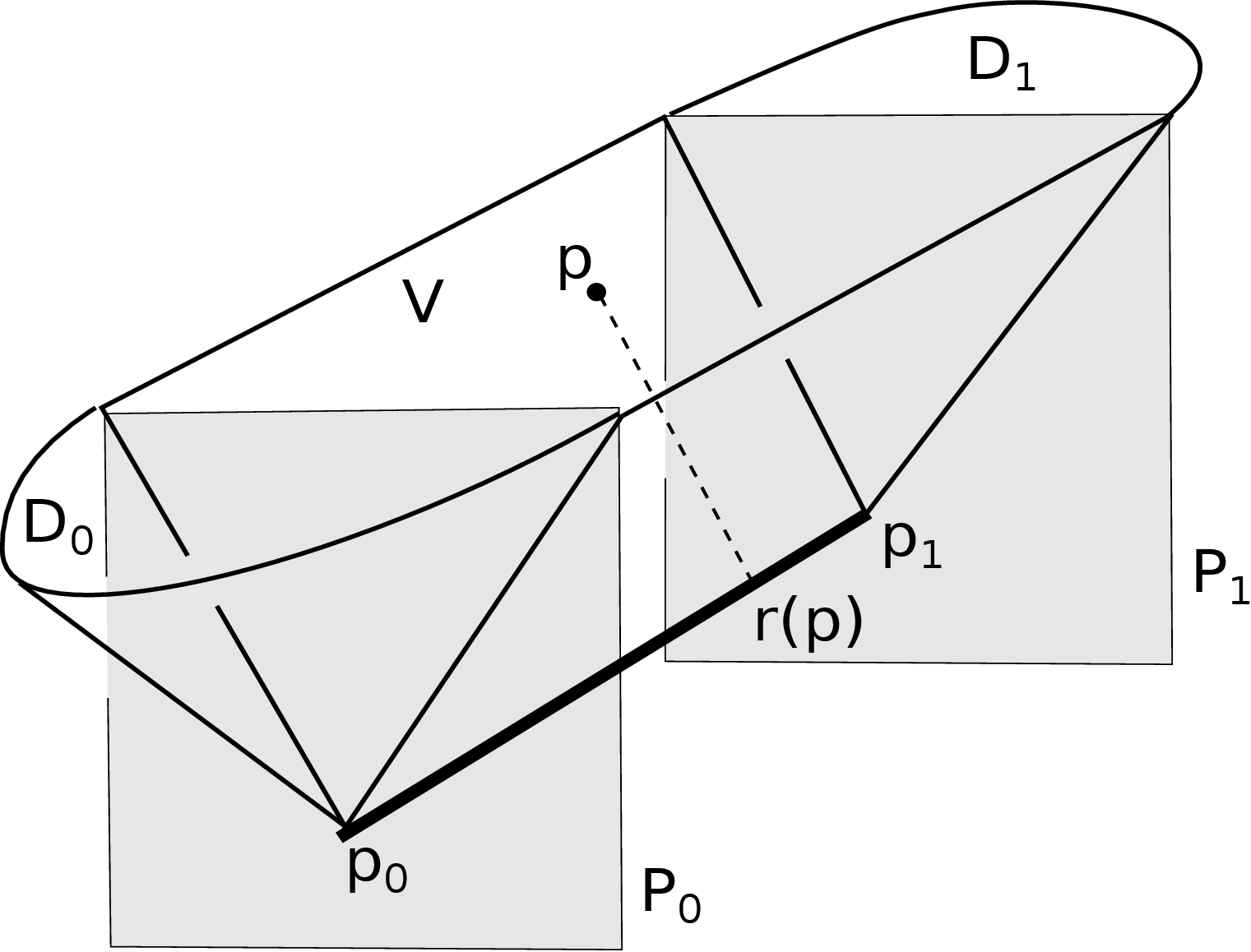}
  \caption{The domains in Example \ref{domainex}}
  \label{fig:domainex}
  \end{center}
\end{figure}

\subsection{The boundary of a domain of dependence and the initial singularity}

 In this
section we recall important facts about the geometry of the boundary 
$\partial D$ of a  domain of dependence $ D\subset \mathbb R^{1,n}$.  

We start by summarizing  a useful description of the boundary $\partial D$.
If $P$ is a spacelike $(n-1)$-dimensional plane in $\RR^{1,n}$, the
orthogonal projection $\pi:\partial D\to P$ is $1$-to-$1$,
so $\partial D$ can be regarded as the graph of a convex function $u$
on $P$. Since $\partial D$ is achronal, $u$ turns out to be $1$-Lipschitz.
More precisely, one finds  that
 the graph of a convex function on $P$ is the boundary
of a domain of dependence if and only if $||\grad u||=1$ at each point
where $u$ is differentiable.

For any point
$r\in\partial D$, there is at least one future directed lightlike
half-line $l$ contained in $\partial D$  that passes through $r$.  It is
important to note that there are always points on $\partial D$
from which at least two lightlike half-lines contained in $D$ originate.
Indeed this occurs exactly when $p$ is in the image of the map 
$\ret: D\rightarrow\partial D$
introduced in the previous section.
This subset is called the \emph{initial singularity}\footnote{Note that the use of the term {\em initial singularity} here differs from the one in the physics literature. The set of points of $\overline D$ at which causal curves cannot be extended into the past  is the entire boundary  $\partial D$.} of the domain $D$, and will be denoted
by $T$. It is the smallest subset of $\overline{ D}$ such that $ D=I^+(T)$.

If one regards $ D$ as the graph of a convex function $u$,  
lightlike lines in $\partial D$ correspond to integral lines of the gradient of $u$,  whereas
the initial singularity corresponds to  the set of points in which $u$ is not differentiable.
 %Even if in the simple cases where the image of $\ret$ is a closed subset of
%$\partial D$, its shape can be very complicated for many
%interesting cases. 
In many interesting cases, the shape of the initial singularity can be quite
 complicated.
For instance, it was shown by Mess \cite{mess} that  for the universal covering $D$ of a generic (2+1)-dimensional MGHFC spacetime, the image of $\ret$ is a dense subset of $\partial D$. 
On the boundary $\partial D$, we consider the pseudo-distance defined as follows:
\begin{itemize}
\item given a Lipschitz arc $k$ contained in $\partial D$, its
  velocity (defined a.e.) is not timelike.  So we can define the
  length of $k$ as
\[
   \ell(k)=\int\sqrt{\langle \dot k(t),\dot k(t)\rangle}\mathrm{d}t\,.
\] 
\item
Given $r_1,r_2\in\partial D$ the space of Lipschitz arcs $\mathcal
K(r_1,r_2)$ joining them is not empty. So we can define
\[
   d_0(r_1,r_2)=\inf\{\ell(k)|k\in\mathcal K(r_1,r_2)\}~.
\]
\end{itemize}
As the boundary of $ D$ contains lightlike segments (whose
length is clearly $0$), the pseudo-distance $d$ is not a genuine distance, since there exist pairs of points with $d(p,q)=0$. 
% On the other hand, the function $d_0$
%is the limit of the distances of the level sets $ D(a)$ when
%$a\rightarrow 0$, in this precise sense: if $p_n,q_n$ is a sequence of
%points in $ D(a_n)$ with $a_n\rightarrow 0$ and if $p_n\rightarrow
%p_\infty$ and $q_n\rightarrow q_\infty$ with
%$p_\infty,q_\infty\in\partial D$ then the distance between $p_n$ and $q_n$
%on $ D(a_n)$ converges to $d_0(p_\infty, q_\infty)$.
However, the following lemma shows that this can occur only if the images of $p$ and $q$ under the  map $r$ coincides.

%In order to get a distance, we have to identify points whose distance
%is $0$: the {\bf initial singularity} of $ D$ is defined as %%$\Sigma= D/\sim$
%$\Sigma=\partial Omega/\sim$
%where $r_1\sim r_2$ if and only if $d(r_1,r_2)=0$. On $\Sigma$, the function
%\[
%\bar d_0([r_1], [r_2])=d_0(r_1,r_2)
%\]
%is well-defined and is in fact a distance.
%Roughly speaking $\Sigma$ is obtained by collapsing lightlike 
%segments on $\partial D$ to a point.

%The main problem of this definition is that the initial singularity
%%is not a concrete object, for instance it 
%cannot be identified with a subset of Minkowski space. 
%However for the domains of dependence 
%we are interested in, we have the following result.
%{\red Add ref -- to Mess or Francesco's paper?}

\begin{lemma}[\cite{bonsante}] \label{initsing:lm}
If $ D$ is a domain of dependence and $p,q\in D$ 
%the universal covering of a MGHFC spacetime,
then 
%\begin{enumerate}
%\item 
 $\ret(p)\neq\ret(q)$ implies $d_0(\ret(p),\ret(q))\neq 0$.
%\item if 
%\item the end-point of any maximal lightlike ray contained in $\partial D$ 
%lies on the image of $\ret$;
%\item if $r$ is any point of $\partial D$, there is a unique point $s$ 
%in the image of $\ret$ such that $r\sim s$.
%\end{enumerate}
\end{lemma}

%This lemma states that the initial singularity of $ D$ can be
%identified with the image of the retraction.  In particular in those
%cases, we can identify $\Sigma$ with a subset of $\partial D$. 
%begin new
In other words, this lemma states that the restriction of $d_0$ to the initial 
singularity is a distance. 
%end new 
Let
us  also remark that the topology we consider on $T$ is the one induced by
the distance $d_0$, which in general is different from the topology
induced by Minkowski space.

%\begin{remark}
%There are examples of domains of dependence, which are not universal covers
%of a MGHFC spacetime, where the second conclusion of Lemma
%\ref{initsing:lm} fails. In these domains there are lightlike rays whose
%starting point in $\partial D$ is not contained in the image of
%$\ret$. In these cases the projection $\pi:Im(\ret)\rightarrow\Sigma$
%is an isometric injection but it is not surjective. However the image
%is always dense in $\Sigma$, see \cite{mess}. {\red Is the ref right??} 
% 
%For the sake of simplicity we will exclude such domains from our
%investigation.  So from now on we will implicitly  consider only
%domains of dependence that satisfy the conclusions of Lemma \ref{initsing:lm}.
%\end{remark}

\begin{example} $\quad$
\begin{itemize}
\item 
If $ D$ is the future of a point $r_0$, the initial singularity contains only one point
that can be identified with $r_0$.

\item If $ D$ is the future of a affine subspace $E$ of dimension $k\leq (n-1)$, then
the initial singularity is isometric to $E$.

\item If $ D$ is the future of a segment in $\mathbb R^{2,1}$, then the initial singularity is the segment
itself.

\item If $ D\subset\mathbb R^{2,1}$ is the intersection of three
half-spaces bounded by lightlike planes, the initial singularity is
the union of three spacelike rays starting from the intersection point
of the planes.  Note that in this case, the initial singularity is
not a submanifold. In fact, generically the initial singularity does
not have a manifold structure.  
The geometry of the initial singularities of domains of dependence in
dimension $2+1$ is discussed in more depth in the next section. 
\end{itemize}
\end{example}

\subsection{The $2+1$-dimensional case: the Mess construction}
\label{ssc:construction}

Mess \cite{mess} discovered an efficient way to construct regular domains in
$\mathbb R^{2,1}$.  This construction is general in the sense that
every regular domain can be constructed in this way. It also has the major advantage that
 the geometrical features of the initial singularities are readily apparent.

We will describe the Mess construction in simple cases, namely for domains obtained by grafting along weighted multicurves.
These simple cases are dense, in the sense that every domain of
dependence can be approximated by domains of dependence obtained
in this way.

Let us start from a collection of disjoint geodesics of $H^+$,
$L=l_1\cup\ldots\cup l_k$, and a collection of positive numbers $a_1,....a_k$. 
Every geodesic is given as  the intersection of  $H^+$ with a 
timelike linear plane $P_1,\ldots, P_k$.
The planes $P_i$ disconnect $H^+$ into a collection
of regions $ D_1,\ldots, D_h$. (Note that each of them is the cone on some
component of $H^+\setminus L$).

For each plane $P_i$, 
let $v_i$ be the vector in $\mathbb R^{1,2}$ characterized by 
the following conditions:
\begin{itemize}
\item it is orthogonal to $P_i$ with respect to the Minkowski metric (in particular it is spacelike);
\item its norm is equal to $a_i$;
\item it points to the component of $\mathbb R^{1,2}$ that does not contain
$ D_1$.
\end{itemize}
Now for any region $ D_j$  take the sum of all vectors $v_i$
associated with  to planes  $P_i$ that separate $ D_j$ from $ D_1$:
\[
    w_j=\sum_{i:P_i\textrm{ separates }  D_1\textrm{ from } D_j}v_i\,.
\]
Translating  each region $ D_j$ by the vector $w_j$ yields  a collection of
disjoint domains $ D_1',\ldots, D_h'$ which are convex cones
with vertices at  $w_j$.  In particular,  note that if $ D_{j_1}$
is adjacent to $ D_{j_2}$,  then $w_{j_1}-w_{j_2}$ is a vector
orthogonal to the plane $P_i$ separating $ D_{j_1}$ from
$ D_{j_2}$ and is of norm  $a_i$.

In order to connect the domains $ D_j'$ we consider the domains
$V_i$ obtained as follows.  If $ D_{j_1}$ and $ D_{j_2}$
are adjacent along $P_i$, then $V_i$ is the region of the future of
the segment $s_i=[w_{j_1}, w_{j_2}]$ bounded by the two timelike
planes orthogonal to the segment $s_i$ through its end-points.

It turns out that $ D=\bigcup D_j'\cup\bigcup V_i$ is a
domain of dependence. The map $\ret$ can be easily defined on each
piece: $\ret$ sends all points of $ D_j'$ onto $w_j$,
while it sends points of $V_i$ to the segment $s_i=[w_{j_1}, w_{j_2}]$.
%%as explained in the example ???. 
The level surface $H_a$ can be
decomposed into different regions: the regions $H_a\cap D'_i$, which have
constant curvature $-1/a^2$, and the regions   $H_a\cap V_i$, which  are 
Euclidean strips of width  $a_i$.

The initial singularity is then given as  the union of the line segments $s_i$ and the 
vertices $w_j$.  In particular,  it is a graph with a vertex for every
region of $H^+\setminus L$. Two vertices $w_{j_1}$ and $w_{j_2}$ are
connected by one edge if and only if the corresponding regions are
adjacent.  Combinatorially the singularity is a tree, that is, a graph
which does not contain any closed loop.  Notice that the length of
each segment $s_i$ is precisely $a_i$.

Although we summarized this  construction  for a finite number of geodesics,
it works analogously  also when $L$ is an infinite family of disjoint geodesics
that is locally finite (i.e. every compact subset of $H^+$ meet only a finite number of $l_i$).

\subsection{The equivariant construction} \label{ssc:equivariant}

Using the construction from the  previous subsection,  one can construct the universal coverings 
 of MGHFC spacetimes
different from $\fut(0)$ as follows.  Take a hyperbolic
surface $F$ and consider the metric universal covering
$\pi:H^+\rightarrow F$ and covering group $\Gamma<SO^+(1,2)$.
Consider on $F$ a disjoint collection of simple closed geodesics
$c_1\ldots c_k$ and positive numbers $a_1\ldots a_k$.  Then the
preimage $L=\pi^{-1}(c_1\cup\ldots\cup c_k)$  is a union of
infinitely many disjoint geodesics.  The weight of each  geodesic $\tilde l_i\subset L$ 
is the number corresponding to $\pi(\tilde l_i)$.  As
above, the geodesics $\tilde l_i$ correspond to planes that cut $\fut(0)$
into infinitely many pieces $ D_j$.  By the invariance of $L$
under the action of $\Gamma$, elements of $\Gamma$ permute the regions
$ D_j$.

The construction explained in the previous subsection then produces a domain
$ D$, and Mess showed that there is an affine deformation $\Gamma'$
of $\Gamma$, so that $ D$ is $\Gamma'$-invariant and the quotient
is a MGHFC spacetime.  Namely any $\gamma\in\Gamma$ is changed by adding a
translation part of vector $w(\gamma)$ which is the sum of all vectors $w_i$
corresponding to the planes $P_i$ disconnecting $ D_1$ from
$\gamma D_1$.

\begin{remark}
In the example above, it  can be seen that each $ D_j$ bounds infinitely many
planes $P_i$. This implies that the vertex in the initial singularity corresponding to $ D_j$
is the end-point of infinitely many edges, or equivalently has infinite valence.
\end{remark}

In the examples illustrated in the previous section it turns out that
the initial singularity of domains of dependence in $\mathbb R^{2,1}$ is
always a graph, and in fact a tree (possibly with vertices of infinite valence).
In fact, there are more complicated examples in which the initial singularity
does not have a simple graph structure, but it is always a real tree according to the following
definition.

\begin{defi}
A metric space $(T,d)$ is a {\bf real tree} if for every $p,q\in T$
there is a unique arc $k\subset T$ joining them.
Moreover $k$ is the image of an isometric immersion $I\rightarrow T$
where $I$ is an interval of length equal to $d(p,q)$.
\end{defi}

Real trees are generalizations of the usual trees (which, by contrast, are
often called simplicial trees).
The domains of dependence whose singularity is a simplicial tree are exactly 
those constructed in the previous section \cite{benedetti-guadagnini, benedetti-bonsante}.
In particular, every domain of dependence with simplicial tree as initial  singularity
is determined by a simplicial measured geodesic lamination of $H^+$,
which, by definition,  is a
locally finite union $L$ of disjoint geodesics $l_i$, each equipped with  a weight $a_i>0$.

\begin{prop}\cite{benedetti-guadagnini}
If $ D$ is a domain of dependence in $\mathbb R^{2,1}$ then
its initial singularity $T$ is a real tree.
Moreover, the vertices of $T$ are those points in $\partial D$  at which at least three lightlike
segments  in $\partial D$ originate. 
\end{prop}

%{\red Could we give a ref for this, since we don't really provide a proof?}

Given a point $r\in T$, let $ D_r$ be the convex hull in
Minkowski space of the lightlike lines contained in $\partial D$ which start at 
 $r$.  Then $D_r$  is a convex subset of $\overline\fut(r)$.
Notice that the dimension of $ D_r$ is $3$ if and only if $r$ is a vertex,
otherwise $ D_r$ is the intersection between $ D$ and the
timelike plane containing the two lightlike 
rays starting at $r$.

If $\tau_r$ is the translation which send $r$ to $0$, we denote by
$\mathcal F_r$ the intersection of $H^+$ with $\tau_r( D_r)$. 
Note that $\mathcal F_r$ can be interpreted as the set of unit normals of the support
planes of $ D$ at $r$. A number of consequences follow  directly. 
\begin{itemize}
\item If $r$ is a vertex then $\mathcal F_r$ is a region of $H^+$ bounded by disjoint geodesics.
\item If $r$ is not a vertex then $\mathcal F_r$ is a complete geodesic.
\item If $r\neq s$ then $\mathcal F_r$ and $\mathcal F_s$ have disjoint interiors. $\mathcal F_r$ and $\mathcal F_s$ can be disjoint, they can coincide if they are both lines, or they can meet along a 
boundary component.
\end{itemize}

In particular the set $L=\bigcup_{r\textrm{ is a
vertex}}\partial\mathcal F_r\ \cup \bigcup_{r\textrm{ is not a
vertex}}\mathcal F_r$ is a union of disjoint geodesics that are
called the leaves of $L$.
In general,  the set $L$ can be quite complicated. The intersection of a
geodesic arc in $H^+$ and $L$ can be uncountable (and sometimes a
Cantor set).

The simplest case is when the singularity is a tree. In this case, the set 
$L$ is the union of isolated geodesics: any compact subset of $H^+$ meets
only a finite number of leaves.
In this case,  it is also  evident that components of $H^+\setminus L$ corresponds to
vertices of $T$, whereas each leaf of $L$ corresponds to an edge of $T$.

In addition to $L$ we can construct a \emph{transverse measure} that is the 
assignment of a non-negative number for any arc transverse to the leaves of $L$
which verify some additivity conditions, see e.g. \cite{bonahon:laminations}.
If $k$ is an arc on $H^+$that joins two points  in $H^+\setminus L$ and
meets each leaf at most once (for instance if $k$ is a geodesic
segment that is not contained in any leaf), we define
$\mu(k)=d_0(r_0,r_1)$ where $r_0$ and $r_1$ are the points on $T$ such that
the end-points of $k$ are contained in $\mathcal F_{r_0}$ and $\mathcal F_{r_1}$.

If  the lamination is locally finite,  for each leaf $l$ there
is a number $a(l)$ that coincides with the measure of any arc $k$ transversely 
meeting only $l$. Any transverse arc $k$ can be subdivided into a finite number
of arcs $k_1,\ldots,k_p$ such that each $k_i$ meets every leaf at most once. So we can define
$\mu(k)=\sum\mu(k_i)$.

Mess \cite{mess} showed that the data $(L,\mu)$ determines $ D$ up to translation.
In the simple case where the lamination is locally finite, the construction of  $ D$ from 
$(L,\mu)$ is the one summarized in Section \ref{ssc:construction}.

\begin{remark}
In dimension $n+1\geq 4$,  it is no longer true that the initial singularity is
a tree. In fact the geometry of the initial singularity is still not understood.
In \cite{bonsante}  a description of the singularity is given in some special cases.
\end{remark}

\subsection{Holonomies of domains of dependence and hyperbolic structures}
\label{ssc:holonomies}

In Section \ref{ssc:minkowski}, we summarized the construction which assigns a domain of dependence to each flat  MGHFC  manifold. There is also a deep relation between
holonomies of flat Lorentzian manifolds and first-order deformations
of holonomy representations of hyperbolic manifolds, which was already used in dimension $2+1$ for instance in \cite{goldman-margulis}. In the following, we summarize this relation, which behaves somewhat differently in dimension $2+1$ and in higher dimension.

\subsubsection*{Dimension 2+1}

In this subsection, we recall  how a flat  (2+1)-dimensional MGHFC 
 manifold can be obtained from a point in 
Teichm\"uller space together with a deformation 1-cocycle.
For this, note that $\R^{2,1}$ can be identified with the Lie algebra $sl(2,\R)$ 
with its Killing metric. The canonical action of $SO(2,1)$ on $\R^{2,1}$ 
corresponds to the adjoint action of $SL(2,\R)$ %(identified with $SO(2,1)$) 
on $sl(2,\R)$. 
For each  representation of $\pi_1S$ in $SO(2,1)$, it determines a vector bundle over
$S$ with fiber $\R^{2,1}$, which corresponds to the $sl(2,\R)$-bundle over $S$ defined
by the adjoint representation. 

\begin{prop}\cite{mess}
Let $M$ be a flat MGHFC manifold homeomorphic to $S\times \R$, where $S$ is a 
closed surface of genus at least $2$, and let $h:\pi_1S\rightarrow \Iso(2,1)$
be its holonomy representation. Then $h$ decomposes in $\Iso(2,1)=SO(2,1)\ltimes \R^{2,1}$
as $h=(\rho,\tau)$ where $\rho:\pi_1S\rightarrow SO(2,1)$ has maximal Euler number, and
$\tau:\pi_1S\rightarrow sl(2,\R)$ is a 1-cocycle for $\rho$. Conversely, any
couple $(\rho,\tau)$ where  $\rho:\pi_1S\rightarrow SO(2,1)$ has maximal Euler number and
$\tau:\pi_1S\rightarrow sl(2,\R)$ is a 1-cocycle for $\rho$ defines a representation
of $\pi_1S$ in $\Iso(2,1)$ which is the holonomy representation of a flat MGHFC 
manifold.
\end{prop}

One way to obtain a 1-cocycle is by considering  first-order deformations of a surface group
representation in $SO(2,1)$. This is summarized in the following proposition, which allows one to construct the holonomy representation of a flat 
MGHFC is as a first-order deformation of the holonomy representation of a 
hyperbolic metric on a surface.

\begin{prop}
Let $(\rho_t)_{t\in [0,1]}$ be a smooth one-parameter family of morphisms from 
$\pi_1(S)$ to $PSL(2,\R)$. Then the map $\tau=\rho_0^{-1}(d\rho/dt)_{t=0}$ from 
$\pi_1(S)$ to $sl(2,\R)$ is a 1-cocycle for $\rho_0$.   
\end{prop}

\subsubsection*{MGHFC spacetimes as first-order deformations in higher dimension}

We will now consider the construction of flat MGHFC spacetimes as first-order deformations in dimension $n+1>3$. For this, we consider
 a closed, orientable, hyperbolic $n$-dimensional manifold $M$, with 
fundamental group $\Gamma$. The holonomy representation of $M$ is a homomorphism
$\rho_0:\Gamma\rightarrow SO(n,1)$. It is rigid by Mostow's theorem, and also 
infinitesimally rigid, in the sense that any deformation cocycle for $\rho_0$
vanishes. 

However, $M$ can be considered as a totally geodesic hypersurface in a complete,
non-compact hyperbolic manifold $N$ of dimension $n+1$. This corresponds to extending
$\rho_0$ to a representation $\rho:\Gamma\rightarrow SO_0(n+1,1)$ with image in
$SO_0(n,1)\subset SO_0(n+1,1)$. 
Now consider a deformation $(\rho_t)_{t\in [0,1]}$ of $\rho$. As in dimension $2+1$
one obtains the map
$$ \rho_1:=\rho(0)^{-1}\rho'(0):\Gamma\rightarrow o(n+1,1)~, $$
which is a deformation cocycle for $\rho$. 
Moreover, there is an orthogonal decomposition $o(n+1,1)=o(n,1)\oplus \R^{n,1}$, and
we can decompose $\rho_1$ along this direct sum. The component in $o(n,1)$ is a 
deformation cocycle for $\rho_0$, so it vanishes by the infinitesimal rigidity of 
$\rho_0$, and thus $\rho_1$ determines a $\R^{n,1}$-valued cocycle.

This cocycle then determines a MGHFC spacetime, with holonomy representation 
$(\rho_0,\rho_1)$ considered as a homomorphism from $\Gamma$ to $\Iso(\R^{n,1})=
SO_0(n,1)\ltimes \R^{n,1}$. Moreover, the  holonomy 
representations of all MGHFC spacetimes can be obtained in this way, see \cite{mess, bonsante}. 

\begin{prop} \label{pr:defo-higher}
The MGHFC spacetimes for which the linear part of the holonomy is equal to $\rho_0$
are in one-to-one correspondence with the deformations of $\rho$.
\end{prop}

There is another geometrical  interpretation of the deformation cocycle $\rho$, namely as a 
first-order deformation of the flat conformal structure on $M$ underlying its hyperbolic
metric. Indeed, in the situation described above, where $M$ is considered as a totally geodesic
submanifold of $N$, the conformal structure at infinity of $N$ remains conformally flat
under a deformation. Conversely, any first-order deformation of the conformally flat structure
on $M$ determines a deformation of its developing map in $S^n$ and, by taking the 
convex hull of the complement of its image, one obtains a first-order deformation of $N$.

%{\red I think that we should add a ref here.}

One direct consequence Proposition \ref{pr:defo-higher} is that it is much more difficult to
construct examples of MGHFC  spacetimes in higher dimensions than in dimension $2+1$. An (n+1)-dimensional MGHFC spacetime is uniquely determined by a closed $n$-dimensional hyperbolic manifold
along with a $\R^{n,1}$-valued deformation cocycle. 
For $n=2$, the latter corresponds to  a tangent vector to the Teichm\"uller 
space for a surface $S$ of given genus $g$, which implies that the MGHFC spacetimes homeomorphic to $S\times \R$ 
form a manifold of dimension $12g-12$. 
For $n\geq 3$, finding a MGHFC spacetime homeomorphic to $M\times \R$ is more difficult. 
Any closed manifold $M$ admits at most one hyperbolic metric $g$ by Mostow's rigidity theorem. 
Finding a deformation cocycle is equivalent to finding a first-order deformation
of the warped product hyperbolic metric $dt^2+\cosh^2(t)g$ on $\R\times M$. 
For many choices of $(M,g)$, such a deformation cocycle does not exist. However there are also
many examples where $(M,g)$ does admit a $\R^{n,1}$-valued cocycle. 
\begin{itemize}
\item This occurs whenever $(M,g)$ contains a closed, totally geodesic submanifold, and the
cocycle corresponds to ``bending'' along this totally geodesic manifold, see \cite{kourouniotis,johnson-millson}.
There are many (arithmetic) examples of closed hyperbolic manifolds (in any dimension) containing a closed 
totally geodesic surface.
\item Other examples of deformation cocycles can be found in specific cases, see e.g. 
\cite{kapovich:deformations,apanasov:deformations,scannell:deformations}.
\end{itemize}

In Section \ref{sc:3+1} we investigate the examples constructed by Apanasov in \cite{apanasov:deformations} and
 show how the frequency function encodes  information on the holonomy
representation and hence on the topology of the spacetime.
 
\subsection{Reconstructing a domain of dependence from its holonomy representation}
\label{ssc:reconstructing_holo}

In Section \ref{sc:2+1} and \ref{sc:3+1} we  compute  domains of
dependence which are universal covers of MGHFC spacetimes. This requires a practical way of reconstructing (to a good approximation) the
shape of a domain  from the holonomy representation of the   MGHFC spacetimes. For this, we use another characterization of those domains due to Barbot \cite{barbot:globally}.

We consider a MGHFC spacetime $M$ of dimension $n+1$ with fundamental group $\Gamma$. 
As explained in Section \ref{ssc:minkowski}, 
the universal cover of $M$ can be identified
isometrically with a future-complete domain of dependence $ D\subset \R^{1,n}$. The fundamental group   $\Gamma$ acts isometrically on $D$ with a quotient  $D/\Gamma$ isometric to $M$. Moreover all elements of $\Gamma$
except the unit element  act on $\R^{1,n}$ as loxodromic elements. 

\begin{defi}
For  $g\in \Gamma$, we denote by $ D_g$ the set of points $x\in \R^{1,n}$ such that $g^p(x)-x$ is spacelike for all $p\in\mathbb Z$. 
\end{defi}

This is a simpler version of the definition at the beginning of Section 7 in  
\cite{barbot:globally}, but both definitions are equivalent in our case because 
the linear part of each nontrivial element $g\in \Gamma$ is loxodromic.
In (2+1) dimensions, it is easy to give a more explicit description of the set $D_g$. If the linear part of $g$ is loxodromic, then there is a unique
spacelike line $l_g$ in $\RR^{2,1}$ which is invariant under the action of $g$.
It is proved in  \cite{barbot:globally} that the set $ D_g$ is then equal to the union of the
past and the future of $l_g$.

\begin{prop}[Barbot \cite{barbot:globally}] \label{pr:barbot}
The domain of dependence $ D$ is one of the two connected components of $\cap_{g\in \Gamma} D_g$.
\end{prop}

A proof can be found --- in a more general setting --- in Barbot's work \cite{barbot:globally}, see
Section 7 for the definitions and Section 10 for the statements corresponding to Proposition
\ref{pr:barbot}.

%:bms3-bis.tex

\section{Light emitted by the initial singularity}
\label{sc:light}

\subsection{Definitions} \label{ssc:defi}

We consider a domain of dependence $M$  in  (n+1)-dimensional Minkowski space, 
as described in the previous section. 

An observer in free fall in $M$ is characterized by his worldline, which is  a future-oriented timelike geodesic in $M$. This geodesic is specified 
by the choice of a point $p\in M$ and a future-directed timelike unit vector $v\in \mathbb H^n$, 
where we use the identification of $\mathbb H^n$ with the set of future directed timelike unit vectors $H^+$ from section \ref{sc:background}. The point $p\in M$ corresponds to a given event on the worldline of the observer, while the vector $v$
is  his velocity unit vector.  

We consider a uniform light signal emitted near the initial singularity of $M$  which is received by the observer at the point $p\in M$. 
The quantity  measured by the observer is the frequency of this 
light signal, which 
 depends  on the observer and on the direction in which the light
is observed. We can construct this quantity as follows. 
The space of lightlike rays arriving at $p$ can be identified
with the set of unit spacelike vectors orthogonal to $v$ and hence with
 $T^1_{v}\mathbb H^n$, which corresponds to the (n-1)-dimensional sphere $S^{n-1}$.  
We associate  to each unit  vector $u\in T^1_{v}\hyp^n$ the ray through $p$ with the direction given by the lightlike vector $u-v$.
The basic idea is to  define the (rescaled) frequency function as a function
\[
    \rho_{p,v}:T^1_{v}\mathbb H^n\rightarrow\mathbb R\,,
\]
which is given as the renormalized  limit of the functions that
measure the frequency of the light emitted from the surface 
surface $H_\epsilon$  of constant cosmological time $\epsilon$:
\[
   \rho_{p,v}(u)=\lim_{\epsilon\rightarrow 0}\epsilon\rho_{p,v}^\epsilon(u).
\]
The functions $\rho_{p,v}^\epsilon$  are defined by the rule 
\[
   \rho_{p, v}^\epsilon(u)=\langle v, \nu^\epsilon_{p,v}(u)\rangle
\]
where $\nu^\epsilon_{p,v}(u)$ is the normal of the  surface $
H_\epsilon$ at the intersection point of $H_\epsilon$ with the light ray $p+\mathbb R(u-v)$:
$\{q_\epsilon(u)\}=H_\epsilon\cap(p+\mathbb R (-v+u))$.
It is clear that the frequency function describes a frequency shift due to the motion of the observer relative to the initial singularity and is closely related to  the red-shift observed in astronomical observations. 

We consider first the $(2+1)$-dimensional case. In this situation, the frequency function
 of a domain that is the future of a finite spacelike tree
can be understood by considering two main examples. The first is a domain
that is the future of a point, i.~e.~ a light cone,  and the second is a domain which is
the future of a spacelike line. We first consider these two examples and then use them as the building blocks to analyze the general situation. 

\subsection{Example 1: future of a point} \label{ssc:ex1}

We consider the domain of dependence $D$ which is the future of $0\in\mathbb
R^{2,1}$ together with an observer in $D$ which is given by a point  $p\in D$ and a future directed timelike unit vector  $v\in
\hyp^2$ as shown in Figure  \ref{fig:obs1}.  The cosmological time $\tau$ of the event $p$ is then determined by
\begin{align}
\skp p p =-\tau^2.
\end{align}
We also consider the  quantity  $\delta$ defined
by 
\begin{align}\label{deltadef}
 \langle p, v\rangle=-\tau\ch\delta,
\end{align}
which is the hyperbolic distance between $v$ and the point $p/\tau$ (see Figure  \ref{fig:obs1}). 
and measures the discrepancy of the observer's eigentime
and the cosmological time. The cosmological time coincides with the
observer's eigentime up to a time translation if and only if  $\delta=0$. We can also
interpret $\delta$ as the rapidity of the boost from the worldline of the observer to the
 the geodesic through $p$ and the origin.

For a given unit vector $u\in T^1_{v}\mathbb H^2$ we also introduce a  parameter $\phi$
defined by
\begin{align}\label{phidef}
  \langle p, u\rangle=\tau \sh\phi\,,\qquad \phi\in[-\delta,\delta]~.
\end{align}
Geometrically, $\phi$ is the hyperbolic distance of the point $p/\tau$
from the geodesic orthogonal to $u$. It becomes maximal 
 when $u\in T^1_v\hyp^2$ points in the direction of
$p/\tau\in\hyp^2$ and minimal when $u$ points away from
it.

We denote by $t_\epsilon\in\mathbb R^+$ the parameter that characterizes the intersection point $q_\epsilon(u)$ of the light ray $p+\mathbb R(u-v)$ with the surface $H_\epsilon$ of constant cosmological time $\epsilon$. As the latter is the set of points $$H_\epsilon=\{x\in\mathbb R^{2,1}|\; \langle x,x\rangle=-\epsilon^2\},$$
the parameter $t_\epsilon$ is characterized uniquely as the positive solution of the equation
\[
  \langle p+t_\epsilon(-v+u),p+t_\epsilon(-v+u)\rangle=-\epsilon^2.
\]
Inserting the parameters $\delta$ and $\phi$ defined in \eqref{deltadef} and \eqref{phidef} and solving the equation for $t_\epsilon$, we obtain
\[
  t_\epsilon(u)=\frac{\tau^2-\epsilon^2}{2\tau(\cosh\delta+\sinh\phi(u))}~.
\]
The  unit normal vector $\nu^\epsilon_{p, v}$ in the intersection point of $p+\mathbb R(u-v)$ and $H_\epsilon$ is given by
\[
   \nu^\epsilon_{p,v}(u)=\frac{1}{\epsilon}(p+ t_\epsilon(-v+u)),
\]
 and the function $\rho_{p,v}^\epsilon$ by
\[
   \rho_{p,v}^\epsilon(u)=-\left\langle v, \frac{p+t_\epsilon(-v+u)}{\epsilon}
   \right\rangle\,.
\]
A direct computation then shows that the rescaled frequency function then takes the form 
\[
  \rho_{p,v}(u)=\frac{\tau}{2(\ch\delta+\sh\phi(u))}\,.
\]
Using the fact that the function $\phi$ takes values  $\phi(u)\in [-\delta,  \delta]$, one finds  that the maximum and minimum frequency are given by
\begin{equation} \label{eq:rho1}
\rho^{max}_{p, v}=\frac{\tau}{2} e^{\delta}~, \qquad \rho^{min}_{p,
  v}=\frac{\tau}{2} e^{-\delta}~.
\end{equation}
These equations show that the frequency function $\rho_{p, v}$ allows  one to re-construct
the cosmological time $\tau(p)$ of the observer at the reception of the light signal and the discrepancy between his eigentime and the cosmological time, which is given by $\delta$. Moreover, by determining  the direction of the maximum, the observer can deduce the  direction of $p/\tau\in\mathbb H^2$.

\begin{figure}[ht]
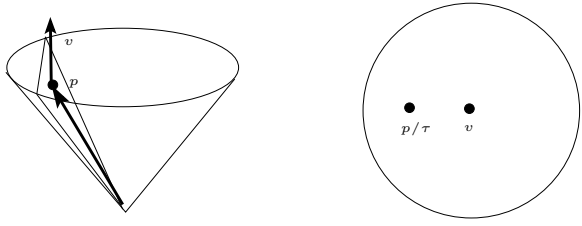

\begin{center}
\input futzero.pstex_t
\end{center}
\caption{Description of an observer for Example 1.}
\label{fig:obs1}
\end{figure}

\subsection{Example 2: future of a spacelike line} \label{ssc:ex2}

We consider the domain of dependence $D$ which is the future of a  spacelike
line $\RR e$ in $\mathbb R^{2,1}$ together with an observer given by a
point $p\in D$ and a timelike future directed timelike unit vector $v\in \hyp^2$. 
Using the symmetry of the system under translations in the
direction of the line and Lorentz transformations with the
line as their axis, we can parameterize the data for the observer as follows
\begin{align}\label{param}
e=(0,0,1)\qquad p=\tau(\cosh\delta,\sinh\delta,0)\qquad v=(\cosh\xi,0,\sinh\xi).
\end{align}
If we denote by $x_0\in\hyp^2$ the timelike unit vector corresponding
to the shortest line in $\mathbb R^{2,1}$ from $\RR e$ to $p$, then
$\xi$ is the hyperbolic distance of $v\in \hyp^2$ from from the
geodesic through $x_0$ that is orthogonal to $\RR e$. The parameter
$\delta$ is the hyperbolic distance from $x_0$ to  the projection of $v$ onto
this geodesic, as shown in Figure \ref{fig:obs2}.

The rescaled frequency function is defined as in Example 1. For $u\in T^1_v\hyp^2$, we have
\begin{align}
\rho_{p,v}(u)=\lim_{\epsilon\rightarrow 0} \epsilon \rho_{p,v}^\epsilon(u)\qquad
\rho_{p,v}^\epsilon(u)=\langle v, \nu_{p,v}(u)\rangle,
\end{align}
where $\nu_{p,v}(u)$ is the unit normal vector to the constant
cosmological time surface $H_\epsilon$ at the intersection point
$\{q_\epsilon(u)\}=H_\epsilon\cap (p+\RR(u-v)).$ As the constant
cosmological time surface $H_\epsilon$ is of the form
\begin{align}
H_\epsilon=\{ x\in \mathbb R^{2,1}\;|\; \skp x x -{\skp x e}^2=-\epsilon^2 \}~,
\end{align}
the intersection point $q_\epsilon(u)$ is given by the equations
\begin{align}
q_\epsilon(u)=p+t_\epsilon(u-v),\qquad \langle p+ t_\epsilon(u-v), p+ t_\epsilon(u-v)\rangle - {\skp
  {p+t_\epsilon(u-v)} e}^2=-\epsilon^2.
\end{align}
Using the parameterization \eqref{param},  in particular  the identities
$\skp {u-v}{u- v}=0$, $\skp p e=0$,  we obtain a quadratic equation in
$t_\epsilon$
\begin{align}
t_\epsilon^2 \langle e, u-v\rangle^2-2 t_\epsilon \langle p, u-v\rangle +\tau^2-\epsilon^2=0
\end{align}
with solutions
\begin{align}
t_\epsilon^\pm=\frac{\langle p, u-v\rangle\pm \sqrt{ {\skp p {u-v}}^2
    -({\tau^2-\epsilon^2}){\skp e{u- v}}^2}}{\skp e {u-v}^2 }.
\end{align}
In the limit $\epsilon\rightarrow 0$ this reduces to
\begin{align}
t_\pm=\frac{\langle p, u-v\rangle\pm \sqrt{ {\skp p {u-v}}^2 -\tau^2{\skp  e{u- v}}^2)}}{ {\skp  e {u-v}}^2 },
\end{align}
and the rescaled frequency function is given by
\begin{align}
\rho=-\skp v p - t_\epsilon (1-\skp {u-v} e \skp v e).
\end{align}
To obtain a concrete parametrization for $\rho$, we parameterize the unit vector $u\in T^1_{v}(\hyp^2)$ as
\begin{align}
u=\cos\theta (\sh \xi, 0, \ch \xi)+\sin\theta(0,1,0)~.
\end{align} 
This implies
\begin{align}
&\skp v p=-\tau\cosh\delta\cosh\xi~,\qquad \skp v e =\sh\xi\,, \nonumber \\
&\langle p, u-v\rangle=-\tau\ch\delta(\cos\theta\sh\xi-\ch\xi)+\tau\sh\delta\sin\theta\,,\nonumber\\
&\skp e {u- v}=\cos \theta \ch \xi-\sh\xi\,,\nonumber
\end{align}
and the expression for $t_\pm$ becomes
\begin{align}\label{texp}
t_\pm=\tau e^{\mp\delta}\,\frac{\ch\xi-\sh\xi\cos\theta\mp\sin\theta}{(\ch\xi\cos\theta-\sh\xi)^2}~.
\end{align}
If we introduce an ``angle variable" $\theta_\xi$ defined by
\begin{align}
\tan\frac{\theta_\xi} 2=e^\xi\qquad \text{with}\;\theta_\xi\in[0,\pi/2],
\end{align}
then we obtain 
\begin{align}\label{tauexp}
t_\pm=\frac{\tau e^{\mp \delta}\sin \theta_\xi}{2\sin^2\left(\frac{\theta\pm \theta_\xi}{2}\right)}.
\end{align}
Note that $t_\pm\geq 0$ and that $t_\pm$  diverges for
$\theta=\mp\theta_\xi$. The two cases for $t_\pm$, $\rho_\pm$
correspond to the intersection points of the ray $p+\RR(u-v)$ with the
two lightlike planes $Q_\pm$ containing $\RR e$. The relevant
intersection point is the one that is closer to $p$, i.e. the one with
$t=\text{min}(t_\pm)$. From \eqref{texp} it follows that this is the
one associated with $t_+$ if
\begin{align}\label{condi}
\left
(\frac{\sin\left(\frac{\theta-\theta_\xi}{2}\right)}{\sin\left(\frac{\theta+\theta_\xi}{2}\right)}\right)^2\leq
e^{2\delta}
\end{align}
and the one for $t_-$ otherwise.
For the associated frequency functions, we obtain
\begin{align}
\rho_\pm(\theta)=&\frac \tau 2
\cosh\xi\left(e^{\pm\delta}-e^{\mp\delta}\left(\frac{\ch\xi-\sh\xi\cos\theta\mp\sin\theta}{\ch\xi\cos\theta-\sh\xi}\right)^2\right)=\frac
\tau 2
\ch\xi\left(e^{\pm\delta}-e^{\mp\delta}\left(\frac{\sin\left(\frac{\theta\mp\theta_\xi}{2}\right)}
{\sin\left(\frac{\theta\pm\theta_\xi}{2}\right)}\right)^2\right).\nonumber
\end{align}
Clearly, $\rho_+(\theta)\geq 0$ if and only if \eqref{condi} is
satisfied, and $\rho_-(\theta)\geq 0$ otherwise. The frequency function is therefore given by
\begin{align}
\rho_{p,v}(\theta)=\text{max}(\rho_+(\theta),\rho_-(\theta))=\begin{cases}
\rho_+(\theta) & \text{for}\;\left
(\frac{\sin\left(\frac{\theta-\theta_\xi}{2}\right)}{\sin\left(\frac{\theta+\theta_\xi}{2}\right)}\right)^2\leq
e^{2\delta}\\ \rho_-(\theta) & \text{for}\;\left
(\frac{\sin\left(\frac{\theta-\theta_\xi}{2}\right)}{\sin\left(\frac{\theta+\theta_\xi}{2}\right)}\right)^2\geq
e^{2\delta}.
\end{cases}\end{align}
The frequency function  has local maxima in
\begin{align}
\phi_\pm^{max}=\pm\theta_\xi,
\end{align}
where it takes the values
\begin{align} \label{eq:rho2}
\rho_\pm^{max}=\frac \tau 2 \cosh\xi e^{\pm\delta},
\end{align}
and it  vanishes if and only if
\begin{align}
\left (\frac{\sin\left(\frac{\theta-\theta_\xi}{2}\right)}{\sin\left(\frac{\theta+\theta_\xi}{2}\right)}\right)^2= e^{2\delta}~.
\end{align}
This corresponds to $t_+=t_-$ or, equivalently,
\begin{align}
\skp {p-\tau e} {u-v}=0.
\end{align}
This  condition is satisfied if and only if  light ray $p+\mathbb R(u-v)$ intersects the line $\RR\cdot e$,
which, for each observer, happens for exactly two values of $\theta$.

The observer can therefore extract all relevant information from the
function $\rho_{p,v}(\theta)$. He can determine the cosmological time $\tau$ at the reception of the light signal, his position relative to the line, which is given by $\delta$, and his velocity
relative to the line, which is given by $\xi$.
The development of the measured  frequency function  in terms of the eigentime of a moving
observer is given by the dependence of his cosmological time and the
parameter $\delta$ on his eigentime. For an observer with a worldline specified by $p\in D$ and $v\in\mathbb H^2$, his  position  at an eigentime $t$ after the event $p$ is given by $p'=p+t v$. This implies that the cosmological time of $p'$ and the associated parameter $\delta$ are given by
\begin{align}
\tau(t)=\sqrt{\tau^2+t^2-2\skp{p}{v}}=\sqrt{\tau^2+t^2+\tau\cosh\delta\cosh\xi}~,\qquad \coth \delta(t)=
\coth\delta+t \frac{\ch \xi}{\ch \delta}~.
\end{align}
The time development of $\tau$ with the eigentime corresponds to an overall rescaling of
the frequency function. The time development of $\delta$ changes
the relation between its constant and its angle-dependent part.

\begin{figure}
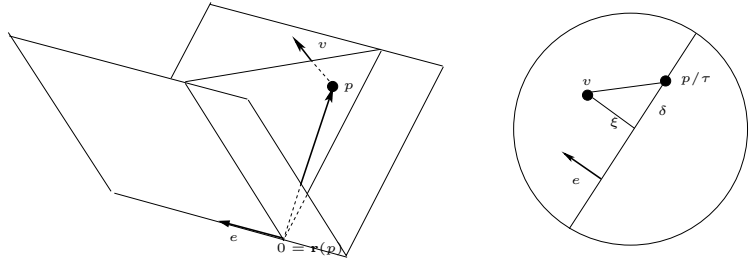

\begin{center}
\input futline.pstex_t
\end{center}
\caption{Description of an observer for Example 2.}
\label{fig:obs2}
\end{figure}

%{\red use geometrical definition of $\xi,\delta$ from later section. Make picture to explain it.}

\subsection{The rescaled density for a general domain} \label{ssc:basic}

The results in the previous subsections allow one to construct the
rescaled frequency function for a domain which can be obtained from the
light cone by the Mess construction \cite{mess}, i.~e.~via grafting along a weighted multicurve.  To show that the
rescaled density is well-defined also for the case of a general
geodesic lamination, one has to prove that the limit
$\lim_{\epsilon\rightarrow 0} \epsilon \rho_{p,v}^\epsilon$ exists for
these domains and for a general observer $(p,v)$.
As we will see, in general the problem is more subtle than it appears and some care is needed to
pass to the limit.

Let $D$ be a generic domain of dependence in $\mathbb R^{n,1}$.
Fix an observer in $D$  by specifying a point $p\in D$ and a future directed
timelike unit vector $v$. To analyze the behavior of the frequency function, it is 
convenient to express the function  $\epsilon\rho_{p,v}(\epsilon)$ 
as the composite of two functions  $q_\epsilon:T_v^1\mathbb H^2\rightarrow H_\epsilon$
and $\iota_\epsilon: H_\epsilon\rightarrow \mathbb R$, where
$q_\epsilon(e)$ is the intersection of the light ray $p+\mathbb R\cdot(e-v)$
with the level  surface of the cosmological time  $H_\epsilon=\tau^{-1}(\epsilon)$, and 
\[
\iota_\epsilon(q)=-\epsilon\langle \nu_\epsilon(q), v\rangle\,,
\]
where $\nu(q)$ denotes the unit normal vector of $H_\epsilon$ in $q$. 
It is clear that  the maps $q_\epsilon: T^1_v\mathbb H^2\to H_\epsilon$ converge to a map $q_0:T^1_v\mathbb H^2\to\partial D$   as $\epsilon\rightarrow 0$.
The idea is to show that the maps $\iota_\epsilon: H_\epsilon\to\mathbb R$ converge to a function $\iota: \partial D\to\mathbb R$ as 
$\epsilon\rightarrow 0$. 
Note, however,  that functions $\iota_\epsilon: H_\epsilon\to\mathbb R$ are defined on different domains, so we need to make this statement more precise.

Let $P_v$ be the hyperplane in $\mathbb R^{n,1}$ orthogonal to $v$.
By \cite{mess, bonsante},
the surfaces $H_\epsilon$ can be realized as  the graphs of  convex functions
$u_\epsilon: P_v\rightarrow\mathbb R$. More precisely, the level surfaces of the cosmological time can be identified with the  set of points $H_\epsilon=\{q=x+u_\epsilon(x)v|x\in P_v\}$. As $\epsilon\rightarrow 0$,
the functions $u_\epsilon$ converge to a convex function $u_0$, whose graph is 
the boundary $\partial D$.

Thus, there is a natural identification between 
$P_v$ and $H_\epsilon$ given by $x\mapsto x+u_\epsilon(x)v$.
In particular,  we may consider the functions $\iota_\epsilon$ as functions defined on $P_v$.
The following result then shows that the frequency function can be defined pointwise on the
boundary of any domain of dependence.

\begin{prop}\label{pr:conv-intensity}
For a fixed $x\in P_v$ the function $\epsilon\mapsto\iota_\epsilon(x)$ is increasing.
Moreover the function
\[
   \iota(x)=\lim_{\epsilon\rightarrow 0}\iota_\epsilon(x)=
   \inf_{\epsilon}\iota_\epsilon(x)
\]
is finite-valued at each point and locally bounded.
\end{prop}

The proof of this proposition will be based on the following technical lemma.

\begin{lemma}\label{lm:df}
The square of the cosmological time $\tau^2$ is convex along
each timelike line.
\end{lemma}

\begin{proof}
Take $r\in \mathbb R^{n,1}$ and
consider the cosmological time function on $I^+(r)$ which is defined
by
\[
   \tau_r(p)=\sqrt{-\langle p-r, p-r\rangle}
\]
It is clear that the restriction of $\tau_r^2$  
along every line $p+\mathbb R\cdot w$
is a degree two  polynomial function of the affine parameter
with leading coefficient  $-\langle w,w\rangle$.
In particular, the function $\tau_r$ is convex along all  timelike directions.

Given a point $q\in D$, let $r=\ret(q)$ be the corresponding point
on the singularity. Then $I^+(r)\subset D$ and on $I^+(r)$ we have  $\tau\geq\tau_r$.
Moreover, the cosmological time of $q$ and its gradient are given by $\tau(q)=\tau_r(q)$ and 
$\grad\tau(q)=\grad\tau_r(q)=\frac{1}{\tau}(q-r)$.

Take a timelike vector $w$ and consider the functions
$f(t)=\tau^2(q+tw)$ and $g(t)=\tau_r^2(q+tw)$. 
They are $C^{1,1}$-functions, which coincide with their derivatives at $t=0$.
As $f(t)\geq g(t)$ we deduce 
that if $f''$ exists in $0$ then $f''(0)\geq g''(0)>0$.
Thus  $f''(t)>0$ for all $t$ for which $f''$ exists. Since $f'$ is Lipschitz, for $s<t$
\[
  f'(t)-f'(s)=\int_s^tf''(x)dx>0~,
\]
and hence $f$ is convex.
\end{proof}

\begin{proof}[Proof of Proposition \ref{pr:conv-intensity}]
Let $D$ be a domain with cosmological time function $\tau$. It then follows from the results in \cite{bonsante} 
that  for all points  $q\in H(\epsilon)$ one has  $\grad\tau(q)=-\nu_\epsilon(q)$.
This implies
 $-\epsilon\nu_\epsilon=\tau\grad\tau=\frac{1}{2}\grad(\tau^2)$,
and we deduce 
\[
  \iota_\epsilon(x)=\frac{1}{2}\langle\grad\tau^2(x+u_\epsilon(x)v), v\rangle~.
\]
For a given point  $x\in P_v$,  we can consider the restriction of
$\tau^2$ to the vertical line $x+\mathbb R\cdot v$ 
\[
   f_x(s)=\tau^2(x+sv),
\]
which is a convex function of $s$ by Lemma \ref{lm:df}. As we have
\[
 \iota_\epsilon(x)=\frac{1}{2}(f_x)'(u_\epsilon(x)),
\]
 the monotonicity of $\iota_\epsilon$ then follows from the monotonicity of
$(f_x)'$.
\end{proof}

Proposition \ref{pr:conv-intensity} allows us to define 
 the rescaled frequency function for  a domain $D$ and an  observer $(p,v)$ as the
map 
\[
\rho_{(p,v)}: T^1_v\mathbb H^2\rightarrow\mathbb R_{\geq 0},\quad   \rho(e)=\rho_{(p,v)}(e)=\iota_v(q_0(e))~.
\]
Note, however,  that  it is in general not true that $  \rho_{\epsilon}(e)\rightarrow\rho(e)$ as $\epsilon \rightarrow 0$,  since the convergence $\iota_\epsilon\rightarrow\iota$
is  not necessarily uniform.
In the next section, we will investigate the regularity of $\iota$ and
show  that the convergence  of $\rho_{\epsilon}$ to $\rho$
holds for  generic observers and  generic directions.

\subsection{Domains with a closed singularity in dimension $2+1$} \label{ssc:closedsing}

To analyze the convergence properties of the rescaled frequency functions $\rho_\epsilon$, we  first consider the  case where the initial singularity $T$ is a closed subset of $\partial D$. In this situation,  the frequency function simplifies considerably. Note, however, 
that this condition never holds
for the universal covering of MGHFC spacetimes, as will be proved in the next section.
Nevertheless,  the results for this case are useful to compute the frequency function for domains of dependence that are the intersection of a finite number of half-spaces.

The simplifications in the case of a closed initial singularity arise from the fact that there is  an extension of the map $\ret$ on the boundary
of $D$, based on the following geometric idea.
For each point $q_0\in\partial D$,
there is a lightlike ray $R$ through $q_0$ which is contained in
$\partial D$. We will suppose that the lightlike ray $R$ is maximal with respect to
inclusion. The ray $R$ can always be extended to infinity in the future, but it has a
past endpoint $r_0\in\partial D$. The ray $R$ is  unique unless  $q_0=r_0$, which implies that
the point $r_0$ is uniquely determined by $q_0$. 

This defines a natural
map $\ret_0:\partial D\rightarrow\partial D$, and 
it  follows directly from its definition that $\ret_0\circ\ret_0=\ret_0$.
The image of $\ret_0$ is called the {\it extended singularity} and denoted by $\hat T$.
It contains all points which are the past endpoint of a maximal lightlike ray contained
in $\partial D$. In particular, the initial singularity $T$ is contained in $\hat T$.
Note that, unless $\hat T$ is closed,  the map $\ret_0$ cannot be continuous.
We will see in the next section that this non-continuity occurs in many interesting and relevant  examples.

\begin{prop} \label{pr:cont-int}
Consider the function $I_v:\partial D\rightarrow\mathbb R_{\geq 0}$ defined by
\[
     I_v(q_0)=\skp{q_0-\ret_0(q_0)}{v}~.
\]
If the singularity is closed in $\partial D$ then $T=\hat T$ and
$\iota_\epsilon$ uniformly converges to $I_v$. 
As a consequence, the function $\iota_v=I_v$ is  continuous. 
\end{prop}

\begin{proof}

For any point $x\in P_v$ we denote by $q_\epsilon(x)$ the point
$x+u_\epsilon(x)v\in H_\epsilon$ and by  $r_\epsilon(x)$ be the projection
of $q_\epsilon(x)$ on the initial singularity.
The results in  \cite{bonsante} imply that these points are related by the following equation \begin{equation}\label{eq:cl}
q_\epsilon(x)=r_\epsilon(x)+\epsilon\nu_\epsilon(x)~.
\end{equation}

Take now any sequence $x_n \in P_v$ that converges to $x$ and $\epsilon_n\rightarrow 0$. Then the associated sequence 
 $q_\epsilon(x)$ converges to
$q_0(x)=x+u_0(x)v$, and we obtain
\[
  \iota_{\epsilon_n}(x_n)=-\langle q_{\epsilon_n}(x_n)-r_{\epsilon_n}(x_n), v\rangle~.
\]
To prove that $\iota_\epsilon$  converges uniformly to $I_v$, it is then
sufficient to check that  $r_n=r_{\epsilon_n}(x_n)$ converges to $\ret_0(q_0)$
For this, note that the sequence
$r_n$
is contained  in a compact subset of the boundary $\partial D$ and hence has a subsequence
$(r_{n_k})_{k\in\NN}$ which converges to a point $r_0\in\partial D$. By the assumption 
on $T$, the point $r_0$ is also contained in $T$. 
We prove in the next paragraph that $r_0=\ret(q_0(x))$. The uniqueness of the limit implies
that the whole sequence $r_n$ converges to $r_0$.

That $r_0=\ret_0(q_0)$ can be established as follows.
The sequence  of timelike vectors 
$q_{n_k}-r_{n_k}$  converges to $q_0-r_0$.
This implies that $q_0-r_0$ is not spacelike. 
As $\partial D$ is an achronal surface, it must be lightlike and
the lightlike ray $R=r_0+\mathbb R_{\geq 0}\cdot(q_0-r_0)$  is contained
in $\partial D$. Since $r_0$ is on the singularity,  we obtain that the ray $R$ is maximal
so that $r_0=\ret(q_0)$.
This also shows that the image of $\ret_0$ is contained
in $T$.
\end{proof}

\begin{remark}
In the general case, we cannot conclude because the limit 
point $r_0$ may not be on the singularity.
However, it is always true that the ray $q_0+\mathbb R\cdot (q_0-r_0)$
is contained in $\partial D$, so it is contained in the maximal lightlike
ray through $q_0$. In other words the point $r_0$ lies on the segment
$[q_0, \ret_0(q_0)]$.
This  shows that in general
\[
  \limsup \iota_{\epsilon_n}(x_n)\leq I_v(x)
 \]
for any sequence $\epsilon_n\rightarrow 0$ and $x_n\rightarrow x$.
In particular, it implies  $\iota_v(x)\leq I_v(x)$ and hence that $\iota_v$ is zero
on the initial singularity.
\end{remark}

\subsection{Generic domains in dimension $2+1$} \label{ssc:generic21}

Although at a first sight,  the hypothesis of Proposition \ref{pr:cont-int}
could appear to hold generally, this is not the case. Indeed,  
 if $D$ is the universal covering of a MGHFC
spacetime whose holonomy representation  is not linear, then the condition cannot be
satisfied.

\begin{remark}\label{rk:sing-struct}  \cite{mess}
Let $D$ be the universal covering of of a MGHFC spacetime. Then 
$T$ is never closed in $\partial D$.
\end{remark}

%{\red Add a reference to Mess, or something else?}

The following proposition describes  the regularity  properties of 
of the functions $\iota_v$ and $I_v$ for in general domains in 2+1 dimensions.
The result is that at generic points,  these functions are continuous and coincide.

\begin{prop} \label{pr:meagre}
The following properties hold for the functions $\iota_v$ and $I_v$:
\begin{itemize}
\item The function $\iota_v$ is upper  semicontinuous.
\item The set of discontinuity points of $\iota_v$ is meagre.
\item The function $I_v$ is upper semicontinuous.
\end{itemize}
\end{prop}

\begin{proof}
The first property holds since $f$ is the supremum of a family of continuous functions.
Moreover, as  it is the limit of continuous functions, by a classical result of Lebesgue,
its discontinuity points form a meagre set.

Let us prove that   $I_v$ is upper semi-continuous.
For this, take a  sequence of points  $x_n\in P_v$ that converges to $x$.
Up to passing to a subsequence we may assume that
$\limsup I_v(x_n)=\lim I_v(x_n)$.
If $\lim I_v(x_n)=0$, then clearly
$I_v(x)\geq \limsup I_v(x_n)$. 

Let us treat the case where $\lim I_v(x_n)>0$.
The  sequence of points
$q_n=x_n+u(x_n)v\in\partial D$ converges to
$q=x+u(x)v \in\partial D$. 
By the assumption on $\lim I_v(x_n)$ we have that $\ret_0(q_n)\neq q_n$ for $n$ sufficiently  large. 
Consider the sequence of lightlike rays $R_n$ containing
$q_n$ and $\ret_0(q_n)$. Up to passing to a subsequence,  
we may assume that it converges to a lightlike ray $R$ 
through $q$. The sequence $(\ret_0(q_n))_{n\in \NN}$ converges to the past endpoint
of $R$. Since $R$ is contained in $\partial D$ we deduce that
$r_1=\lim\ret_0(q_n)$ is a point on the segment $[\ret_0(q),q]$ and
we  have
\[
  \lim_{n\rightarrow+\infty} I_v(x_n)=-\langle q-r_1, v\rangle\leq \langle q-\ret_0(q), v\rangle=I_v(q)~.
\]
%Finally let us prove that if $\iota_v$ is continuous at $x\in P_v$ then
%$\iota_v(x)=I_v(x)$.
%Since we are assuming that $\hat T$ coincides with $T$, the point
%$\ret_0(q)$ lies on the singularity. In particular this means that the direction
%of $q-\ret_0(q)$ corresponds to an ideal point $\xi$ 
%of $\mathcal F_{\ret_0(q)}$. Take a sequence $v_n\in\mathbb H^2$ such that
%$v_n$ lies in $\mathcal F_{\ret_0(q)}$ and converges towards $\xi$.
%
%There is a sequence of $\epsilon_n\rightarrow 0$, such that
%$\epsilon_nv_n\rightarrow q-\ret_0(q)$.
%
%Consider the sequence of points $q_n=\ret_0(q)+\epsilon_n v_n$.
%By construction we have
%\begin{itemize}
%\item $q_n\in D$ and $\ret(q_n)=\ret_0(q)$;
%\item $\tau(q_n)=\epsilon_n$;
%\item $\iota_{\epsilon_n}(q_n)=-\langle \epsilon_nv_n, v\rangle$;
%\item $q_n\rightarrow q$ as $n\rightarrow+\infty$.
%\end{itemize}
%
%Now let $x_n$ be the orthogonal projection of $q_n$ to $P_v$.
%Clearly $x_n\rightarrow x$ so $\iota_v(x)=\lim \iota_v(x_n)$ by the assumption
%that $\iota_v$ is  continuous at $x$.
%
%Since 
%$\iota_v(x_n)\leq\iota_{\epsilon_n}(x_n)$ we deduce that
%$\iota_v(x)\geq\limsup \iota_{\epsilon_n}(x_n)$.
%By construction $\limsup\iota_{\epsilon_n}(x_n)=\lim \iota_{\epsilon_n}(x_n)=I_v(x)$ so we get
%$\iota_v(x)\geq I_v(x)$.  And this implies that $\iota_v(x)=I_v(x)$.
\end{proof}

%\begin{remark}
%The assumption that $T$ coincides with $\hat T$ could be removed
%in the last point of Proposition \ref{}. On the other hand, the proof
%would be more technical. Since we will be interested only
%in domains of dependence where the assumption is satisfied we 
%preferred to prove the properties using this technical hypothesis.
%\end{remark}

\section{Stability of the frequency function} \label{sc:stability}

In this section we  investigate the stability of the frequency function.
Given a sequence of  domains of dependence 
$D_n$ that converges to $D$ and  an observer $(p,v)$ in $D$, then $(p,v)$ is also an observer in $D_n$ for $n$ sufficiently large. Let  $\rho^n_v$ be the frequency function for  $D_n$ as seen by the observer $(p,v)$ and $\rho$ the associated frequency function for $D$. 
We will investigate  under which conditions the frequency functions  $\rho_v^n$ converge to the frequency  function $\rho$.

We will first show by a counterexample in subsection \ref{ssc:example}  that in general there is no
convergence even in the weak sense.  However, we identify a subclass of
domains of dependence, called domains of dependence with a flat boundary, which includes the interesting examples.  In
subsection \ref{ssc:flat} we  prove that the convergence holds for these
domains.  In subsection \ref{ssc:coverings} we will then  prove that universal
coverings of MGHFC spacetimes in dimension 2+1 are contained in this class.
As in the previous subsection, it is  advantageous to
work with  the frequency functions $\iota_v$ and $\iota^n_v$ introduced there.

\subsection{An example} \label{ssc:example}

We fix coordinates $x_0, x_1, x_2$ on $\mathbb R^{2,1}$, so that
the Minkowski metric takes the form $-dx_0^2+dx_1^2+dx_2^2$ and consider the timelike vector  $v=(1,0,0)$. Let $P$ be the horizontal plane
at height equal to $1$.  Then the  intersection of $P$ with the cone
$I^+(0)$ is a circle $C$ of radius $1$. Let $C_k$ be the regular
polygon with $k$ edges tangent to $C$. Clearly,  $C_k$ converges to $C$
in the Hausdorff sense as $k\rightarrow+\infty$.

Now observe that  for each edge of $C_k$ the plane  that contains $0$ and this edge  is
lightlike, since it is tangent to $I^+(0)$. Let $D_k$ be the
intersection of the future of the $k$ lightlike planes containing $0$ and the
edges of $C_k$. Then $D_k$ is a domain of dependence and converges
 to $D$ for $k\rightarrow \infty$  on compact subsets.

We denote by  $\iota^k_v$ and $\iota_v$, respectively,  the frequency function of $D_k$ and $D$ with respect to the observer $(p,v)$ and  prove that $\iota^k_v$ does not converge to $\iota_v$ on a set of positive measure. 
Regarding $L^1$-functions as
continuous functionals on the set of continuous functions with compact
support, we will show that $\iota_k^v$ weakly converges to
$\frac{1}{2}\iota_v$. In other words, for every continuous function $\phi$
with compact support we have
\[
   \int_{P_v}\phi \iota^n_v dV\rightarrow\frac{1}{2}\int_{P_v}\phi \iota_v dV
\]
where $dV$ is the area measure of the horizontal plane $P_v$.  In
particular,  in this example $\iota^k_v$ does not converge to $\iota_v$ even in
this weak sense.

The computation of $\iota_v$ can be performed as follows. The initial
singularity of $I^+(0)$ reduces to the point $0$, and hence is a closed subset. 
It turns out that, for every point $x\in P_v$ the frequency function $\iota_v$ is given by $\iota_v(x)=-\langle x+\|x\|v,v\rangle=\|x\|$.

Consider now the domain $D_k$. The initial singularity of $D_k$ is the set of
spacelike lines joining $0$ to the vertices of the polygon $C_k$. So
when we project on the plane $P_v$, the initial singularity appears as
the union of $k$ rays $s_1,\ldots s_k$ starting from $0$, so that the
angle between $s_j$ and $s_{j+1}$ is $2\pi/k$.  Let
$f_k:P_v\rightarrow\mathbb R$ the function whose graph if the boundary
of $D_k$.  On the region $P_j$ of $P_v$ that is bounded by $s_j$ and
$s_{j+1}$, the function $f_k$ is differentiable, and the gradient of $f_k$ is a unit
vector whose angle with $s_j$ is equal to $\pi/k$. The integral lines of the gradient are parallel lines that form 
an angle equal to $\pi/k$ with both $s_j$ and $s_{j+1}$.
If we denote by  $r(x)$ the intersection point of the line through $x$ with
the singularity, then the frequency function is given by $\iota_v^k(x)=f_k(x)-f_k(r(x))=||x-r(x)||$.

Let now $s'_{j}$ be the unique line of this foliation which starts at $0$.
Clearly,  it is the bisector of $P_j$. If $x$ is on the right of $s'_j$ then 
$r(x)\in s_j$. If $x$ is on the left of $s'_j$ then $r(x)\in s_{j+1}$.
 We now consider the set 
$$ E_j^k=\{x\in P_j~|~ \iota_v^k(x)\leq \|x\|/2\}~. $$ 

For each point $x\in P_j$, consider the triangle with vertices
at $x$, $r(x)$ and $0$. Note that $\iota_v^k(x)$ is the length of the edge joining
$x$ to $r(x)$. The sine formula of Euclidean triangles then shows that
\[
       \iota_v^k(x)=\|x\|\frac{\sin\phi}{\sin(\pi/k)}
\]
where $\phi$ is the angle at the vertex $0$ in the above triangle.
Thus, let $\phi_k$ be such that $\sin\phi_k=\frac{1}{2}\sin(\pi/k)$.
Let $s''_j$ and $s'''_j$ be the rays in $P_j$ forming an angle $\phi_k$
respectively with $s_j$ and $s_{j+1}$. Then, $E_j^k$ is the union of two sectors
bounded respectively by $s_j$ and $s''_j$ and by $s_{j+1}$ and $s'''_j$.

By the concavity of the function $\sin$ in $[0,\pi/2]$ we have
$\phi_k>\pi/2k$, so for any radius $R$ the area of
$E^k_j\cap B(0,R)$ is bigger than $\frac 1 2$ times the area of $P_j\cap B(0,R)$.
Now let us consider the set
$$ E^k=\{x\in P_v| \iota_v^k(x)\leq \|x\|/2\}=\bigcup E^k_j~. $$ 
%% We will prove that the NB this seems to be clear??
The area of $E^k\cap B(0,R)$ is the sum of the areas of the
$E^k_j\cap B(0,R)$, so that
\[
Area(E^k\cap B(0,R))\geq \frac{1}{2}\sum Area(P_j)=\frac{\pi R^2}{2}~
\] 
and, consequently, 
\[
  \int_{B(0,R)}(\iota_v-\iota^k_v)dV\geq
  \int_{B(0,R)\cap E^k}(\iota_v-\iota^k_v)dV\geq
  \int_{B(0,R)\cap E^k}\frac{\| x\|}{2}dV~.
\]
As $E_k$ is a cone from the origin and  the function $x\to \| x\|/2$ depends only on the distance from the origin, it follows that
\[
\int_{B(0,R)}\iota_v-\iota^k_v\geq \frac 12\int_{B(0,R)}\frac{\| x\|}{2} = \pi R^2/4~.
\]
In particular, this shows that 
$\iota^k_v$ does not converge weakly to $\iota_v$.

\begin{prop} \label{pr:weakly}
The sequence $\iota_v^k$ weakly converges to $\iota_v/2$ in $L^1_{loc}(P_v)$. 
\end{prop}

\begin{proof}
Note that since $\iota_v^k(x)\leq\| x\|$, up to passing to a subsequence,
there is a weak limit 
in $L^1_{loc}(P_v)$, say $J$.
We will prove that $J=\iota_v/2$. This is 
sufficient to deduce that the whole sequence converges to $\iota_v/2$.

For this, we consider the following sequence $   \omega_k=*du_k$
of $1$-forms on $P_v$, where $u_k$ is the function whose graph is $\partial D_k$ and $*$ is the  Hodge operator.
Note that $\omega_k$ is a $L^\infty$ $1$-form defined in the complement of the singularity.

As $du_k\rightarrow du$ at every differentiable point of $u$ and $||du_k||\leq 1$, 
by the Dominated Convergence Theorem we have   that $du_k\rightarrow du$ strongly in 
$L^1_{loc}(P_v)$ as $k\rightarrow+\infty$.
%{\red FB: here I do not have found the right reference}
This implies that $\omega_k\rightarrow\omega=*du$ strongly in $L^1_{loc}(P_v)$.

We claim that for any  compact supported smooth function $f$ the following formula holds:
\begin{equation}\label{mainex:eq}
\int_{P_v}df\wedge J\omega= \frac{1}{2}\int_{P_v}df\wedge \iota_v\omega~.
\end{equation}
Since $df\wedge\omega=\frac{\partial f}{\partial\rho} dV$, this formula implies 
\[
  \int_{P_v}\frac{\partial f}{\partial\rho}(\iota_v/2-J)dV=0,
\]
and by a simple density argument we conclude that $J=\frac{1}{2}\iota_v$.

To prove the claim, first note that
\[
\int_{P_v}df\wedge J\omega=\lim\int_{P_v}df\wedge \iota_v^k\omega_k~.
\]
Now, note that on each region $P_j$ of the complement of the
singularity,  $\iota^k_v\omega_k$ is a smooth $1$-form, and its differential
is equal to 
\begin{eqnarray*}
d(\iota^k_v\omega_k) &=& 
d\iota^k_v\wedge \omega_k + \iota^k_vd\omega_k  \\
& = &  d\iota^k_v\wedge \omega_k + \iota^k_vd(*du_k)  \\
& = & d\iota^k_v\wedge\omega_k+\iota^k_v\Delta u_kdV~. 
\end{eqnarray*}
Since  the
function $u_k$ is affine on $P_j$, the last term vanishes. Moreover, we have
$d\iota^k_v\wedge\omega_k= d\iota^k_v(\grad u_k)dV=dV$, since, on the integral
lines of the gradient of $u_k$, $\iota^k_v$ is an affine function with
derivative equal to $1$.

Using the fact that $\iota^k_v\wedge\omega_k$ vanishes on the singularity, we
 obtain 
\[
  \int_{P_j}df\wedge J\omega=-\int_{P_j}fdV,
\]
which implies $\int_{P_v}df\wedge \iota^k_v\omega_k=-\int_{P_v}fdV$, and we conclude that
\begin{equation}\label{J:eq}
\int_{P_v}df\wedge J\omega=\lim\int_{P_v}df\wedge \iota^k_v\omega_k=-\int_{P_v}fdV~.
\end{equation}
On the other hand %, in oder to compute $\int_{P_v}df\wedge \iota_v\omega$,
note that on $P_v\setminus\{0\}$ we have the identity
$$ d(\iota_v\wedge\omega)=d\iota_v\wedge\omega+\iota_v\wedge d\omega~. $$ 
Now as before
$d\iota_v\wedge\omega=dV$, but $\iota_vd\omega=\iota_v\Delta u dV=\frac{\iota_v}{u}dV=dV$, where the last equality
holds since $\iota_v=u$ in this  example. 
So if $B_\epsilon$ is the disk centered at $0$ with radius $\epsilon$ we
have
\[
\int_{P_k\setminus B_{\epsilon}}df\wedge \iota_v\omega=-\int_{P_k\setminus B_{\epsilon}}2fdV-\int_{\partial B_\epsilon}f\iota_v\omega~.
\]
As $\omega$ is bounded, the last term vanishes as $\epsilon\rightarrow 0$ so we deduce
\[
\int_{P_k}df\wedge \iota_v\omega=-\int_{P_k}2fdV.
\]
Equation \eqref{mainex:eq} then follows by comparing this equation with \eqref{J:eq}. 
\end{proof}

\begin{remark}
Proposition \ref{pr:weakly} makes it clear that the reason why $\iota_v^k$ does not converge to $\iota$
is the fact that $\Delta u$ is not concentrated on the singularity, 
whereas $\Delta u_k$  vanishes outside the singularity.
This remark will lead us below to introduce the notion of domain of dependence with flat boundary, where this problem is excluded.  We will then  prove (Theorem \ref{stability:thrm}) that, for domains of dependence with
flat boundary, the convergence does hold. We will then show in Section \ref{ssc:bounds} that, without this
hypothesis, although the sequence  $\iota_v^k$ do not converge to $\iota_v$, it does have a limit, and this limit differs only by bounded factor (attained in the example presented above) from $\iota_v$.  
\end{remark}

We  conclude this section with a simple remark.
In the example above we have seen a sequence of domains of dependence
$D_n$ which converges  to a domain $D$, but for which the corresponding  sequence of
 frequency functions $\iota^n_v$ does not converge to the frequency $\iota_v$ of $D$.
However,  in it is clear that 
\[
   \iota_v(x)\geq\limsup_{n\rightarrow+\infty}\iota_v^n(x)
\]
This estimate holds in general and is a consequence of two facts:
 \begin{itemize}
 \item
 The frequency functions $\iota_\epsilon^n$
of the  surfaces of $H_\epsilon\subset D_n$ converge to the frequency function of
the $H_\epsilon\subset D$. 
\item
The frequency function of any domain is the infimum of the     frequency functions     of  its surfaces $H_\epsilon$ of constant cosmological time.
\end{itemize}
We  include a proof  of the first statement for the sake of completeness.

\begin{prop}\label{pr:usc}
Let $D_k$ be a sequence of domain of dependence converging to a domain $D$.
Denote by $\iota_v^k$ the frequency function of $D$ and by $\iota_v$ the frequency function
of $D$. Then for every $x\in P_v$ we have 
$\iota_v(x)\geq\limsup_{k\rightarrow+\infty}\iota_v^k(x)$.
\end{prop}

\begin{proof}
Let us fix $\epsilon$.
Denote by $\iota_\epsilon^k$ the frequency function of the level surface 
$H^k_\epsilon=\tau_k^{-1}(\epsilon)$ of the cosmological time of $D_k$.
By \cite{bonsante}, we know that  the sequence of surfaces $H^k_\epsilon$ converges to the level surface 
$H_\epsilon=\tau^{-1}(\epsilon)$ of $D$ as $k\rightarrow+\infty$.
This means that the function $u_\epsilon^k:P_v\rightarrow\mathbb R$, whose graph is $H^k_\epsilon$, 
converges as $k\to \infty$ to the function $u_\epsilon:P_v\rightarrow\mathbb R$ which defines $H_\epsilon$.
By convexity, $\grad u_\epsilon^k(x)$ converges to $\grad u_\epsilon(x)$.
As $\iota_\epsilon^k$ is given by  $$\iota_\epsilon^k(x)=\frac{\epsilon}{\sqrt{1-\|\grad u_\epsilon^k\|^2}},$$ it follows that
$\iota_\epsilon^k(x)\rightarrow\iota_\epsilon(x)$ as $k\rightarrow+\infty$.
Now note that $\iota_v^k(x)\leq\iota_\epsilon^k(x)$ for every $k$. 
So passing to the $\limsup$ we obtain
\[
   \limsup\iota_v^k(x)\leq\iota_\epsilon(x)~
\]
and by taking the infimum over $\epsilon$
\[
  \limsup\iota_v^k(x)\leq\iota_v(x)~.
\]
\end{proof}

\subsection{Domains of dependence with flat boundary}
\label{ssc:flat}

Let $P$ be a lightlike plane in $\mathbb R^{n,1}$ and denote by  $g$ 
the degenerate metric on $P$ induced by the Minkowski metric.  We note that $P$ is foliated by lightlike lines
which are parallel to the kernel of $g$ and denote by  $P/L$ be the
space of leaves.

\begin{lemma} \label{lm:flat}
$P/L$ is equipped with a flat metric $\hat g$ that makes it isometric
  to $\mathbb R^n$ such that $g$ is the pull-back of $\hat g$ by the
  natural projection $P\rightarrow P/L$
\end{lemma}

Lemma \ref{lm:flat} implies that there is a natural $(n-1)$-form $\omega_P$ on
$P$ defined as the pull-back of the area form of $g$. This form has
the following characterization. If $S$ is any spacelike compact
hypersurface in $S$ oriented by a future-oriented transverse direction, 
its area is equal to the integral of $\omega_P$ on $S$.

Now, given a domain of dependence $D$ 
we consider the $1$-form $\omega$ on $\partial D$ defined in the complement
of the singularity. If $p$ is not on the singularity and $P$ is the unique support
plane at $x$, then $\omega_x=\omega_P$.

Regarding $\partial D$ as the graph of a function $u$ on some fixed spacelike plane $P_v$,
it turns out that the identification between $P_v$ and $\partial D$ is differentiable at each
point where $u$ is differentiable, so in the complement of the singularities.
We can therefore express  $\omega$ as a form on $P_v$.

\begin{lemma}
$\omega=*du$ where $*$ is the Hodge star operator of $P_v$.
\end{lemma}

Note in particular that $\omega$ is a $L^\infty$-form.
Moreover, since $u$ is convex, the differential of $\omega$ defined as a distribution
on $\mathbb R^n$ is in fact a positive locally finite Radon measure. 
More precisely, we have the following result.

\begin{lemma}\label{lm:radon}
$d\omega=\Delta u$, where $\Delta u$ is a positive Radon measure.
\end{lemma}

\begin{proof} %% [Proof of Lemma \ref{lm:radon}]
Let $(u_n)_{n\in \NN}$ be a sequence of smooth convex functions converging to $u$. 
Then $d*du_n=\Delta u_n\to \Delta u$ as distributions. Since the  $u_n$ are convex, the 
$\Delta u_n$ are positive, so their distribution limit $\Delta u$ is a positive
distribution. Therefore (essentially by the Riesz representation theorem) it is 
a Radon measure.
\end{proof}

\begin{defi}
The boundary of  a domain $D$ is called \emph{flat}, if $d\omega$ is a measure concentrated on the singularity.
\end{defi}

Let us recall that a Radon measure $\mu$ on $P_v$ is concentrated on $A$ if $\mu(P_v\setminus A)=0$. It is not difficult to check that if the boundary of $D$ is locally given as the union of a finite number of 
lightlike planes, then it is flat. In Section \ref{ssc:coverings} we will see that a domain of dependence in $\R^{2,1}$ that
is the universal cover of a MGHC flat spacetime of genus $g\geq 2$ always has flat boundary.
The choice of terminology is due to the following lemma.

\begin{lemma} \label{flat:lm}
Let $D$ be a domain of dependence with a flat boundary. 
Suppose that $A$ is an open subset of $\partial D$ which does not meet the singularity.
Then $A$ is contained in a lightlike hyperplane.
\end{lemma}

\begin{proof} %[Proof of Lemma \ref{flat:lm}]
%Lemma \ref{lm:boundII} can be applied at each point of $A$, where it shows that 
%$\Delta u\geq 0$, with 
%equality if and only if $\hess(u)=0$. 
%Since $\Delta u$ is concentrated on the singularity, it vanishes
%on $A$, so that $\hess(u)$ also vanishes on $A$, and $A$ is therefore contained in 
%a lightlike hyperplane.
As $\hess(u)$ is a measure with values in the positive definite quadratic forms, one finds that
if $\Delta u=0$, then $\hess(u)$ is zero on $A$.
This implies that $-u$ is a convex function too, so $u$ is an affine function,  i.~e.~the graph of $u|_A$ is a plane. Since it is contained in the boundary of $D$, this plane must be lightlike.
\end{proof}

\begin{remark}
If $D$ is a domain of dependence obtained as the intersection of the future
of a locally finite family of lightlike planes in $\mathbb R^{n,1}$, then its boundary is clearly
flat according to the previous definition.
Moreover,  Lemma \ref{flat:lm} shows that if $D$ is a domain with flat
boundary and the initial singularity is closed in $D$, then $D$ is the intersection
of the future of a locally finite family of lightlike planes.

However, in the next subsection we will show some interesting examples of
domains of dependence with flat boundary which are not of polyhedral type.
This depends on the fact that in those examples the initial singularity is not
a closed subset.  In fact,  in many cases the singularity is dense.
\end{remark}

If $D$ is a domain of dependence with flat boundary, the following proposition holds.
This proposition  will be the key ingredient in  the proof of the stability of the frequency functions.

\begin{prop}\label{mainflat:prop}
Let $v$ be a  timelike unit vector.
If $D$ is a domain of dependence with flat boundary then for every compactly supported 
smooth function $f$ on $P_v$ the following identity holds:
\[
\int_{P_v}df\wedge(\iota_v\omega)=-\int_{P_v} fdV~.
\]
\end{prop}

The proof of the proposition is based on the following lemma,
which is also valid if $D$ does not have flat boundary. %% note I have added this remark (useful below).

\begin{lemma} \label{lm:g}
Let $u_\epsilon:P_v\rightarrow\mathbb R$ be the function whose graph is the level surface $H_\epsilon$
and let $\iota_\epsilon$ be the frequency function for $H_\epsilon$.
Then at every point where $\grad u_\epsilon$ is differentiable we have
\begin{eqnarray}
   0<\langle\grad\iota_\epsilon, \grad u_\epsilon\rangle\leq 1~,\label{eq:g}\\
   \iota_\epsilon\Delta(u_\epsilon)\leq n-||\grad u_\epsilon||^2~.\label{eq:gg}\nonumber
\end{eqnarray}
\end{lemma}

\begin{proof}
Since the vector $\grad u_\epsilon+v$ is orthogonal to the surface $H_\epsilon$, the unit normal future-oriented vector
is obtained by normalizing it  as
$$ \nu_\epsilon=\frac{1}{\sqrt{1-\|\grad u_\epsilon\|^2}}(\grad u_\epsilon+v), $$ 
which implies
\[
   \iota_\epsilon=-\epsilon\langle \nu_\epsilon,v\rangle=
   \frac{\epsilon}{\sqrt{1-\|\grad u_\epsilon\|^2}}~. 
\]
Thus at every point where $\|\grad u_\epsilon\|$ is differentiable we have
\[
\langle \grad\iota_\epsilon, \grad u_\epsilon\rangle=
\frac{\epsilon}{(1-\|\grad u_\epsilon\|^2)^{3/2}}\langle \hess(u_\epsilon)\grad u_\epsilon, \grad u_\epsilon\rangle~,
\]
which is positive by the convexity of $u_\epsilon$.

To prove the estimate from above we use a comparison argument.
Consider the retraction of the point  $q=x+u_\epsilon(x)v$ on the singularity, say $r=\ret(q)$.
Note that on $I^+(r)\cap D$ the function $f(p)=\sqrt{-\langle p-r,p-r\rangle}$ 
is less than than the cosmological time.
It follows that the level surface $H'=f^{-1}(\epsilon)$ is contained in the closure of the future
of $H_\epsilon$. Moreover, since $\tau(q)=f(q)=\epsilon$ 
those surfaces are tangent at the point $q$.

Note that $H'$ is the graph of the function $$h:P_v\rightarrow\mathbb R, \quad
h(x)=c+\sqrt{\| x-\bar r\|^2+\epsilon^2},$$ where $c\in\mathbb R$ and $\bar r\in P_v$ are determined by the orthogonal decomposition
     $r=\bar r+cv~.$

From the fact that $H'$ is tangent to $H$ at $q$ and that it is contained in its epigraph, one deduces
\begin{itemize}
\item $u_\epsilon(q)=h(q)$,
\item $\grad u_\epsilon(q)=\grad h(q)$,
\item $\hess u_\epsilon(q)\leq \hess h(q)$,
\end{itemize}
where the last inequality implies that the difference is a positive definite matrix.
In particular, we find that at the point $q$
\[
   \langle \grad\iota_\epsilon, \grad u_\epsilon\rangle\leq
   \frac{\epsilon}{(1-\|\grad u_\epsilon\|^2)^{3/2}}\langle \hess(h)\grad h, \grad h\rangle~.
\]
Now an explicit computation shows that
\[
  \grad h=\frac{1}{h-c}(x-\bar r)~, \qquad \hess h=\frac{1}{h-c}(Id-\grad h\otimes\grad h),
\]
which implies that, still at the point $q$, the following inequalities hold
\begin{eqnarray*}
\langle \grad(\iota_\epsilon), \grad (u_\epsilon)\rangle 
& \leq & \frac{\epsilon}{(1-\|\grad (u_\epsilon)\|^2)^{3/2}(h-c)}(\|\grad (h)\|^2-\|\grad (h)\|^4) \\
& \leq & \frac{\epsilon}{(1-\|\grad (u_\epsilon)\|^2)^{3/2}(u-c)}(\|\grad (u_\epsilon)\|^2-\|\grad (u_\epsilon)\|^4)
\end{eqnarray*}
Using  $\|\grad u_\epsilon\|<1$  and the identities
$$ u_\epsilon(q)-c= -\langle q-\ret(q), v\rangle=\iota_\epsilon(q)=\frac\epsilon{\sqrt{1-\|\grad (u_\epsilon)\|^2}},$$ 
we obtain 
\[
\langle \grad\iota_\epsilon, \grad u_\epsilon\rangle\leq\|\grad u_\epsilon\|^2\leq 1.
\]
To prove (\ref{eq:gg}), it is then sufficient to note that
$$\Delta u_\epsilon(q)\leq \Delta h=\frac{1}{\iota_{\epsilon}(q)}(n-||\grad u_\epsilon||^2)~.$$ 

\end{proof}

\begin{proof}[Proof of Proposition \ref{mainflat:prop}]
Let $\iota_\epsilon$ be the frequency function of the surface $H_\epsilon$,
let $u_\epsilon: P_v\rightarrow\mathbb R$ be the $C^{1,1}$-function whose graph
is the surface $H_\epsilon$, and  set $\omega_\epsilon=*du_\epsilon$.
Then $\omega_\epsilon\rightarrow\omega$ in $L^1_{loc}$.
Moreover $\omega_\epsilon$ is a Lipschitz form.
Analogously, we find that $\iota_\epsilon$ is a Lipschitz function since
\[
   \iota_\epsilon(x)=\frac{1}{\sqrt{1-\|\grad u_{\epsilon}\|^2}}.
\]
It follows that for every   smooth function with compact support $f$ we have
\begin{equation}\label{eq:decomposition}
 \int_{P_v}df\wedge(\iota_\epsilon\omega_\epsilon)=
-\int_{P_v}f d\iota_{\epsilon}\wedge\omega_{\epsilon}
-\int_{P_v}f\iota_{\epsilon}\Delta u_\epsilon dV~.
\end{equation}
As $\iota_\epsilon\searrow\iota$ pointwise and $f$ has compact support,
there exists a constant $C$ such that
$|f\iota_{\epsilon}|<C$ for $\epsilon<1$. It then follows by the Dominated Convergence Theorem that
\[
\int_{P_v}f\iota_{\epsilon}\Delta u \rightarrow\int_{P_v}f\iota\Delta u =0~,
\]
where the last equality holds because $\iota$ is zero on the singularity and $\Delta u$ is
concentrated there.

On the other hand, we have  $\Delta u_\epsilon dV\rightarrow\Delta u dV$ as measures, which implies
\[
\left|\int f\iota_{\epsilon}\Delta u_\epsilon dV-\int f\iota\Delta u \right|<
C\left|\int_K\Delta u_\epsilon dV-\int_K\Delta u \right|\rightarrow 0~.
\]
As $\iota$ vanishes on the singularity, it follows that 
the last term on the right-hand side of \eqref{eq:decomposition} converges to $0$.
To conclude,  it is sufficient to show that
\[
\int_{P_v}f d\iota_{\epsilon}\wedge\omega_{\epsilon}\rightarrow\int_{P_v}fdV~.
\]
Now the $2$-form $d\iota_{\epsilon}\wedge\omega_{\epsilon}$ is equal to $d\iota_{\epsilon}(\grad u_{\epsilon})dV$.
So it is sufficient to prove that 
$g_\epsilon=\langle\grad \iota_\epsilon, \grad u_\epsilon\rangle$ weakly converges to $1$ in $L^1_{loc}$. 

First we consider the case where $D$ is the intersection of the futures of a finite number of lightlike
hyperplanes -- we will describe this situation by saying that $D$ is ``finite''.
Then the singularity is a finite tree and, in particular, it is closed. 
In this case, given a point  $x\in P_v$,  we denote by  $\ret_0(x)$  the starting point of the 
lightlike ray through $q(x)$.
Then the restriction of the map
 $\ret_0$ on each region $E$ of $\partial D\setminus T$ is a smooth projection and satisfies
\[
   \iota_v(x)=-\langle q(x)-\ret_0(x), v\rangle~.
\]
A simple computation shows that  $ \grad \iota_v = \grad u$ on $E$, which implies 
\begin{equation}\label{eq:xsr}
  \langle \grad\iota_v,\grad u\rangle=1~.
\end{equation}

On the other hand, if $\ret_0(x)$ lies in the interior of a segment $e$ of  $T$, it is not difficult to
check that $\grad\iota_\epsilon(x)\rightarrow\grad\iota(x)$. Indeed consider the domain
$\hat D$ defined as the future of the spacelike line which contains the segment $e$.
Note that $\hat D\supset D$, and $e$ is contained in the singularity of $\hat D$. 
Thus if $\hat\ret:\hat D\rightarrow \hat e$ denotes the retraction on the singularity,
then that $\hat\ret^{-1}(e)=\ret^{-1}(e)=U$, and the cosmological time of $D$ coincides
with the cosmological time of $\hat D$ on $U$.
Then the frequency function $\iota_v$ and $\iota_\epsilon$ around $q$ 
can be computed by considering  the domain $\hat D$ instead of the domain $D$. 
In that case an explicit computation shows that the convergence to (\ref{eq:xsr}) holds.

In particular,  the function $g_\epsilon$ converges to $1$ almost everywhere.
Since Lemma \ref{lm:g} shows that $g_\epsilon$ is bounded by $1$, the Dominated
Convergence Theorem implies that $g_{\epsilon}$ converges strongly to $1$ in $L^1_{loc}$.
So the proposition is valid whenever  $D$ is finite.

Consider now the general case.
By Lemma \ref{lm:g}, we have $g_\epsilon<1$ at every point.  So
for every sequence $\epsilon_n\rightarrow 0$, up to passing to a subsequence, 
 we can take the weak limit in $L^1_{loc}$. 
That is, there exists a function $g$ such that
\[
   \int fg dV=\lim_{n\rightarrow+\infty} \int fg_{\epsilon_n}dV~,
\]
and $\|g\|_{L^\infty}\leq 1$.
Note in particular that for every smooth function with compact support we obtain
\[
\int df\wedge\iota\omega=-\int fgdV
\]
Clearly, the same formula  holds for any function $f$ which is the limit in the Sobolev space
$W^{1,1}(K)$ of smooth functions with compact support.

To prove that $g=1$ we use an approximation argument.
Let $D_k$ be a sequence of finite domains of dependence converging to $D$, and 
let $u_k:P_v\rightarrow\mathbb R$ be the functions whose graph is $\partial D_k$.
Consider  a smooth function $\phi:[0,+\infty)\rightarrow[0,+\infty)$ which is
decreasing and such that its support is contained in  $[0,M]$, and define
\[
  f_k(x)=\phi(u_k(x))~,\qquad f(x)=\phi(u(x))~.
\]
Then the functions $f_k$ are $C^1$ with compact support and 
are  limits of smooth functions with compact support.
They satisfy
\[
   df_k=\phi'(u_k)du_k~,\qquad df=\phi'(u) du~.
\]
As  \eqref{eq:xsr} holds for finite domains we obtain for every $k$
$$ \int-\phi'(u_k)\iota_k dV=\int\phi(u_k)dV~, $$ 
and with the inequality  $\phi'\leq0$ we deduce 
\[
   \int |\phi'(u_k)|\iota_k dV=\int\phi(u_k)dV~.
\]
As  $u_k\rightarrow u$ uniformly on compact subsets and 
$\phi(u_k)$ is zero outside a compact region, which is independent of $k$, we obtain
\begin{equation}\label{eq:wco1}
  \int|\phi'(u_k)|\iota_kdV\rightarrow\int\phi(u)dV~.
\end{equation}
On the other hand,  Fatou's Lemma implies
\begin{equation}\label{eq:wco2}
   \lim\int|\phi'(u_k)|\iota_kdV\leq\int\limsup(|\phi'(u_k)|\iota_k)dV~.
\end{equation}
Note that by Proposition \ref{pr:usc} the right-hand side is less than
$\int |\phi'(u)|\iota dV=\int \phi(u)gdV$.
Comparing (\ref{eq:wco1}) and (\ref{eq:wco2}) yields
\[
  \int|\phi'(u)|dV\leq\int|\phi'(u)|gdV,
\]
and since $g(x)\leq 1$ almost everywhere, it follows that $g(x)=1$ 
almost everywhere.
\end{proof}

\begin{remark}
In other words,  Proposition \ref{mainflat:prop} states that $\iota_v\omega$ is a primitive 
of the volume form of $P_v$ in a distributional sense.
\end{remark}

\subsection{Stability of the frequency function for domains with flat boundary}
\label{ssc:stability}

We can state now the main stability theorem of this section.

\begin{theorem} \label{stability:thrm}
Let $D_k$ be a sequence of regular domains converging to $D$.
If the boundaries of $D_k$ and $D$ are all flat then $\iota_v^k\rightarrow \iota_v$
strongly in $L^1_{loc}(P_v)$.
\end{theorem}

\begin{proof}
Consider the convex functions $u_k:P_v\rightarrow\mathbb R$ 
whose graphs are identified with $\partial D_k$. Note that the sequence $u_k$ converges to
$u$ uniformly on compact subsets. Moreover, since $du_k$ converges to
$du$ almost everywhere and $\|du_k\|<1$, it follows from the Dominated Convergence
Theorem that $du_k\rightarrow du$ in $L^1(P_v)$, which implies
 $\omega_k\rightarrow\omega$ in $L^1(P_v)$.

We may assume that $u_k>0$ for every $k$. By this assumption,
$\iota^k_v(x)<u_k(x)$ for every $x$. So the functions $\iota^k_v$ are uniformly bounded on compact subsets.  Consequently,  there exists a weak limit $J$ of  the sequence $\iota^k_v$, that is  a $L^1$ function $J$ with
\begin{equation}\label{weak:eq}
   \int f\iota^k_v dV\rightarrow \int fJdV %% added J on right-hand side
\end{equation}
 for every compactly supported continuous function $f$.

Now we claim that for every compactly supported smooth function $f$ we have
\begin{equation}\label{dv:eq}
\int_{P_v}df\wedge(J\omega)=-\int_{P_v}fdV=\int_{P_v}df\wedge(\iota_v\omega)~.
\end{equation}
From this claim it then follows that $J\omega=\iota_v\omega$ almost everywhere and hence $J=\iota_v$.

To prove the claim, first note that the left-hand side of (\ref{dv:eq}) can be rewritten
as
\begin{equation}\label{weak2:eq}
\int_{P_v}df\wedge(J\omega)=\int_{P_v}(J-\iota^k_v)df\wedge\omega+\int_{P_v}\iota^k_v df\wedge(\omega-\omega_k)+
\int_{P_v}df\wedge(\iota^k_v\omega_k)~.
\end{equation}
By Proposition \ref{mainflat:prop},  the last term is equal to the  right-hand side of (\ref{dv:eq})
and hence independent of $k$.

As  $\omega_k$ converges to $\omega$ in $L^1_{loc}(P_v)$, the second term 
in (\ref{weak2:eq}) vanishes as $k\rightarrow+\infty$. (Note that the functions $\iota^k_vdf$ are uniformly
bounded in $L^\infty_{loc}$).

Finally by the density of continuous functions in $L^2_{loc}$,
$J$ is the weak limit of $\iota^k_v$ in $L^2_{loc}$, that is, (\ref{weak:eq}) holds for every
$L^2$-function defined on some open subset with compact closure.
In particular, the first term 
on the right %% added
in (\ref{weak2:eq}) also vanishes as $k\rightarrow+\infty$ (indeed
$df\wedge\omega=gdV$ for some compactly supported $L^\infty$-function $g$).
So letting $k$ go to $+\infty$ in (\ref{weak2:eq}),  we deduce that (\ref{dv:eq}) holds, and the claim is proved.

To show that $\iota_v$ is a strong limit of $\iota^k_v$ we use the following remark:
\[
\forall x\in \partial D:\; \iota_v(x)\geq\limsup \iota^k_v(x)~.
\] 
This holds due to the following: if $r_k$ is  the maximal lightlike 
ray    in  $\partial D_k$ that contains the point $x+u_k(x)v$, then the sequence $r_k$  converges to a lightlike ray
$r$ contained in $\partial D$, which contains the point $x+u(x)v$. 
So the same argument as in the proof of Proposition \ref{pr:usc} may be used.

Let us now fix an open subset $A\subset P_v$ with compact closure.
Then there is a constant  $M>0$ such that $\iota_v$ and $\iota_v^k$ are bounded by $M$ on $A$.
This implies that the functions $-(\iota_v^k)^2$ are uniformly bounded  and allows one to apply 
Fatou's Lemma. It follows that
\[
  \int_A-(\iota_v)^2\leq \int_A-\limsup (\iota^k_v)^2\leq\liminf\left(-\int_A(\iota^k_v)^2\right) 
\]
which implies 
\[
  \|\iota_v\|_{L^2(A)}\geq \limsup \|\iota^k_v\|_{L^2(A)}~.
\]
As the sequence $\iota^k_v$ converges weakly to $\iota_v$ in $L^2(A)$, this estimate implies 
the strong convergence in $L^2(A)$.
The $L^2$-convergence on a compact set then implies the $L^1$-convergence.
\end{proof}

\subsection{Uniform bounds for domains with non-flat boundaries}
\label{ssc:bounds}

We have seen in Section \ref{ssc:stability} that if $D$ is a domain of dependence with
flat boundary, and $(D_k)_{k\in \NN}$ is a sequence of domains of dependence with flat
boundaries converging to $D$, then the frequency function of  $D_k$ converges to the frequency function  of $D$.
Suppose now that $D$ is any domain of dependence (not necessarily with flat
boundary) and that $(D_k)_{k\in \NN}$ is a sequence of domains with flat
boundaries converging to $D$. We know (Proposition \ref{pr:usc}) that the frequency function of $D$
is at least the $\limsup$ of the  frequency functions of the domains $D_k$, but the example  in Section
\ref{ssc:example} shows that equality does not always hold. We will now see that the
opposite inequality does holds, albeit with a multiplicative constant.

This result is important  in view of the numerical computations in the following sections, which allow one to  visualize the 
frequency function measured by an observer in an MGHFC manifold of dimension 3+1. These computations
are done by approximating the corresponding domain of dependence by a sequence of finite
domains with  flat boundaries. However it might happen that the limit domain does not
have flat boundary. 
The computed frequency  function then coincides with the  limit of the  frequency functions
of the finite domains (with flat boundary),  and can differ from the actual frequency function 
of the limit domain. However,  Theorem \ref{tm:non-flat}  then ensures that the actual frequency function 
is at least equal to the computed limit frequency, and that it is at most three times this
computed frequency function. %(As mentioned below the statement of the theorem, we believe that it is actually at most twice the computed frequency, but have not proved it here.)

\begin{theorem} \label{tm:non-flat}
Let $D\subset \RR^{n,1}$ be a domain of dependence, and let $(D_k)_{k\in \NN}$ be a sequence of
domains of dependence with flat boundaries converging to $D$. Let $\iota_v$ be the frequency function 
of the boundary of $D$ with respect to a unit timelike direction $v$, considered as a
function on $P_v$, and let $\iota_v^k$ be the frequency function of  $D_k$. Then the sequence
$(\iota_v^k)_{k\in \NN}$ converges in $L^1_{loc}$ to a limit $\iota_{lim}$, and 
$$ \iota_{lim}\leq \iota_v\leq n\iota_{lim}~. $$
\end{theorem}

In Section \ref{ssc:coverings}, we will show that all  domains $D$ that arise as universal
covers of   2+1-dimensional   MGHFC spacetimes have flat boundaries.  Consequently, in that case  the inequalities in statement of
the theorem can be improved to an equality. In higher dimensions, it appears unlikely
that the boundary of the universal cover of a MGHFC manifold always has a flat boundary. 
It is conceivable that  for such domains, the inequality can be improved to
$ \iota_{lim}\leq \iota_v\leq (n-1)\iota_{lim}~, $
but we do not pursue this question further here.

\begin{proof}
We proceed as in the proof of Theorem \ref{stability:thrm} and only indicate the steps that differ from that proof.
As already noted, Lemma \ref{lm:g} still holds, but differences occur in the proofs of Proposition 
\ref{mainflat:prop} and of Theorem \ref{stability:thrm}. As  we want to obtain inequalities
on $\iota_v$ we consider a positive test function $f$.
Equation \eqref{eq:decomposition} still holds. However,  the inequality is weakened to
$$ 0\leq \int_{P_v} f\iota_v\Delta u \leq (n-1)\int_{P_v} f dV~.$$
The last inequality descends by the estimate (\ref{eq:gg}), when one takes the limit  $\epsilon\rightarrow 0$ and uses  that
 $\Delta u_\epsilon\rightarrow\Delta u$ as measure and that 
$\iota_\epsilon\searrow\iota_v$.
%because $0\leq \iota_v\Delta_u\leq n-1$ by Lemma \ref{lm:boundII}.
Following the proof of Proposition \ref{mainflat:prop}, we then obtain the inequality 
$$ -n\int_{P_v} f dV \leq \int_{P_v} df\wedge (\iota_v\omega)\leq -\int_{P_v} fdV~. $$
In the proof of Theorem \ref{stability:thrm}, Equation \eqref{dv:eq} is therefore replaced by 
$$ \int_{P_v} df\wedge (J\omega) = -\int_{P_v}fdV~,~~ 
-n\int_{P_v}fdV \leq \int_{P_v} df\wedge (\iota_v\omega) \leq -\int_{P_v}fdV~. $$
The rest of the proof of Theorem \ref{stability:thrm} goes through  and leads to the
statement.
\end{proof}

\subsection{Universal coverings of MGHFC spacetimes in dimension $2+1$}
\label{ssc:coverings}

Let $M$ be a 2+1-dimensional  MGHC flat spacetime of genus $g\geq 2$. 
We know that the universal covering of $M$ is a domain of dependence $D\subset \mathbb R^{2,1}$. We will assume  in the following that $M$ is not a Fuchsian spacetime. This means that
$D$ is not the future of a point, or, equivalently, that the holonomy  representation of $M$ is not conjugate to a linear representation in $SO(2,1)$.
In this subsection we will prove that the boundary of $D$ is flat.
%%In fact
This is a consequence of the following geometric property of the boundary of $D$.

\begin{prop} \label{pr:measure}
Identify the boundary $\partial \mathbb H^2$ with the set of lightlike directions in $\mathbb R^{2,1}$ and
let $D^*$ be the subset of $\partial \mathbb H^2$ consisting of  lightlike directions parallel to lightlike rays
contained in $\partial D$. Then  the set $D^*$ has Lebesgue measure zero in $\partial \mathbb H^2$.
\end{prop}

%%By Mess work, 
We know by the work of Mess \cite{mess} that the linear part of the holonomy representation of $M$ defines a Fuchsian group $\Gamma$, 
which determines a hyperbolic surface $S=\mathbb H^2/\Gamma$. Moreover, there
is a measured geodesic lamination $\lambda$ on $S$ such that $M$ is obtained by
a Lorentzian grafting on the Minkowski cone of $S$.
%%
%% Denote by $\tilde\lambda$ its lifting to the universal covering $\mathbb H^2$.
Denote by $\tilde\lambda$ the lifting of $\lambda$ to the universal covering $\mathbb H^2$.
We say that a point $\xi\in\partial\mathbb H^2$ is {\it nested} for the lamination $\lambda$ if 
%there is no leaf of $\lambda$ with end-point at $\xi$ and, 
for some point $v\in\mathbb H^2$, the intersection of the ray joining $v$ to $\xi$ with $\tilde\lambda$ is $+\infty$.

\begin{lemma}
If $\xi$ is a nested point for $\lambda$ then
\begin{itemize}
\item the point $\xi$ is not the end-point of any leaf of $\tilde\lambda$,
\item the intersection of any ray ending at $\xi$ with $\lambda$ is $+\infty$.
\end{itemize}
\end{lemma}

\begin{proof}
Let us consider the upper half-plane model of $\mathbb H^2$. Without loss
of generality we may assume that $\xi=\infty$.
Suppose there is a leaf $l$ of $\tilde\lambda$ ending at $\xi$, and take any compact ray
$r_0$ joining a point $v\in\mathbb H^2$ to $l$.
Now any sub-arc of the ray $[v,\xi)$ can be deformed through a family of transverse arcs to
a subarc of $r_0$. This implies that the intersection of any subarc of $[v,\xi)$ with $\tilde\lambda$
is uniformly bounded by the intersection of $r_0$ with $\tilde\lambda$.
This proves that $\xi$ is not nested.

For the second statement, consider a point  $v_0\in\mathbb H^2$ such that the intersection of $[v_0,\xi)$ with
$\tilde\lambda$ is $+\infty$.
Take a family of leaves $l_n$ meeting $[v_0,\xi)$ at a point 
$v_n\rightarrow\xi$. With the first statement, it is easy to check that $l_n$ bounds a neighborhood
$U_n$ of $\xi$, and that $\{U_n\}$ is a fundamental family of neighborhoods of $\xi$.

In particular there is a leaf, say $l_1$, cutting both
$[v_0,\xi)$ at a point $v_1$ and $[w,\xi)$ at a point $w_1$. Then 
every leaf of $\tilde\lambda$ cutting $[v_1,\xi)$ must cut also $[w_1,\xi)$
This implies that the intersection of $[w,\xi)$ with $\tilde\lambda$ is bigger than 
the intersection of $[v_1,\xi)$ with $\tilde\lambda$, which is clearly infinite.
\end{proof}

We will see that the set  $D^*$ does not contain any nested points, so the proof of the proposition  is obtained  from the following lemma.

\begin{lemma}
Almost all points in $\partial\mathbb H^2$ are nested for $\lambda$.
\end{lemma}

\begin{proof}
The proof is based on the ergodicity property of the geodesic flow on $S$.
For $(x,v)$ in the unit tangent bundle of $S$ let $r(x,v)$ the geodesic ray
$\{\exp_{x}(tv)|T\geq 0\}$. Consider now the following  subset of $T^1(S)$:
\[
B_n=\{(x,v)\in T^1(S)|\iota(r(x,v),\lambda)<n\}~.
\]
We claim that  $B_n$ is a set of measure zero for the Liouville measure.

Before proving the claim, let us show how the claim proves the statement.
Indeed we get that the measure of the set $B=\bigcup B_n$ is zero.
Let  $\tilde B\subset T^1\mathbb H^2$ be the lifting of $B$ 
on the universal covering. By definition we have that $\tilde B$ is
made of pairs $(x,v)$ such that the endpoint
of the ray $\exp_{x}(tv)$ is not nested.
In particular, if $E$ is the complement in  $\partial\mathbb H^2$ 
of nested points, the Liouville measure of $B$ can be computed as
as
\[
   \int_{K}\mu_x(E)dA~,
\]
where $K$ is a fundamental region and $\mu_x$ is the visual measure
from $x$. As the measure of $B$ is zero, it immediately follows that $E$ is
a set of measure zero.

It remains to prove the claim.
Let $\phi_t$ denote the geodesic flow on $T^1S$.
Clearly we have
\[
   \phi_t(B_n)\subset B_n.
\]
More precisely,  $t<s$ implies $\phi_t(B_n)\subset \phi_s(B_n)$.
It follows that $\hat B_n=\bigcup_{t>0}\phi_t(B)=\bigcup_{k\in\mathbb N}\phi_k(B_n)$ is a
subset invariant by the geodesic flow. Moreover its Liouville measure is equal to
\[
  \mu(\hat B_n)=\inf_{k}\mu(\phi_k(B_n))=\mu(B_n)
\]
where the last equality holds %%, since 
because $\mu$ is invariant by the geodesic flow.

By the ergodicity of the flow,  we have either  $\mu(B_n)=0$ or $\mu(T^1S\setminus B_n)=0$.
In order to prove that the latter is not true,  it is sufficient to prove that the complement of $B_n$
contains a non-empty open subset.

First note that if $(x,v)$ corresponds to a closed geodesic which intersects $\lambda$, then
the intersection of $\lambda$ with the ray $\exp_x(tv)$ is $+\infty$. In particular $(x,v)\notin B_n$.

Moving $x$ on the ray,  we may assume that it is not on $\lambda$.
Now take $M>0$ so that the intersection of $\lambda$ with the segment 
$r=\{\exp(tv)|t\in[0,M]\}$ is bigger than $2n$ and $\exp_x(Mv)$ is not on the lamination.

We want to show that a neighborhood of $(x,v)$ is contained in the complement of $B_n$.
Indeed if $(x_k,v_k)$ converges to $(x,v)$,  then the intersection 
of $\lambda$ with the segment $r_k=\{\exp_{x_k}(tv_k)|t\in[0,M]\}$ converges to
the intersection of $\lambda$ with $r$. So for $k$ sufficiently large, $(x_k, v_k)$ does not lie on 
$B_n$.
\end{proof}

In order to relate nested points with points in $D^*$ we need the following technical
lemma from Lorentzian geometry.

\begin{lemma}\label{lm:convray}
If $R$ is a lightlike ray contained in $\partial D$ %%$D$ 
which is maximal with respect to the inclusion,
then there is a sequence of points $r_n$ on the singularity $T$ which converges to
a point on $R$, and a sequence of spacelike support planes $P_n$ at $r_n$ which converges to
the lightlike plane containing $R$.
\end{lemma}

\begin{proof}
Let $v$ be any future oriented timelike vector.
Let $q$ be a point of $r$ and consider the segment of points $q_\epsilon=q+\epsilon v$ for 
$\epsilon\in[0,1]$ and
the path on $\Sigma$ given by $r_\epsilon=\ret(q_\epsilon)$.

Note that this path is contained in the closure of $D\cap I^-(q_1)$, which is a compact
region of $\mathbb R^{2,1}$. Thus there exists a sequence $\epsilon_n\rightarrow 0$ such that
$r_{\epsilon_n}$ converges to some point $\bar r$. 

We claim that $\bar r$ is contained in $R$.
In order to prove the claim,  note that sequence of vectors 
$q_{\epsilon_n}-r_{\epsilon_n}$ converges to $q-\bar r$.
Since they are timelike, their limit cannot be spacelike. 
%%Being $\partial D$ 
But $\partial D$ being achronal forces $q-\bar r$ to be lightlike and the segment
$[\bar r, q]$ to be %%is 
contained in $\partial D$.
As  the lightlike plane $P$ containing $R$ is a support plane for $D$, it follows that
$[\bar r, q]$ is contained in this plane, so in particular is on $R$.

To construct the sequence of lightlike support planes, it is sufficient to set
$P_n$ to be the plane orthogonal to $q_{\epsilon_n}-r_{\epsilon_n}$ passing through
$r_{\epsilon_n}$.
\end{proof}

We are now ready to prove Proposition \ref{pr:measure}.

\begin{proof}[Proof of Proposition \ref{pr:measure}]
We will prove that if $\xi\in D^*$ then $\xi$ is not nested.
Assume by contradiction that $\xi$ is nested, and let  $R$ be a ray parallel to $\xi$. By Lemma \ref{lm:convray},
there exists a sequence of points $r_n$ on the singularity, %%singularities 
converging to a point
on $R$ and a sequence of spacelike support planes $P_n$ converging to
the lightlike support plane containing $R$. Let $u_n$ be the unit timelike vector
orthogonal to $P_n$. Clearly we have that $u_n\rightarrow\xi$ in $\overline{\mathbb H^2}$.

By Mess' construction \cite{mess},  we have
\[
   r_n-r_0=\int_{c_n }w_n(x)d\mu_{\lambda},
\]
where $c_n$ is the segment joining $u_0$ to $u_n$ and $w_n(x)\in\mathbb R^{2,1}$
is defined to be $0$ if $x$ is not in the support of $\lambda$ and is the unit tangent vector at
$x$ orthogonal to the leaf through $x$ and pointing towards $u_n$ otherwise.

As by the hypothesis $\xi$ is nested, the ray $r$ joining $u_0$ to $\xi$ transversely  
meets the lamination.
In particular, by changing $u_0$ to %%with 
a point on $r\cap \lambda$ 
we may assume that $u_0$ is on the lamination.
Let $e$ be the unit vector at $u_0$ orthogonal to the leaf  $l_0$ through $u_0$ and 
pointing towards $\xi$.
We claim that if $\xi$ is nested then 
\[
\langle r_n-r_0, e\rangle\rightarrow+\infty,
\]
which contradicts the assumption that  the sequence $r_n$ converges in Minkowski space.

First note that since $u_n\rightarrow\xi$, we may assume that $u_n$ is 
on the half-plane bounded by $l_0$ and containing $\xi$.

Thus if $l$ is a leaf that intersects $c(u_0, u_n)$, $l$ disconnects $l_0$ from 
$\xi$, and the scalar product of vectors $e$ and $w(x)$ is positive. Since the corresponding
geodesics are disjoint, the reverse of Schwarz %%Schwartz 
inequality holds, that is,
$\langle w(x), e \rangle>1$.
This implies
\[
\langle r_n-r_0, e\rangle\geq\iota(\tilde\lambda, c_n)~.
\]

Let us prove that $\iota(\tilde\lambda, c_n)\rightarrow+\infty$. 
The reason is that for every leaf $l$ of $\tilde\lambda$
cutting the ray $[u_0,\xi)$,   
$u_n$ is definitively contained in the region bounded by $l$ containing $\xi$. So
for every point $x$ on the segment $[u_0,\xi)$, for $n$ sufficiently large, we have
$\iota(\nu, c_n)\geq\iota(\nu, [u_0,x])$.

Since we are assuming that $\nu([u_0, \xi))=+\infty$, we can choose $x$ so that
$\iota(\nu, [u_0,x])$ is arbitrarily big, so the conclusion follows.
\end{proof}

Let us fix a unit timelike vector $v$, and let $u:P_v\rightarrow D$ be
the convex function whose graph is the boundary of $D$. Note that if
$e$ is a unit vector in $P_v$, then $e+v$ is a lightlike
vector. In this way,  the unit circle $S^1$ in $P_v$ is identified to
$\partial H^2$ by the map $e\mapsto [e+v]$.  Under this
identification, the subset $D^*$ corresponds to the image of the map
  \[
    \delta: P_v\setminus T \rightarrow S^1~, ~~\delta(x)=\grad(u)(x)~.
  \]

Fix a unit vector $e$ in $P_v$ and
take linear orthogonal coordinates $(x,y)$ on $P_v$ such that $\partial_x=e$,
and consider the restriction of the function $u$ on each line parallel to $e$.
That is, for $y\in\mathbb R$, let $u_y:\mathbb R\rightarrow\mathbb R$ be defined
as $u_y(x)=u(x,y)$.

Note that whenever $(x,y)$ does not correspond to a point on $T$, then $u_y$ is
differentiable at $x$ and 
\[
   (u_y)'(x)=\langle \grad u(x,y), e\rangle
\]
By Proposition \ref{pr:measure}, at those points, the derivative takes value in a subset of zero measure of $\mathbb R$.

Now to prove that the boundary of $D$ is flat we will proceed in three steps.
\begin{itemize}
\item[Step 1.] We will prove that for a generic choice of the vector $e$, 
if $u_y$ is differentiable at $x$, then  $(x,y)$ does not correspond to a point on $T$.
In particular, the derivative $(u_y)'$ takes value in a subset of zero measure of $\mathbb R$.
\item [Step 2.] We will use this fact to show that for every $y$, the measure $(u_y)''$ is atomic
with support on $T\cap R_y$. %% changed this sentence
\item [Step 3.] 
Using a disintegration formula for %%of  
$\partial_{xx}u$ in terms of the family of measures $(u_y)''$
we conclude that this measure $\partial_{xx}u$ is concentrated on $T$.
\end{itemize}

\begin{lemma}\label{lm:step1}
There is a subset $A$ of $S^1$ such that:
\begin{itemize}
\item If $e\in A$ then, for every $y\in\mathbb R$, the points 
%%where $u_y$
%%is differentiable corresponds exactly to points where $u$
%%if the function $u_y$ is differentiable at $x$ then $u$ is differentiable at $(x,y)$ and we have
$x$ where $u_y$ is differentiable are exactly the points such that $u$ is differentiable at $(x,y)$.
Moreover, at those points,
\[
    (u_y)'=\langle \grad u(x,y), e\rangle~.
\]
%% NB I have changed the statement of this lemma so please check that this is what you had in mind.
\item The measure of %%od 
$S^1\setminus A$ is zero.
\end{itemize}
\end{lemma}

\begin{proof}
Let $A$ be the set formed by vectors $e$ such that the geodesic in $\mathbb H^2$ 
starting %%strating 
from $v$ with direction $e$ does not meet any leaf of $\lambda$ orthogonally.
We will prove that $A$ fulfills the requirements %%requirents 
of the statement.

First let us prove that the only differentiable points of $u_y$ correspond %%corresponds 
to differentiable points
of $u$. By contradiction suppose that $u_y$ is differentiable at %%to 
a point $x$ so that $(x,y)$ corresponds to a point on the singularity. %%tree. 
Up to translation we may suppose that $x=y=u(x,y)=0$. 

As  the point $0$ is on the singularity, %%tree, 
there are two lightlike planes $P_1, P_2$ through $0$, which are support planes
for $D$. Let $V$ be the vertical plane containing $e$ and $v$. Note that $P_1\cap V$ and $P_2\cap V$ are support lines
for $\partial D\cap V$ at the point $0$. As $\partial D\cap V$ corresponds to the graph of  $u_y$, and we are %% added
assuming that
$u_y$ is differentiable at $x=0$, those support lines must coincide, $P_1\cap V=P_2\cap V$. This implies that $V$ must contain the line
$r=P_1\cap P_2$. Note, however,  that this line is spacelike, and its dual geodesic in $\mathbb H^2$ is a leaf of $l\in\lambda$.
On the other hand $V\cap\mathbb H^2$ is the geodesic $g$  starting from $v$ with direction $e$, so the condition implies that 
$g$ meets orthogonally $l$, contradicting the choice of $e$. %% $E$.

It remains to  show that the complement of $A$ in $S^1$ is a set of measure zero.
Note that if $e\in A$ then $-e$ is also in $A$, so we may regard $A$ as a subset of
the projective line $P(P_v)$.

We will argue as follows. For any geodesic $l$ of $\mathbb H^2$, let $e(l)$ be the unit tangent vector
at $v$ such that the geodesic $\exp_{v}(te(l))$ hits orthogonally $l$. Note that $e(l)$ is defined up to the sign,
so it should be considered more properly as an element of $P(P_v)$.
The complement of $A$ is the set of unit vectors $e(l)$ where $l$ is a leaf of the lamination $\lambda$.

Let us now  fix any ray $r$ starting %%strating 
from $v$ and %% denote by $E_r=\{e(l)|\,l\textrm{ is a leaf of }\lambda\textrm{ hitting r }\}$.
define 
$$ E_r=\{e(l)|\,l\textrm{ is a leaf of }\lambda\textrm{ hitting r }\}~. $$
Note that if $r_n$ is a dense subset of the ray from $v$ we clearly have
\[
   \bigcup_{n} E_{r_n}=P(P_v)\setminus A~.
\]
So in order to argue that the measure of the complement of $A$ is zero, it is sufficient to show that $E_r$ has measure zero.

Now on $r\cap\lambda$ we may define a vector field $w$ by taking for $w(x)$ to be the unique vector orthogonal to the leaf $l$ through
$x$. Note that $e(l)$ coincides with the orthogonal projection of $w(x)$ on $P_v$ up to renormalization.
By a classical result \cite{epstein-marden}, the field $w$ can then be extended to a Lipschitz vector field, still denoted $w$, on the whole line.
In particular,  we obtain a map
\[
   \hat e:r\rightarrow P(P_v)
\]
by defining $\hat e(x)$ to be the projective class of $w(x)$. It is not difficult to show that this map is locally Lipschitz and, 
by definition, $E_{r}=\hat e(r\cap\lambda)$. As the measure of $r\cap\lambda$ is zero, this concludes the proof.
\end{proof}

Lemma \ref{lm:step1} concludes the proof of step 1.
In particular,  note that if $e$ is on the set $A$, then for every $y\in\mathbb R$,
the derivative function $u_y$ takes value on the set
\[
   \{\langle \grad u(x,y), e\rangle\}
\]
which by Proposition \ref{pr:measure} has measure zero.
The proof of step 2 is then based on the following simple lemma on convex functions.

\begin{lemma} \label{lm:step2}
Let $u:\mathbb R\rightarrow\mathbb R$ be a convex function. Suppose that 
the measure of the image $u':\mathbb R\rightarrow\mathbb R$ is zero. Then 
$u''$ is an atomic measure, and its support coincides with the set of discontinuity of $u'$.
\end{lemma}

\begin{proof}
By the %% added
standard theory of convex functions the measure $u''$ can be split as the %%in a
sum of a measure $\mu$ without atoms and an atomic part, say $\nu$, with
$\nu=\sum_k a_k\delta_{q_k}$, where $\delta_{q_k}$ is the Dirac measure concentrated
on $q_k$ and $\sum a_k$ is an absolutely convergent series.

Now we claim that the measure of  the image  of $u'([a,b])$ is equal to $\mu([a,b])$.
Indeed, note that $u$ is not differentiable exactly on the points $\{q_n|n\in\mathbb N\}$.
Moreover at every point there exists the left derivative and the right derivative that can be expressed as follows. 
Assume that $0$ is a differentiable point of $u$, then 
\[
   u'_l(x)=u'(0)+\phi(x)+\sum_{q_n\in[0,x)}a_n~,~~
   u'_r(x)=u'(0)+\phi(x)+\sum_{q_n\in[0,x]}a_n~,
\]
where we  put $\phi(x)=\mu([0,x])$.

If $u$ is differentiable at some point $x$, then the values
that $u'$ takes on the interval $[0,x]$ can be %% then 
described as
\[
  u'([0,x])= [u'(0), u'(x)]\setminus \bigcup_{q_n\in[0,x]}I_n
\]
where $I_n=[u'_l(q_n), u'r(q_n)]$ is the interval between the left and right derivative
at $q_n$, which are pairwise disjoint.

So the measure of this set $u'([0,x])$ is given by
\[
   \phi(x)\sum_{q_n\in[0,x]}a_n-\sum_{q_n\in[0,x]}a_n=\phi(x)~.
\]

By the assumption, $\phi(x)=0$ for any $x$, so $\mu=0$.
\end{proof}

Finally, in order to prove step 3 we need the following disintegration result
of the measure $\partial^2_{xx}u$ in terms of the measure $(u_y)''$.

\begin{lemma}\label{lm:step3}
Let $x,y$ be coordinates on $\mathbb R^2$ and consider a convex function $u$. For every $y\in\mathbb R^2$ denote by $u_y$
the convex function  $x\mapsto u(x,y)$. 
If $\partial_{xx}u$  %% $\partial_xxu$ 
denotes the second derivative of $u$ along the $x$ axis
(which is a Radon measure on $\mathbb R^2$) and $(u_y)''$ denotes the second
derivative of $u_y$ (which is a Radon measure on $\mathbb R$) 
then the following formula holds
\[
\int_{\mathbb R^2} f(x,y)\partial_{xx}u=
\int_{\mathbb R} dy\int_{\mathbb R}f(x,y)(u_y)''
\]
for every bounded Borel  function $f$ with compact support. %% supported on a some compact subset of the plane.
\end{lemma}

The proof of this analytical Lemma can be found in \cite{savare} (Theorem 1.3 formula (1.31))
for the wider class of bounded Hessian functions. 
We are ready now to prove that the boundary of $D$ is flat.

\begin{prop} \label{pr:2+1flat}
If $D$ is the universal covering of a MGHC flat spacetime of dimension $2+1$, 
then  its boundary is flat.
\end{prop}

\begin{proof}
We will prove that $\hess(u)$, considered as a matrix-valued %% that is a matrix valued 
measure on $P_v$, is supported on $T$.

Indeed it is sufficient to prove that there are three independent directions $e_1, e_2, e_3$ such that
$D^2_{e_i, e_i}(u)$ is a measure supported on $T$.
As the subset $A$ from Lemma \ref{lm:step1} is dense, it is sufficient to prove that
$D^2_{e,e}(u)$ is supported on $T$ for $e\in A$. %%$e\in T$.

If $x,y$ are the Cartesian coordinates on $P_v$ such that $e=\partial_x$, then
$D^2_{e,e}(u)$ coincides with $\partial^2_{xx}(u)$. So we have to prove that if
$f$ is a measurable function which is zero on $T$, then
\[
   \int f\partial^2_{xx}(u)=0~.
\]
We may compute the integral above using Lemma \ref{lm:step3}, which yields
\[
  \int f\partial^2_{xx}=\int dy\int (f_y)(u_y)''~.
\]
Now, by Lemma  \ref{lm:step2} and Lemma \ref{lm:step1}, $(u_y)''$ is supported on
the discontinuity of $u'_y$ which corresponds to points $x$ such that $(x,y)$ in on the projection of the singularity. %%tree.
It follows that $f_y$ is zero on the support of $(u_y)''$, and hence the integral is zero.
\end{proof}

%:bms4.tex

\section{The frequency function for spacetimes constructed from measured geodesic laminations}

The aim  of this section is to understand to what extent and under what conditions an observer 
in a $2+1$-dimensional domain of dependence
can reconstruct the geometry and topology of the ambient space from his observation --- either
at one time or over a fixed time interval --- from the frequency function of the signal emitted by the initial singularity. In particular, we investigate this question for 
 a domain of dependence, which is the universal cover of a (non-Fuchsian) MGHFC 
spacetime $M$ containing a closed Cauchy surface of genus $g$.

A basic remark, which somewhat complicates the statements and the analysis below, is that the observer
can only ``see'' the universal cover of $M$, so he can in no way distinguish $M$ from any of its finite
covers. In other words, the observer can only determine the largest discrete subgroup of $\Iso(\R^{2,1})$
compatible with the signal she observes. Moreover, he can  only be certain to have determined
correctly the fundamental group of $M$ if he knows the genus of $S$, since otherwise it remains possible
that his spacetime is topologically a finite cover of $M$, with a flat metric which is ``almost'' 
lifted from a flat metric on $M$, with only a small change in a region not visible by her. 

In Section \ref{ssc:cosmological} we study the relationship between the frequency function measured
by an observer and her cosmological time (see Proposition \ref{estimate}). Then in Section \ref{ssc:reconstructing},
we show (see Proposition \ref{pr:construct}) that an observer in the universal cover of a non-Fuchsian MGHFC Minkowski spacetime
can reconstruct in finite eigentime the geometry and topology of the space, if the genus of the Cauchy surface
is  known to him. In Section \ref{ssc:higherdim}, we briefly explain how those arguments can be adapted to
higher dimensions.

%\subsection{Behavior of the frequency function}
%\begin{prop}
%The segments of the circle that correspond to vertices of the tree form a Cantor set.
%\end{prop}
%\begin{proof}
%\end{proof}

\subsection{Estimating the cosmological time from the frequency function}
\label{ssc:cosmological}

We consider a domain $D$ and an observer in $D$ given by a point  $p\in D$ and a 
future directed timelike unit vector $v\in \mathbb H^2$. 
We consider the associated rescaled frequency function
$$
\rho_{p,v}^{\mathcal D}: S^1\rightarrow \mathbb R_0^+
$$
defined as in the previous section. 
We will be mostly interested in the case where $D$ is the universal cover of 
a spacetime constructed by grafting along a measured geodesic lamination.
In this case, the observer will see a division of the circle into intervals, 
on which the frequency function behaves like that of a spacelike line, 
and intervals in which the frequency function behaves like the one of a light cone. 
The former correspond to the edges of the singular tree of $D$, the latter to its vertices. 
It follows from the results in the previous sections that the  frequency function is analytic 
on the segments of the circle that correspond to the edges, 
while it is generally not analytic on the segments that correspond to the vertices.
In general, the segments of the circle that 
correspond to vertices of the singular tree form a Cantor set.
%% note that I have removed Prop 5.1 and replaced it simply by this sentence, without a proof,
%% I don't think we use this later on so it's more a remark for the interested reader.
%% I'm not sure exactly what the form of this Cantor set is, esp. in the irrational case.

We define the {\em maximum frequency} as
$$
\rho_{p,v}^{D,max}=\sup_{u\in S^1} \rho_{p,v}^{D}(u)
$$
To understand its properties, we consider again %% added
our two main examples.

\begin{example}\label{lcone}
Consider a cone  $D=I^+(q)$ and an observer characterized by $p\in I^+(q)$, $v\in\mathbb H^2$.
Then
$$
\rho^{D,max}_{p,v}=\frac {T(p)}  2 e^{\delta},
$$
where $T(p)=|p-q|$ is the cosmological time of the observer and $\delta=d_{\mathbb H^2}(v,\text{grad}_p\,T)$ is %% added "is"
the hyperbolic distance between $v$ and $\text{grad}_p\, T$. This follows from Equation \eqref{eq:rho1}.
\end{example}

\begin{example}\label{splikeline}
Consider the future $D=I^+(l)$ of a spacelike line $l$ and an observer with $p\in I^+(l)$, $v\in\mathbb H^2$. 
It follows from Equation \eqref{eq:rho2} that 
$$
\rho^{D,max}_{p,v}=\frac{T(p)} 2  e^{\delta}\cosh\xi ,
$$
where $T(p)$ is the cosmological time of $p$ and $\delta,\xi$ are defined as follows. Denote by $\tilde l$ the geodesic in $\mathbb H$ that is stabilized by the $PSL(2,\mathbb R)$ element that fixes the direction of $l$. Then $\text{grad}_p\; T$ defined a point on $\tilde l$,  $\delta$ is the hyperbolic distance of $v$ and $\tilde l$ and $\xi$ the distance of the projection of $v$ on $\tilde l$ from the point in $\mathbb H^2$ defined by $\text{grad}_p\; T$.
\end{example}

We will now determine an estimate on the cosmological time. The central ingredient is the following proposition. %%:

\begin{prop} \label{maxint}
The maximum frequency is given by
$$\rho^{D,max}_{p,v}=\sup\{\langle v, y-x\rangle\;:\; x,y\in I^-(p)\cap D, y-x\;\text{future directed and lightlike}\}~.$$
\end{prop}

\begin{proof}
By definition, we have
$$
\rho^{D,max}_{p,v}\leq \sup\{\langle v, y-x\rangle\;:\; x,y\in I^-(p)\cap D, y-x\;\text{future directed and lightlike}\}~.
$$
To show the opposite inequality, %%other direction, 
we choose sequences $(x_n)_{n\in\mathbb N}$, $(y_n)_{n\in\mathbb N}$ with $x_n,y_n\in I^-(p)\cap D$ that satisfy
\begin{align}\label{supcond}
\lim_{n\to\infty} \langle v, y_n-x_n\rangle=\sup\{\langle v, y-x\rangle\;:\; x,y\in I^-(p)\cap D, y-x\;\text{future directed and lightlike}\}~.
\end{align}
As the intersection $I^-(p)\cap \bar D$ is compact, there exist convergent subsequences 
$(x_{n_k})_{k\in\mathbb N}$, $(y_{n_k})_{k\in\mathbb N}$ with limits  
$x_{n_k}\rightarrow \bar x\in I^-(p)\cap \bar D$, $y_{n_k}\rightarrow\bar y\in I^-(p)\cap\bar D$. 

If the segment $[\bar x,\bar y]$ is extensible, i.~e.~ if there exist  $\bar x',\bar y'\in I^-(p)\cap \bar D$ with $[\bar x,\bar y]\subset [\bar x',\bar y']$, then we obtain a contradiction to \eqref{supcond}. Therefore $[\bar x,\bar y]$ is inextensible,  $\bar x$ lies on the tree and $\bar y\in \partial D\cap I^-(p)$. 
This implies that there exists a $\theta\in S^1$ such that the past directed lightlike ray starting at $p$ that is defined by $\theta$ intersects $\partial D$ in $\bar y$ and
$\langle v, \bar y-\bar x\rangle=\rho_{p,v}^D(\theta)$.
\end{proof}

%% added following
An immediate consequence is that if an  observer moves along a timelike geodesic then the maximum frequency
is increasing with time.
More generally, 
%% end added part
Proposition \ref{maxint} allows us to give estimates for the maximum frequency of domains that are contained in each other.

\begin{cor}\label{domrel}
Let $D,D'$ be domains with  $p\in D\subset D'$. Then for all $v\in\mathbb H^2$:
$$\rho_{p,v}^{D,max}\leq \rho_{p,v}^{D',max}~.$$ %%.
\end{cor}

In particular, we can estimate the maximum frequency  for any domain. %%by

\begin{cor}\label{intestimate}
For a domain $D$ and any observer characterized by $p\in D$ and $v\in\mathbb H^2$, the following inequalities hold:
\begin{align}
&\rho_{p,v}^{D,max}\geq \sup_{q\in I^-(p)\cap D} \rho^{I^+(q),max}_{p,v}~, %%\qquad
&\rho_{p,v}^{D,max}\leq \inf_{\substack{  {l\;\text{spacelike line} }\\{D\subset I^+(l)}}} \rho^{I^+(l),max}_{p,v}~.
\end{align}
\end{cor}

By applying this corollary to a general domain and using the results of Examples \ref{lcone} and \ref{splikeline} we obtain
the following statement.

\begin{prop} \label{estimate}
Let $D$ be a domain with an observer characterized by $p\in D$ and $v\in\mathbb H^2$. 
Then the following inequalities relate the maximum frequency and the  cosmological time:
\begin{align}
\frac{T(p)} 2\leq \rho^{D,max}_{p,v}\leq T(p)\cosh d_{\mathbb H}(v,\text{grad}_p\; T)~.
\end{align}
\end{prop}

\begin{proof}
From the first inequality in Corollary \ref{intestimate} and Example \ref{lcone} we have
$$
\rho^{D,max}_{p,v}\geq \sup_{q\in I^-(p)\cap D} \rho^{I^+(q),max}_{p,v}=
\frac 1 2 \sup_{q\in I^-(p)\cap D} |p-q| e^{d_{\mathbb H}\left(v, \frac{p-q}{|p-q|}\right)}\geq \frac 1 2 \sup_{q\in I^-(p)\cap D} |p-q|~.
$$
By definition of the cosmological time, $T(p)=\sup_{q\in I^-(p)\cap D} |p-q|$, which proves the first inequality.

To prove the second inequality, 
we use the fact that for any lightlike vector %% added
$\xi$ and any two  future-directed timelike unit vectors $u,v$, we have
\begin{align}\label{helpid}
|\langle v,\xi\rangle|\leq 2 |\langle u, v\rangle| |\langle u,\xi\rangle|~. 
\end{align}
This can be seen as follows: after applying suitable elements of $SO(2,1)^+$, we can suppose that the vectors $\xi, u,v$ are given by
$$
u=\left(\begin{array}{c} 1\\ 0\\ 0\end{array}\right)\qquad v=
\left(\begin{array}{c} \cosh\alpha\\ \sinh\alpha\\ 0\end{array}\right)\qquad 
\xi=\left(\begin{array}{c} a\\ b\\ c\end{array}\right)\quad\text{with}\quad a^2=b^2+c^2~.
$$
This yields 
$$ |\langle u,\xi\rangle|=|a|~, 
\qquad |\langle v, \xi\rangle|=|a\cosh\alpha-b\sinh\alpha|\leq 2|a|\cosh\alpha~, 
\qquad |\langle u,v\rangle|=\cosh\alpha~, $$
and proves \eqref{helpid}.
By combining \eqref{helpid} with Proposition \ref{maxint}, we obtain for all $p\in D$ and  $u,v\in\mathbb H^2$
$$
\rho^{D,max}_{p,v}\leq 2|\langle u, v\rangle|\rho^{D,max}_{p,u}~.
$$
For $u=\text{grad}_p\; T$ this yields
$$
\rho^{D,max}_{p,v}\leq 2\cosh d_{\mathbb H}(v,\text{grad}_p(T))\rho^{D,max}_{p,\text{grad}_p\; T}~.
$$
For the future of a spacelike line, we have from Example \ref{splikeline}
$$
\rho^{I^+(l), max}_{p,\text{grad}_p\; T}=\frac {T(p)} 2,
$$
because the parameters $\delta,\xi$ in Example \ref{splikeline} vanish. This implies
together with Corollary \ref{intestimate} that 
$$
\rho_{p,v}^{D,max}\leq 2\cosh d_{\mathbb H}(v,\text{grad}_p(T)) \inf_{\substack{  {l\;\text{spacelike line} }\\{D\subset I^+(l)}}} \rho^{I^+(l),max}_{p,\text{grad}_p T}.
$$
By  definition of the domain $D$ there exists a point $q$ in the tree with $T(p)=|p-q|$ and 
two lightlike support planes that contain $q$. 
Let $\tilde l$ be the line obtained by intersecting these support planes. 
Then the cosmological time $\tilde T(p)$ of $p$ with respect to $I^+(\tilde l)$ 
and its gradient $\text{grad}_p\;\tilde T$ at $p$  coincides with the cosmological time $T(p)$ 
with respect to  $D$ and its gradient $\text{grad}_p\; T$. This implies 
$$
\inf_{\substack{  {l\;\text{spacelike line} }\\{D\subset I^+(l)}}} \rho^{I^+(l),max}_{p,\text{grad}_p T} =\rho^{I^+(\tilde l), max}_{p,\text{grad}_p (T)}=\frac {T(p)} 2
$$
and proves the claim.
\end{proof}

%\subsection{Extracting the information about the spacetime}
%{\bf to be edited}

\subsection{Reconstructing the holonomy from the frequency function}
\label{ssc:reconstructing}

From the frequency function on the circular segments, 
the observer can reconstruct the relevant data (position of the edges and vertices, his geodesic distance from the edge segments and vertices), 
but only for {\em  the pieces of the singular %% added
tree he sees}.  
If the observer is very close to the singularity, he will only see a single edge of the tree  
and the picture will  look like the one for a line. 
With time, he moves away from the tree and more and more intervals 
corresponding to the edges and vertices of the tree will appear. 
In the limit where his eigentime and his cosmological time go to infinity, he will see the image of the whole tree. 

This implies that the observer can reconstruct the domain (up to a global Poincar\'e transformation) 
from his measurements if he waits infinitely long. 
From his observations, he can reconstruct the edges of the singular tree, 
and --- if the spacetime is obtained by Lorentzian
grafting on a closed hyperbolic surface --- %% changed sentence
the action of the fundamental group on the tree. 
This amounts to recovering the underlying measured geodesic lamination. 

If the spacetime corresponds to a grafted genus $g$ surface {\em and} the observer knows the associated Fuchsian group 
(i.~e.~the linear part of the holonomy), he can construct the {\em complete domain} in finite eigentime. We will consider
below to what extend the observer can reconstruct the geometry and topology of the spacetime without knowing the linear
part of the holonomy.

%{\bf estimate of the time it takes}
%\subsection{Reconstructing the holonomy from the frequency function}

We now concentrate on the case, where
$M$ is a maximal flat globally hyperbolic space-time with closed Cauchy surface $S$
of genus $g\geq 2$. The universal covering of $M$ is then isometric to a regular domain $D$.
More precisely, there is a subgroup $G$ of %%into 
$Isom(\mathbb R^{2,1})$ such that
$D$ is invariant under the action of $G$ and $M=D/G$.
Let $\Gamma$ be the subgroup of $SO^+(2,1)$ consisting of the $SO^+(2,1)$ components
of  elements of $G$.
It is known (see \cite{mess}) that $\Gamma$ is a discrete subgroup of $SO^+(2,1)$ and that
$\mathbb H^2/\Gamma$ is a surface diffeomorphic to $S$. Moreover the measured geodesic
lamination $\lambdat$ dual to the initial singularity of $D$ %% added
 is invariant under the action of $\Gamma$ and induces
a measured geodesic lamination $\lambda$ on $\mathbb H^2/\Gamma$.
%% I changed the notations and used \tilde\lambda for the lamination in H^2 and \lambda for the quotient.

The main result we present in this section (Proposition \ref{pr:construct})
states that if $M$ is not a conformally static space-time 
(which would correspond to the empty lamination)
and if its initial singularity is on a simplicial tree, then an observer can construct in 
finite time a finite set of elements of $SO^+(2,1)$ which generates
a finite extension  of $\Gamma$. In other words, we will prove the result only when
the lamination $\lambda$ is rational (that is, its support is a disjoint union of closed
curves). We believe that the result could hold also for a 
general lamination. However, in that case some technical issues arise  which
make the analysis more complex, and we prefer to focus on the simpler case where 
$\lambda$ is rational.

The central  idea is to consider  the isotropy group $\Gamma_0$ of $\lambdat$
\[
  \Gamma_0=\{\gamma\in SO^+(2,1)|\gamma(\lambdat)=\lambdat\}~.
\]
It is clear that $\Gamma_0$ is a discrete subgroup of $SO^+(2,1)$ containing $\Gamma$.
The quotient $\mathbb H^2/\Gamma_0$ is a surface, possibly with  singular points, which
arise from points in $\mathbb H^2$ that are fixed by an element   of $\Gamma_0$.
There is a natural projection map
\[
   \pi:\mathbb H^2/\Gamma\rightarrow\mathbb H^2/\Gamma_0
\]
which %% that 
is  a  finite covering. In particular, the index of $\Gamma$ in $\Gamma_0$
is equal to the cardinality of the fibers of $\pi$ and, consequently, is finite.

\begin{remark}
Any element of $\Gamma_0$ is the linear part of an affine transformation
that preserves the regular domain $D$. 
So elements of $\Gamma_0$ are the linear parts of the elements in the isotropy
group $G_0$ of $D$
\[
  G_0=\{g\in Isom(\mathbb R^{2,1)}~|~g(D)=D\}\,.
\]
It should be noted that in principle there are many subgroups $G'$
of $G_0$ (of finite index) such that $D/G'$ is a MGH spacetime with compact Cauchy surface.
Clearly the frequency function measured by an observer in $M$ is equal to the frequency function
of some observer in such spacetimes. This suggests that an observer
cannot precisely determine the group $G$ (or $\Gamma$), but only
the group $\Gamma_0$.
It should also be  noted that in the generic case, $\Gamma=\Gamma_0$
and it does not contain proper cocompact subgroups.
\end{remark}

We say that a leaf $l$ of $\lambdat$ is {\it seen} by an observer $(p,v)$
if the intersection of $D$ with the support plane orthogonal to some point on $l$
intersects $I^-(p)$. Note that if $x,y\in l$,  then the intersection of $D$ with
the support plane orthogonal to $x$ is equal to the intersection of $D$ with
the support plane orthogonal to $y$.

In the following proposition (and therefore in the final result of this section)
we restrict  attention to a lamination $\lambdat$ with a simplicial dual tree,
although it appears quite likely that the proposition holds for  general laminations.

\begin{prop} \label{pr:rational}
Suppose that the dual tree of the lamination $\lambdat$ is simplicial. 
Then the frequency function of an  observer $(p,v)$
allows one to reconstruct the sublamination $\lambdat_{(p,v)}$ consisting of the
leaves of $\lambdat$ seen by $(p,v)$
\end{prop}

%[I think that we can prove this proposition when $\lambda$ is a simplicial lamination.
%This proposition should be probably moved to the previous section]

\begin{proof} %% added this proof -- Francesco, you might want to check that it's sufficiently complete.
As mentioned above, the frequency function seen by the observer can be split into frequency functions of  different regions which correspond,
respectively, to the edges and to the vertices of the singular tree. 
In the regions corresponding to the edges, the frequency function is analytic and behaves as in Example 2 in Section 
\ref{ssc:ex1}. It is shown there  that knowing the frequency function on an open subset of $S^1$ is sufficient
to determine the positions of the edges, and therefore the leaves of the lamination $\lambdat_{(p,v)}$.
\end{proof}

Let us fix a point $x_0\in\mathbb H^2$, and denote by $B_d$ the ball in $\mathbb H^2$
centered at $x_0$ with radius $d$. For simplicity,  suppose that the point
$x_0$ does not lie in a leaf of $\lambdat$.
We denote by $\lambdat_d$ the sublamination of $\lambdat$ made of leaves
that intersects $B_d$:
\[
    \lambdat_d=
    \bigcup_{\begin{array}{l}l\textrm{ leaf of }\lambdat\\l\cap B_d\neq\emptyset\end{array}}l\,.
\]

\begin{lemma}\label{td:lm}
For any $d>0$ there is a time  $T$ such that for $t\geq T$ 
the observer $(p+tv, v)$ sees all the leaves in $\lambdat_d$,
or, equivalently, $\lambdat_d\subset \lambdat_{p+tv,v}$
\end{lemma}

\begin{proof}
There is a compact subset $K$ of $\partial D$ such that if $x\in B_d$ then 
the support plane orthogonal to $x$ intersects $\partial D$ in $K$.
Since $I^-(p+tv)\cap D$ is an increasing sequence of open subsets that cover $D$,
there is a constant $T$ such that $I^-(p+tv)$ contains $K$ for $t\geq T$.
By definition, we  then have $\lambdat_{p+tv}\subset\lambdat_d$,
and the conclusion follows.
\end{proof}

We now  consider the elements of $SO^+(2,1)$ that send
leaves of $\lambdat_d$ either out of $B_d$ or to other leaves of $\lambdat_d$:
\[
   \Gamma_d=\{\gamma\in SO^+(2,1)~|~\gamma(\lambdat_d)\cap B_d\subset\lambdat_d\}~.
%% formally we should distinguish between a lamination and its support...
\]
It is easy to check that $\Gamma_0=\bigcap_{d>0}\Gamma_d$.
Note that $\Gamma_d$ is not discrete. Indeed, transformations $\gamma$ such that 
$\gamma(\lambdat_d)\cap B_p=\emptyset$ form an open subset of $SO^+(2,1)$ that is  contained in
$\Gamma_d$.
On the other hand, we will prove that the intersection of a neighborhood of the identity with
$\Gamma_d$ is discrete and that this neighborhood can be chosen arbitrarily large,
by choosing  $d$  sufficiently large.

\begin{lemma}\label{dr:lm}
For any compact neighborhood $H$ of the identity in $SO^+(2,1)$, there is a constant 
$d$ such that $\Gamma_d\cap H$ is finite.
\end{lemma}

\begin{sublemma}\label{tec:lm}
For any $a>0$  and $d>0$ there is a finite number of strata $F$ of %% $\lambda$
$H^2\setminus \lambdat$
such that $F\cap B_d$ contains a point at distance exactly $a$ from $\partial F$.
\end{sublemma}

\begin{proof}
By contradiction, suppose there are countable many  strata $F_n$ as in the Lemma,
and denote by  $x_n\in F_n\cap B_d$ the points such that $d(x_n,\partial F_n)=a$.

Up to passing to a subsequence, we can suppose that $x_n\rightarrow x$.
If $x$ does not lie in the lamination, then $x_n$ definitively lies in the stratum  $F$ through
$x$, so $F_n=F$, and this contradicts the assumption on $F_n$.

If $x$ lies on $\lambdat$, then $d(x_n,\partial F_n)=d(x_n,\lambdat)\rightarrow 0$, which contradicts
the assumption that this distance is a constant larger than $0$.
\end{proof}

\begin{proof}[Proof of Lemma \ref{dr:lm}]
Let $d_0$ be a fixed number such that $B_{d_0}$ intersects two
two leaves $l_1$ and $l_2$ on the boundary of the stratum $F_0$ through $x_0$.

By the compactness of $H$, there is a constant  $r>0$ such that $d_{\mathbb H^2}(x,\gamma(x))<r$
for any $x\in B_{d_0}$ and $\gamma\in H$.
Note that $\gamma(l_i)$ intersects $B_{d}$ with $d=d_0+r$.

If $\gamma\in H\cap\Gamma_{d}$ with $d>d_0+r$, then
$\gamma$ sends $l_i$ to some leaves $c_1$ and $c_2$ of $\lambdat_d$.
Clearly,  $\gamma(x_0)$ lies in a stratum bounded by $c_1$ and $c_2$, and the
distance between $\gamma(x_0)$ and $c_1$ is the same as the distance
between $x_0$ and $l_1$ (say $a>0$). 
On the other hand, by Sublemma \ref{tec:lm},
there are finitely many strata $F_1,\ldots, F_N$
of $\lambdat_d$ such that $F_i\cap B_{d}$ contains a point at distance $a$ from
$\partial F_i$. Moreover, the boundary of each $F_i$ intersects $B_d$ into a finite 
number of segments.

In particular, there are a finite number of leaves $t_1,\ldots, t_M$ such that
 every $\gamma$  in $H\cap\Gamma_d$ sends $l_i$ to one of the leaves  $t_i$.
However, for  two pairs of geodesics $(l_1, l_2)$ and $(t_1, t_2)$
  in $\mathbb H^2$, there is at most one isometry sending $l_i$ to $t_i$.
  This implies that $H\cap\Gamma_d$ contains at most $2^M$ elements.
 \end{proof}

Let us now fix an observer $(p,v)$.
It then follows from Lemma \ref{td:lm} and Lemma \ref{dr:lm} that for any compact
subset $H\subset SO^+(2,1)$ and $d$ sufficiently large,  there is  a time $T=T(H,d)$
such that the observer at proper time $t>T$ can list the elements of $\Gamma_d\cap H$.
On the other hand, in principle, an observer would have to wait  an infinite amount of  time to determine if a given element
in $\Gamma_{d_0}$ lies also in $\Gamma_0$. Indeed, this amounts to  determining  whether %%that 
such an element
lies also in all $\Gamma_{d}$ for $d>d_0$.
The following lemma ensures that this is not the case and that the observer can be sure after a finite amount of time  that 
elements of $\Gamma_d\cap H$ also lie in $\Gamma_0$.

\begin{lemma}\label{finite:lm}
For any compact subset $H\subset SO^+(2,1)$ there is a constant $d$ such that  $\gamma\in H\cap\Gamma_d$
implies $\gamma\in\Gamma_0$.
\end{lemma}

\begin{proof}
By contradiction,  suppose that %% added
there is a diverging sequence $d_n$ and a sequence $\gamma_n\in H$ such that
$\gamma_n\in\Gamma_{d_n}\cap H$, but $\gamma_n\notin\Gamma_0$.

Up to passing to a subsequence,  we may suppose that $\gamma_n$ converges to $\gamma_\infty$.
On the other hand, by Lemma \ref{dr:lm}, we can choose %% for 
$n_0$ big enough so that $\Gamma_{d_{n_0}}\cap H$ is a finite set.
Now $\gamma_n\in\Gamma_{d_{n_0}}\cap H$ for $n\geq n_0$, and since it is a convergent sequence
we have that  $\gamma_n=\gamma_\infty$ for $n\geq n_1$. This implies that 
$\gamma_\infty\in\Gamma_0$ and contradicts the assumption on the sequence.
\end{proof}

We can now  state and prove the main result of this section.

\begin{prop} \label{pr:construct}
Let  $(p,v)$ be an observer in a domain of dependence $D$ which is the universal cover of a Minkowski
spacetime obtained by Lorentzian grafting of a closed hyperbolic surface along a rational measured lamination. 
%% added hypothesis of rational lamination
Then the observer can construct a finite set of generators of $\Gamma_0$ in finite time.
\end{prop}
 
\begin{proof}
As $\Gamma_0$ is finitely generated, there is a compact subset
$H\subset SO^+(2,1)$ such that $H\cap\Gamma_0$ is a set of generators of $\Gamma_0$.
By Proposition \ref{pr:rational}, Lemma \ref{dr:lm} and Lemma \ref{finite:lm}, the observer can detect elements
of $H\cap\Gamma_0$ in a finite time.
\end{proof}

\subsection{Higher dimensions} \label{ssc:higherdim}

In sections \ref{ssc:cosmological} and \ref{ssc:reconstructing} we focussed on flat spacetimes of dimension $2+1$. 
It appears possible to give  an analogous analysis  in dimension $3+1$
or  higher. However, proving the  results is more involved, since
the structure of MGHFC spacetimes is less well understood and their description is more
complicated than in dimension $2+1$.

In dimension $3+1$, MGHFC spacetimes can still be constructed from a hyperbolic metric
on a $3$-manifolds along with a ``geodesic foliation''. Those foliations, however, are 
different from those occurring on hyperbolic surfaces, since they have two-dimensional
leaves which can possibly meet along one-dimensional strata. 

Still, it appears plausible that the same conclusions can be reached as in dimension 
$2+1$  for observers in a domain of dependence which is the universal cover of 
a MGHFC spacetime in dimension $3+1$:
\begin{itemize}
\item If the  linear part of the holonomy is known, the observer can reconstruct  
 the complete holonomy in finite time.
\item The observer can determine in finite time  the part of the initial singularity
corresponding to the part of space he ``sees'',  which is increasing with time.
\end{itemize}
%To understand whether those statements hold, it would be necessary to extend to
%dimension $3+1$ parts of the study done in the previous sections for dimension $2+1$. 

%:bms5.tex

\section{Examples  in 2+1 dimensions}
\label{sc:2+1}

In this section, we determine explicitly  the frequency function measured by an observer
in a few simple examples of 2+1-dimensional domains of dependence.  

\subsection{Explicit holonomies}
\label{ssc:fig2d}

We consider below three  examples, one based on a reflection in the edges of
a hyperbolic quadrilateral and two based on a hyperbolic punctured torus. In the first
 example involving the punctured torus, the translation part of the holonomy corresponds to a measured
lamination with support on a closed curve.  In the second,  the
support of the measured lamination is more complicated.

%In dimension 2+1, it is easy to construct many examples of MGHFC spaces

\subsubsection{Example 1:  A hyperbolic reflection group}
\label{sssc:quadri_pi3}

In the first example,  we consider a group $\Gamma_{\pi/3}$ which is generated by the
reflections on the edges of a quadrilateral with angles $\pi/3$. (This
angle condition completely determines a presentation of the group, see e.g. \cite{davis:book}.)
The holonomy representation $\rho_t:\Gamma_{\pi/3}\rightarrow \Iso(\R^{1,2})$ depends
on a real parameter $t$. 

We describe first the linear part $\rho_t^l:\Gamma_{\pi/3}\rightarrow O(1,2)$. 
The construction is based on a quadrilateral $p$ with vertices $v_1, \cdots, v_4$.
Consider the hyperbolic plane as a quadric in the Minkowski space  $\R^{2,1}$, and 
let $w_{1},w_{2},w_{3},w_{4}$ be the unit spacelike vectors which are orthogonal
to the oriented plane through $0$ containing the geodesic segments $(v_1,v_2)$, $(v_2,v_3)$, $(v_3,v_4)$ 
and $(v_4,v_1)$. 
The cosine of the exterior angle of $p$ at $v_i$ is then 
equal to the scalar product between $w_{i-1}$ and $w_{i}$, so that $p$ has
interior angles equal to $\pi/3$ if and only if $\langle w_i,w_{i+1}\rangle = -1/2$
for all $i\in \Z/4\Z$. 

The fact that those scalar products are equal means that $(w_1,w_2,w_3,w_4)$ is a
rhombus in the de Sitter plane. In particular it is invariant under the symmetry
with respect to a timelike line in $\R^{2,1}$, corresponding to a point $o\in H^2$
which is the midpoint of both, $(v_1,v_3)$ and $(v_2,v_4)$. 

Choosing a coordinate system compatible with this symmetry, we can write the $w_i$
as
$$ w_{1}=(\sinh(t),\cosh(t),0), w_{2}=(\sinh(t'),0,\cosh(t')), w_3=(\sinh(t),-\cosh(t),0), 
w_4=(\sinh(t'),0,-\cosh(t')) $$
with $t,t'$ satisfying the condition  $\sinh(t)\sinh(t')=1/2$. 
This yields a representation $\rho_t:\Gamma\rightarrow O(2,1)$ which
can be described as follows. $\Gamma$ is generated by the elements $a_1, \cdots, a_4$
corresponding to the reflections in the edges of $p$, with the relations
$$ a_i^2=1, \quad (a_ia_{i+1})^3=1 $$
for all $i\in \Z/4\Z$. The representation $\rho$ sends $a_i$ to the reflection in 
$(v_i,v_{i+1})$, that is
$$ \rho(a_i)(x) = x-2\langle x,w_i\rangle w_i~. $$
The quotient of $\mathbb H^2$ by $\rho(\Gamma)$ is an orbifold. The subgroup $\Gamma_2$ of $\Gamma$ 
of elements $\gamma\in \Gamma$ for which $\rho(\gamma)$ is orientation-preserving
has index two, and the quotient of $\mathbb H^2$ by $\rho(\Gamma_2)$ is a surface.

There is a unique choice of a deformation cocycle associated to $\rho_t$, which is obtained by
varying $t$.  It can be written as $\tau_t=\rho_t^{-1}d\rho_t/dt$. So we obtain  a one-parameter family of domains of dependence parametrized by $t$. In the following, we mainly consider the simplest case, where $t=t_0=\sinh^{-1}(1/\sqrt{2})$,
so that $t'=t_0$. 

%Figure \ref{fig:dd} shows the shape of the domain of dependence
%with holonomy representation $(rho_{t_0},\tau_{t_0})$.

\subsubsection{Example 2: A punctured torus with a rational measured lamination}
\label{sssc:torus1}

Another simple example can be constructed, by choosing as the linear part of the holonomy  the holonomy representation of a hyperbolic punctured torus.

Then the group $\Gamma$ is  the free group generated by two elements $a,b$. We consider
the situation with an extra symmetry, corresponding to the condition that the images
of $a,b$ by the linear part $\rho$ of the holonomy are hyperbolic translations with
orthogonal axes. The translation lengths of $\rho(a)$ and $\rho(b)$ can then be
written as $2t_a, 2t_b$,  subject to the conditions that $\sinh(t_a)\sinh(t_b)=1$. This corresponds
to the condition that the image of the commutator of $a,b$ is parabolic.

In the computations below, we choose, somewhat arbitrarily, the parameters $t_a=\sinh^{-1}(2),
t_b=\sinh^{-1}(1/2)$. We also choose the translation component of the holonomy as the
cocycle $\tau$ corresponding, through the relation explained in Section \ref{ssc:equivariant},
to a closed curve corresponding to $b$, with weight $1$. 

The domain of dependence obtained
in this way is shown on the left in Figure \ref{fig:dd}. 
To compute this image  as well as frequency function and ``distances'' to the boundary in
other domains of dependence 
 -- see figures \ref{fig:da_torus1} and \ref{fig:da_torus4} below --- 
we use the description of a domain of dependence as an intersection of half-spaces
bounded by lightlike planes from Section \ref{ssc:reconstructing_holo}. The image
is computed by taking a ball of radius $6$ in $\Gamma$, for the distance defined
by the choice of generators described above, computing for each element of 
$\Gamma$ in this ball the corresponding axis and lightlike hyperplanes, and then
determining their intersection.

%The computation is done as 
%described above for the group generated by symmetries in the edges of a quadrilateral 
%with angles $\pi/3$, however we only used a ball of radius $6$ in the fundamental group
%(the growth of the group is faster here since the fundamental group is free).

\subsubsection{Example 3: A punctured torus with an  irrational lamination}
\label{sssc:torus4}

In this example,  the linear part of the holonomy is the same as in the previous one. However,  the translation part of the holonomy is given
by a cocycle corresponding to a measured lamination which does not have its support on 
a closed curve. The corresponding domain of dependence is given on the right in 
Figure \ref{fig:dd}.

%% {\bf To be completed -- the computation I had were still on a rational lamination.
%% Figure \ref{fig:dd_torus3} must be changed, too.}
%% Now the figure should be correct.

\subsection{Results}

\begin{figure}[ht]
  \begin{center}
  \includegraphics[width=8cm]{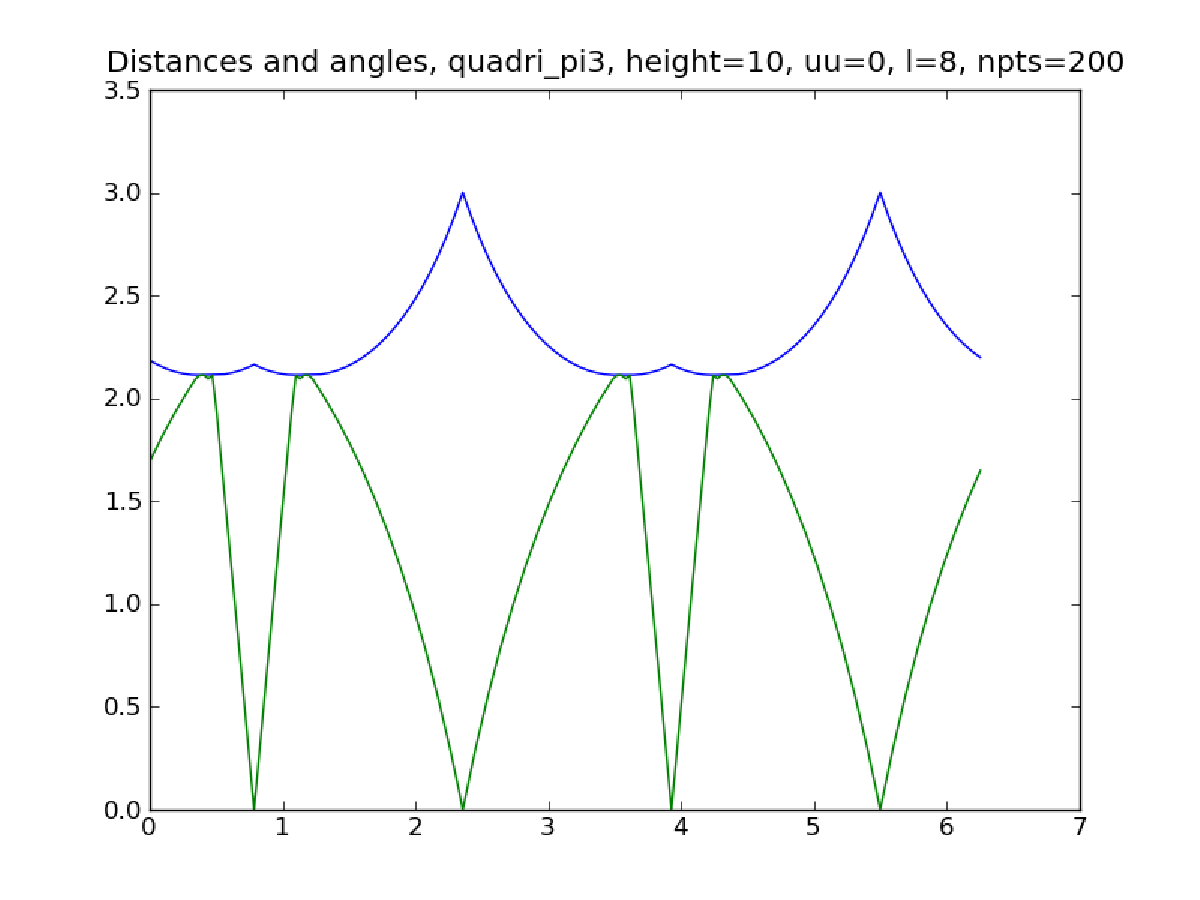}
  \includegraphics[width=8cm]{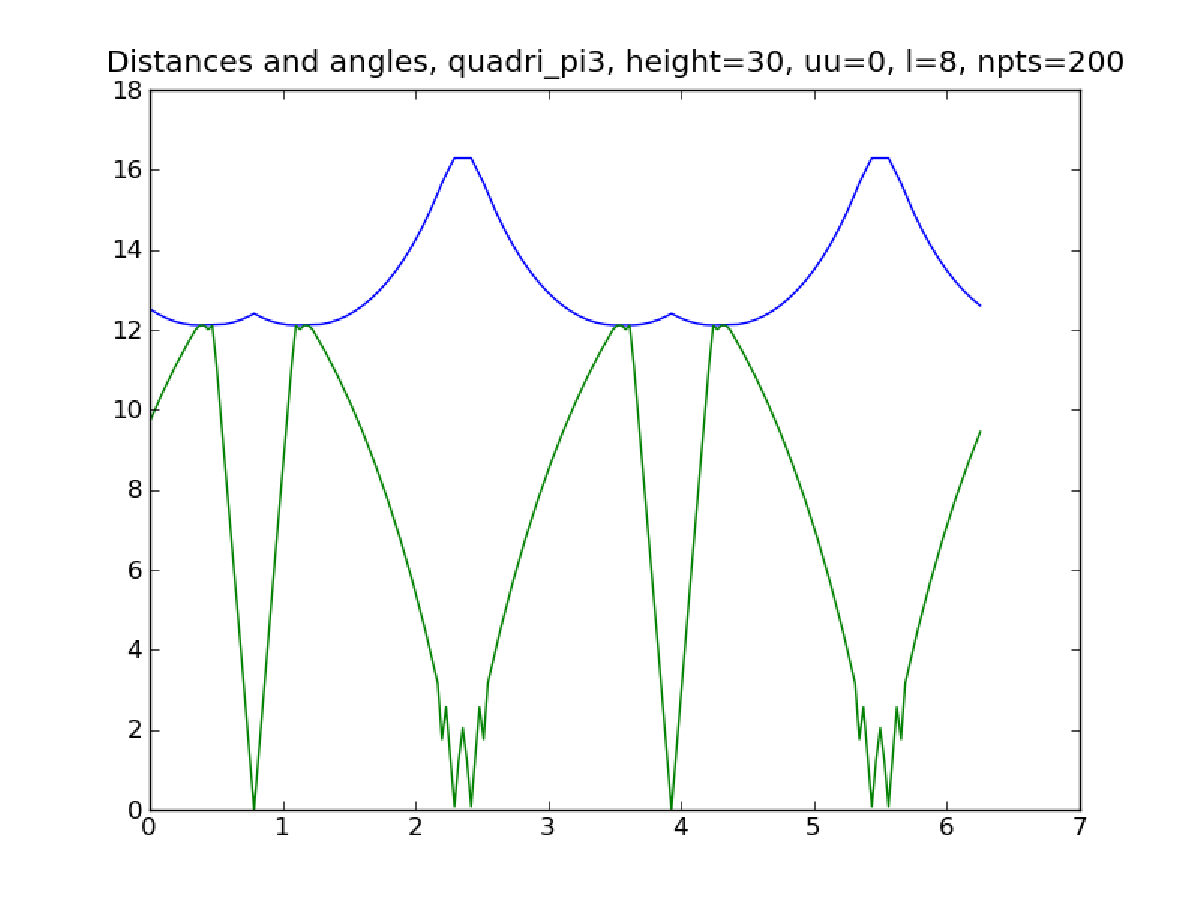}
  \includegraphics[width=8cm]{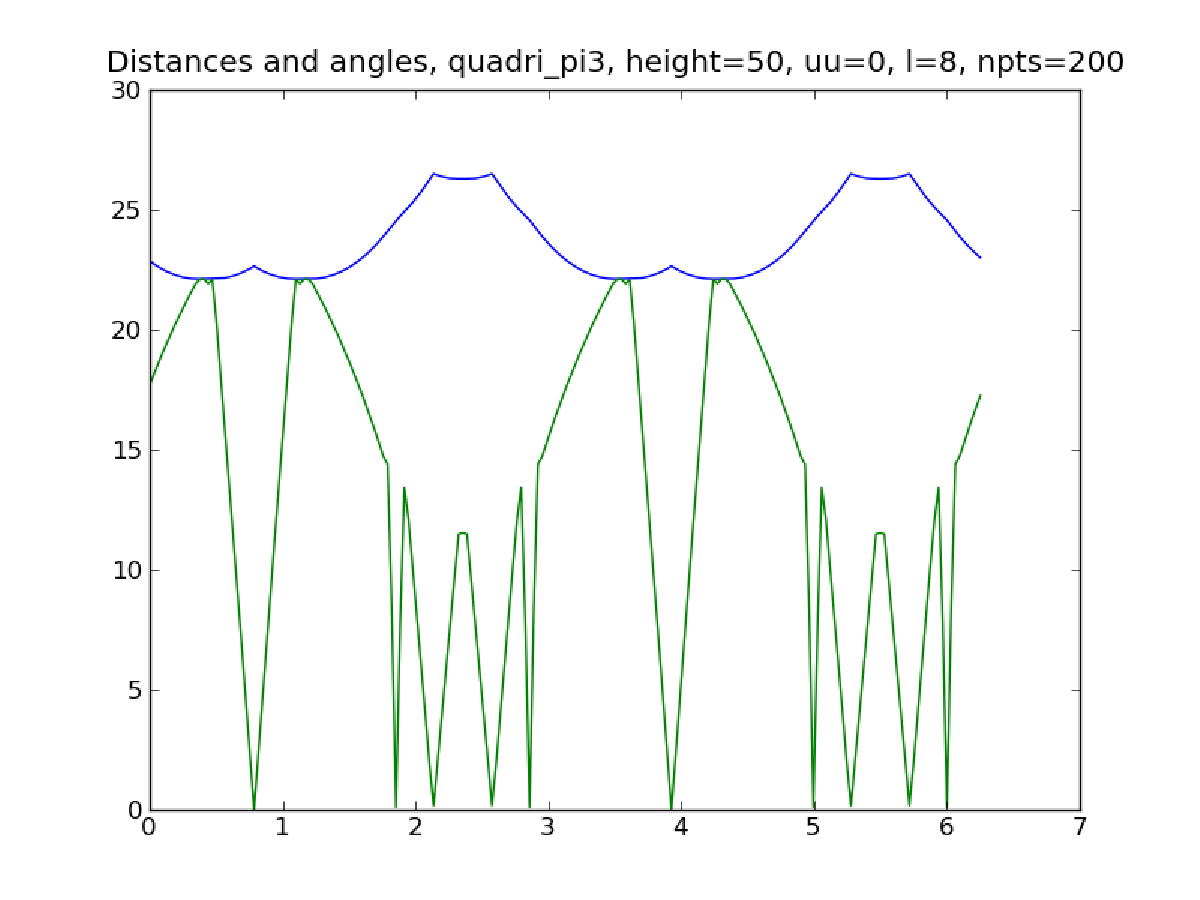}
  \caption{The frequency function (green) and Euclidean distance to the boundary (blue) in a 
domain of dependence based on a quadrilateral with angles $\pi/3$ with a rational lamination, seen from
increasing distance from the initial singularity.}
  \label{fig:da_quadri_pi3}
  \end{center}
\end{figure}

\begin{figure}[ht]
  \begin{center}
  \includegraphics[width=8cm]{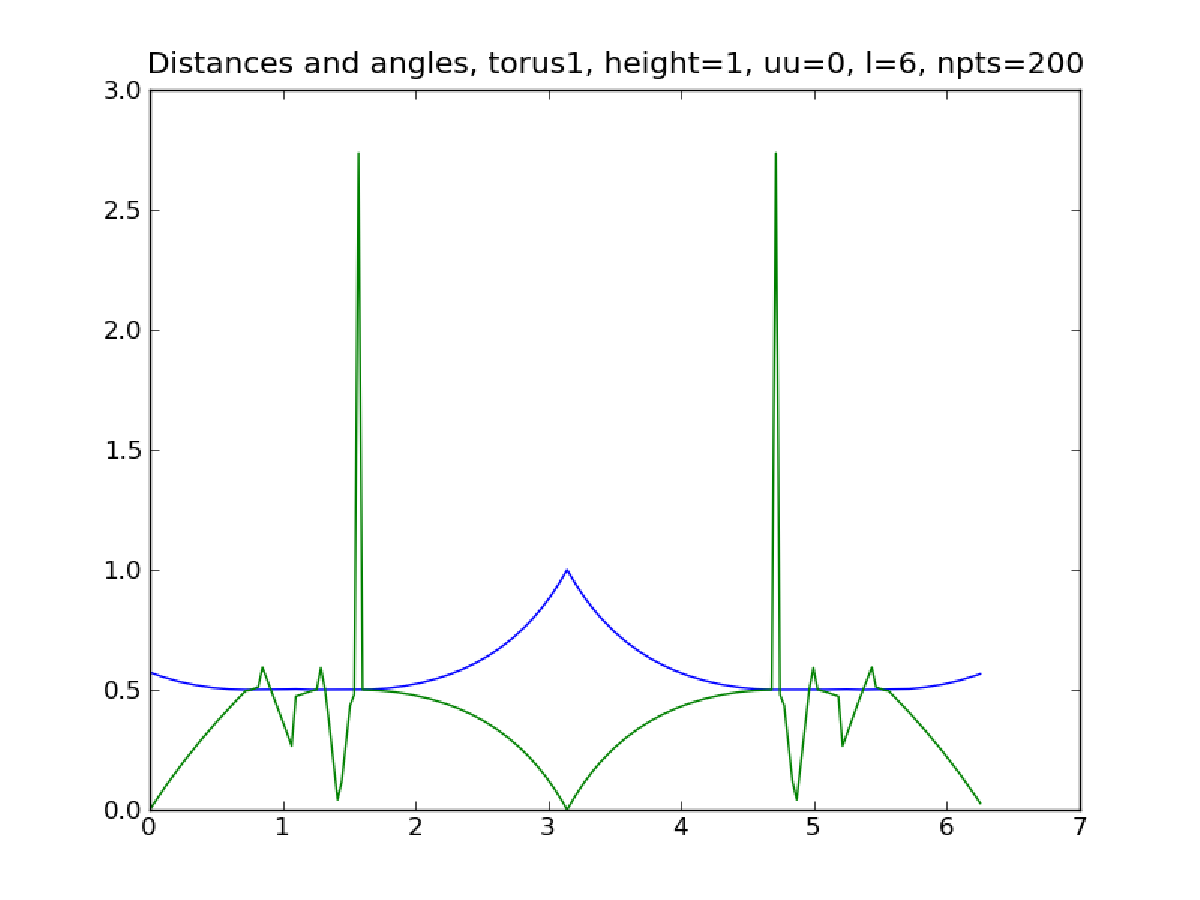}
  \includegraphics[width=8cm]{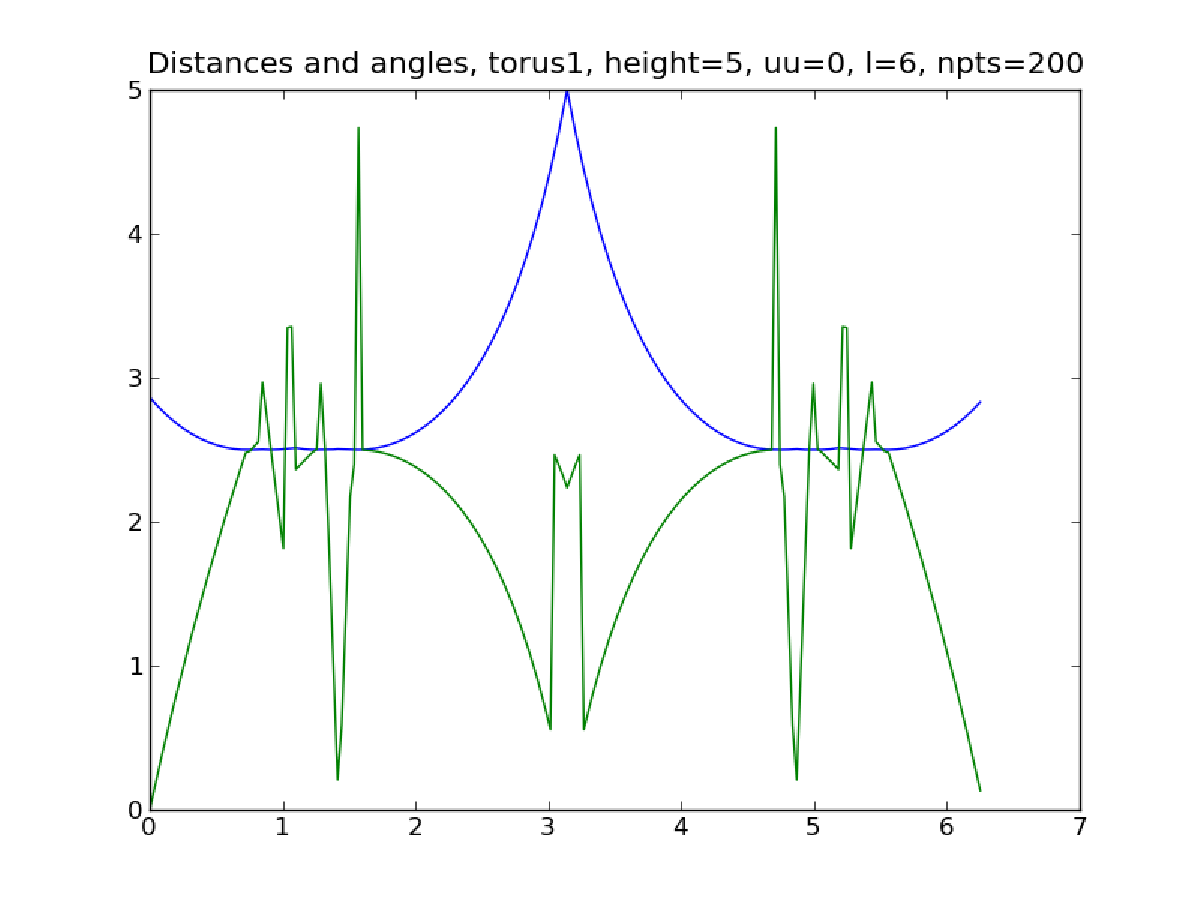}
  \includegraphics[width=8cm]{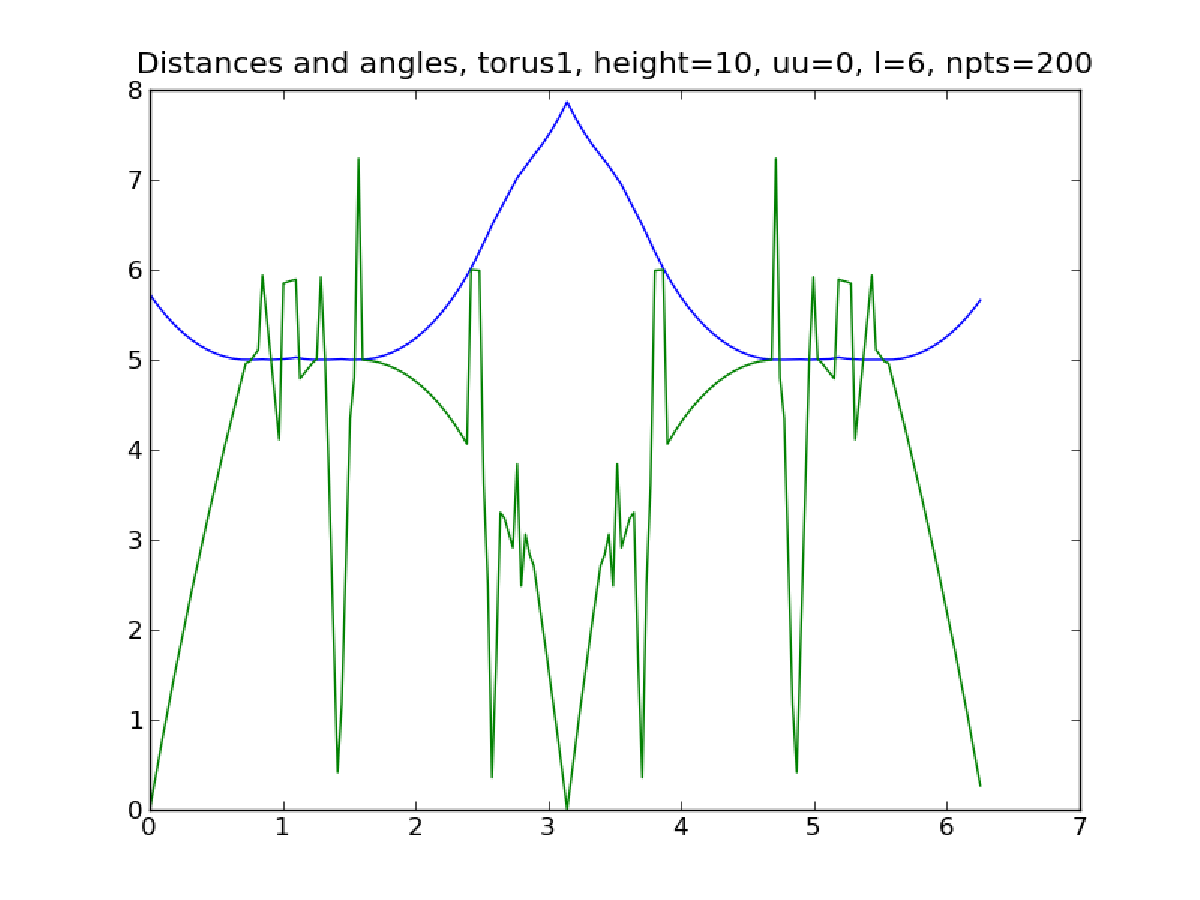}
  \includegraphics[width=8cm]{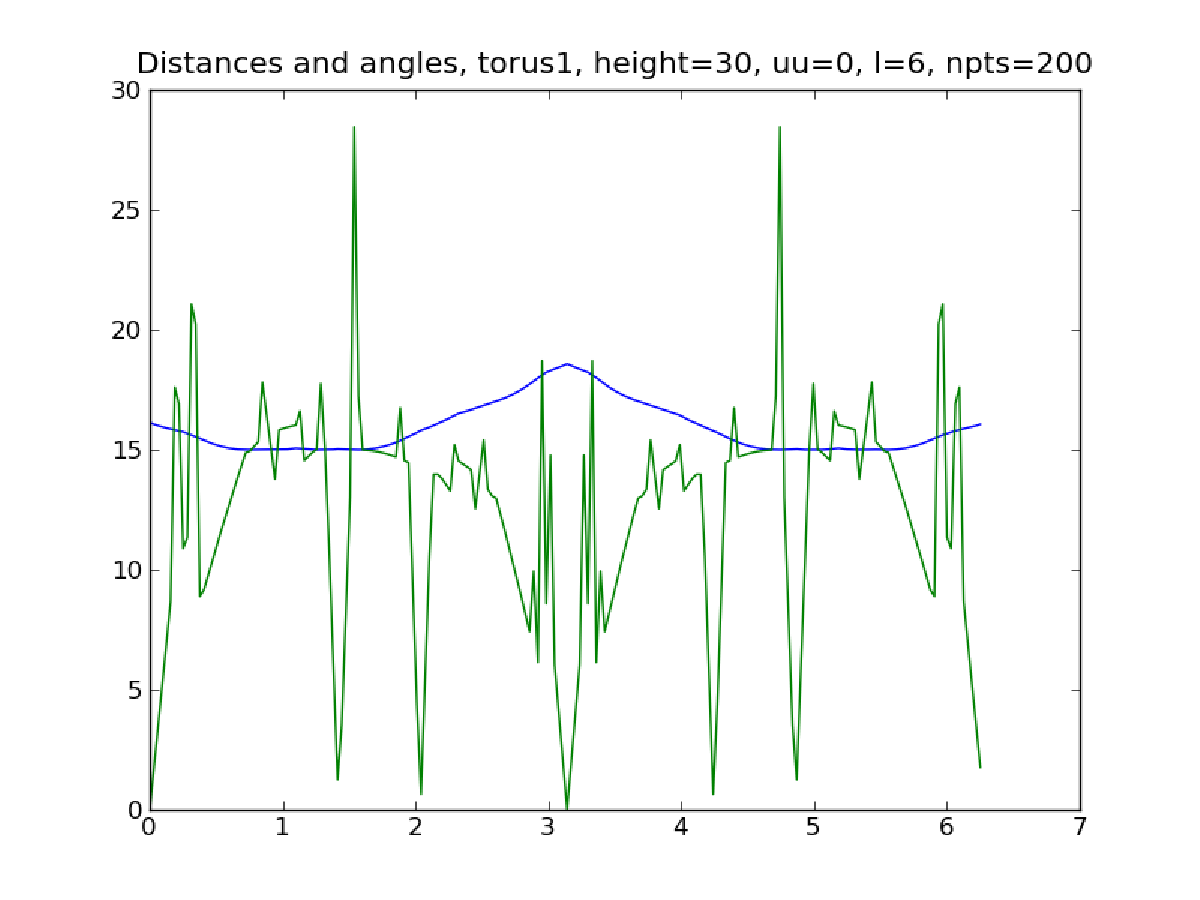}
  \caption{The frequency function  (green) and Euclidean distance to the boundary (blue) in a 
domain of dependence based on a punctured torus with a rational lamination, seen from
increasing distance from the initial singularity.}
  \label{fig:da_torus1}
  \end{center}
\end{figure}

\begin{figure}[ht]
  \begin{center}
  \includegraphics[width=8cm]{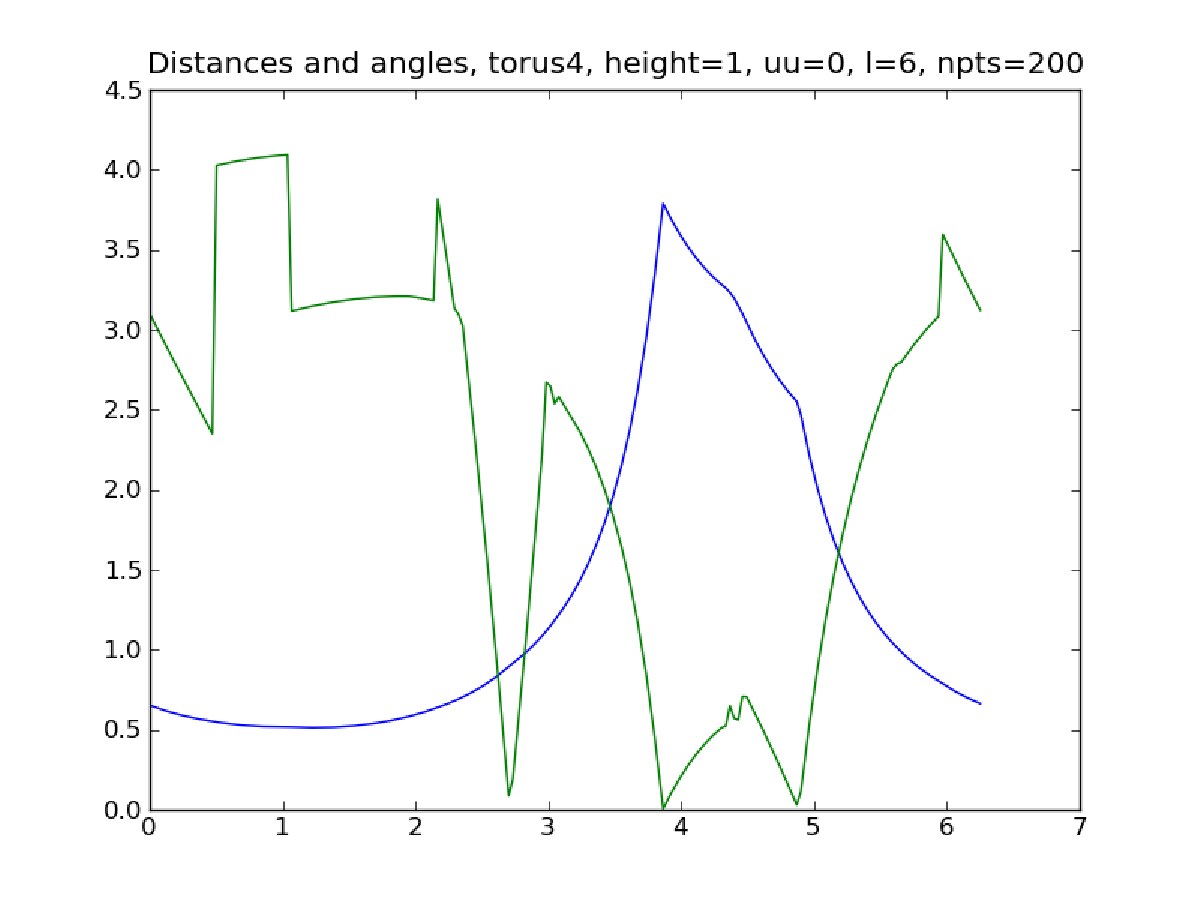}
  \includegraphics[width=8cm]{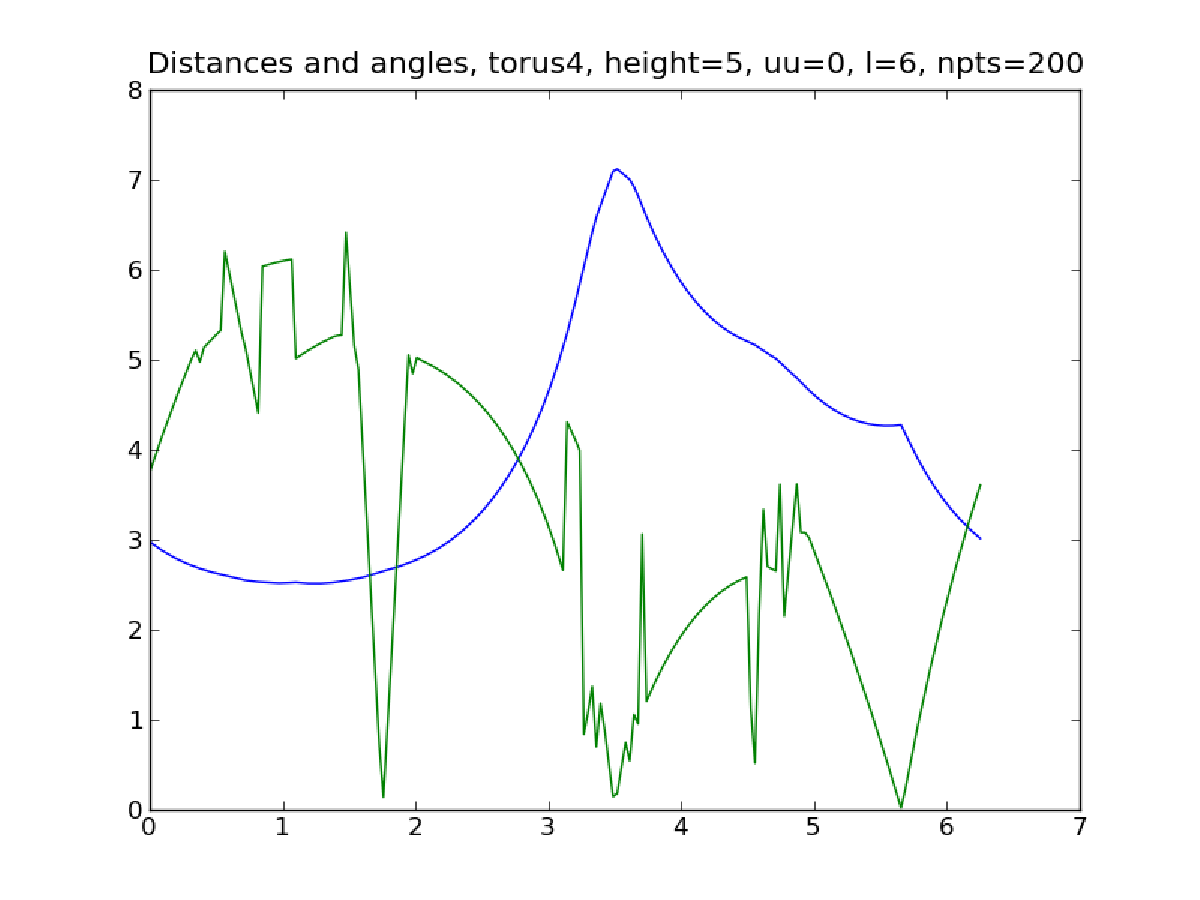}
  \includegraphics[width=8cm]{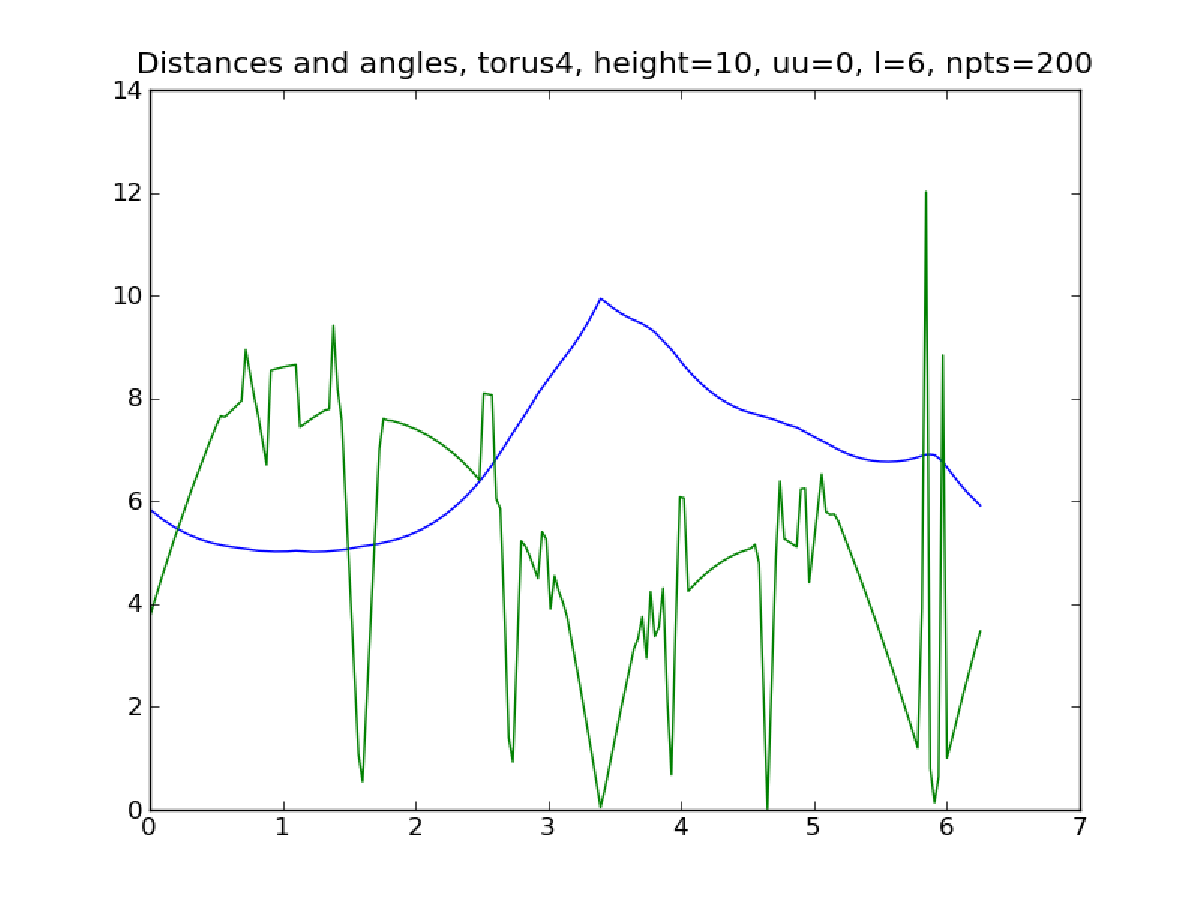}
  \includegraphics[width=8cm]{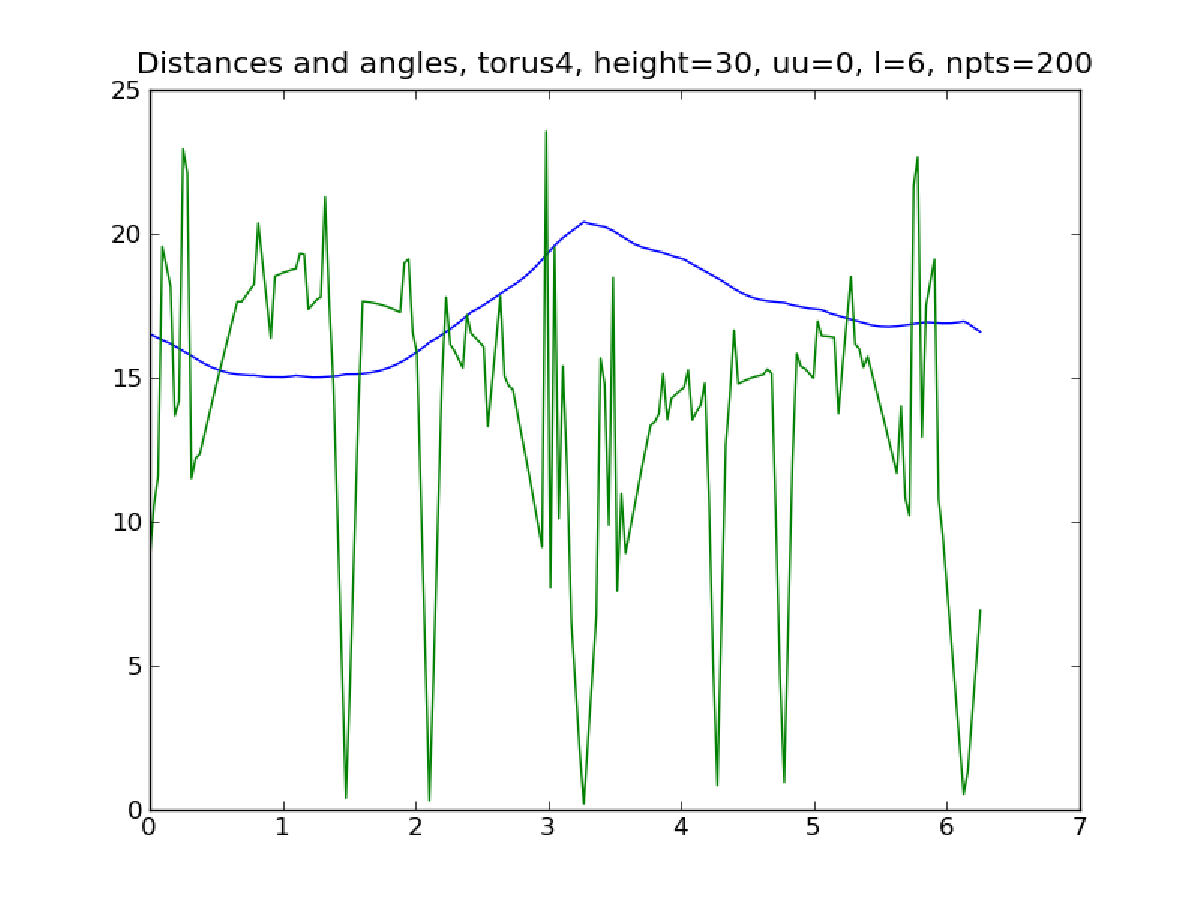}
  \caption{The frequency function (green) and Euclidean distance to the boundary (blue) in a 
domain of dependence based on a punctured torus with an irrational lamination, seen from
increasing distance from the initial singularity. }
  \label{fig:da_torus4}
  \end{center}
\end{figure}

Figure \ref{fig:da_quadri_pi3} presents (in green) the frequency function measured by all directions
by an observer in the $2+1$-dimensional flat spacetime described in Section \ref{sssc:quadri_pi3}.
More precisely, the observer is located above the origin at time distances $10$, 
$30$ and $50$ to
achieve the most ``readable'' results.% --- locating the observer at different points in the domain of dependence leads to different results. 
The Euclidean length of the lightlike
segments from the observer to the boundary of the domain are drawn in blue. The frequency
function becomes more complex as the cosmological time of the observer increases.

Figure \ref{fig:da_torus1} is the analogue of Figure \ref{fig:da_quadri_pi3} 
for the $2+1$-dimensional flat spacetime described in Section \ref{sssc:torus1}.
More precisely, the observer is  located above the origin at time distances $1$, 
$5$, $10$ and $30$.  Figure \ref{fig:da_torus4} shows the analogous  results, 
for the domain of dependence obtained from an irrational lamination, 
described in Section  \ref{sssc:torus4}. 

The computation of both the Euclidean distance and the frequency function are made for
an approximation of the domain of dependence, as explained above. It follows from 
Section 4 that the frequency function  computed in this way is not reliable as a continuous function,
but only --- possibly at least --- as a $L^1$ function, due to Theorem \ref{stability:thrm}.

Those graphs should be considered as preliminary results, since, even for this relatively
simple setting, the computations needed to obtain the results are quite involved relative to
our programming capabilities and computing equipment. It is possible that heavier 
computations --- in particular, computing a better approximation of the domain of dependence
by using a larger subset of the fundamental group --- could lead to notably different results.
However, it  is already  apparent in those pictures,  that the frequency function behaves
in a very non-smooth way, as explained in Section 3 and Section 4.

%:bms6.tex

\section{Examples  in 3+1 dimensions}
\label{sc:3+1}

In this section,  we consider an  example of a domain of dependence in 3+1 dimensions  and show  that the light emitted from the initial singularity and received by an observer  contains  rich information on its geometry and topology. In the first part,  we describe the domain of dependence in 3+1 dimensions, based on a construction of 
Apanasov \cite{apanasov:deformations}.
The second part contains some images of the light emitted by the initial singularity, as
seen by an observer.

%\subsection{Domains of dependence in dimension 3+1}

\subsection{An explicit example} \label{ssc:apanasovex}

We consider a particularly interesting example of domain of dependence, which is used
in computations below. The remarkable property of this example is that, for only one 
linear part of the holonomy, there is a four-dimensional space of possible translation
components. 

\subsubsection{The construction of the group}
\label{ssc:group}

The following example is essentially due to Apanasov \cite{apanasov:deformations}.
It is a discrete group $\Gamma$ of $Isom(\mathbb H^3)$ generated by $8$ reflections, so 
that the quotient $\mathbb H^3/\Gamma$ is a (non-orientable) orbifold of finite volume.

On $S^2_\infty=\mathbb C\cup\{\infty\}$ we consider the following circles
\begin{enumerate}
\item $C_k$ with center at $z_k=\sqrt{3}e^{i\frac{k\pi}{3}}$ and radius $1$.
\item $C$ with center at $0$ and radius $1$.
\item $C'$ with center at $0$ and radius $2$.
\end{enumerate}

\begin{figure}[ht]
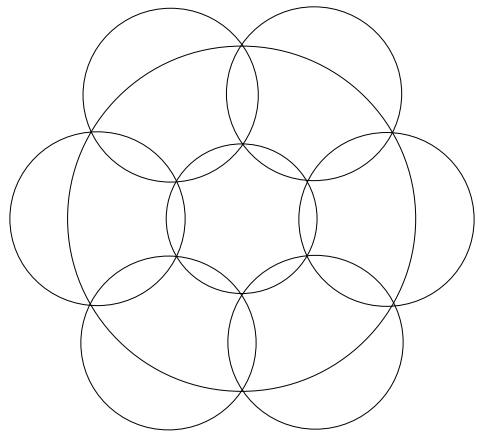

\begin{center}
\input group.pstex_t
\caption{Construction of the linear part of the holonomy}
\end{center}
\end{figure}

It can be shown easily that the configuration of such circles is the one shown in the picture.
Moreover the angle formed by any two circles in the list (that meet each other) is $\pi/3$.

We consider the planes $P_k, P, P'$ in $\mathbb H^3$ that bound at infinity the circles $C_k$, $C$ and $C'$.
We denote by $\gamma_k$ the reflection along $P_k$, by $\gamma$ the reflection along  $P$ and by
$\gamma'$ the reflection along $P'$.

\begin{prop}
The group $\Gamma$ generated by $\gamma_k, \gamma,\gamma'$ is a discrete group in
$Isom(\mathbb H^3)$ and the quotient $\mathbb H^3/\Gamma$ is a non-orientable orbifold of finite volume.
\end{prop}

A fundamental region for the action of $\Gamma$ can be  obtained as follows. 
Denote for any $k$ by  $E_k$ the exterior of the plane $P_k$, e.~g.~the region of $\mathbb H^3\setminus P$ that
contains $\infty$. Analogously, denote by $E$ the exterior of $P$ and  by $I'$  the interior region
bounded by $P'$. Then a fundamental region for $\Gamma$ is given by
\[
   K=\bigcap_{k=0}^5E_k\,\cap\, E\,\cap\, I'\,. 
\] 

As we are interested in  $\mathbb R^{3,1}$-valued cocycles, we need to determine
explicitly the matrices in $O(3,1)^+$ corresponding to the transformations $\gamma_k, \gamma,\gamma'$.

In $\mathbb R^{3,1}$ we consider coordinates $x_0,x_1,x_2,x_3$ so that the Minkowski metric takes the form  $-dx_0^2+dx_1^2+dx_2^2+dx_3^2$. Given two real numbers $v_0,v_1$ and a complex number $z=x+iy$,  we denote by $(v_0,v_1,z)$ the point $(v_0,v_1,x,y)$ in $\mathbb R^{3,1}$.

We have to fix explicitly an isometry between the half-space model of
$\mathbb H^3$ (denoted by $\Pi$ here) and the hyperboloid
model  denoted by $\mathbb H^3$. Such an isometry
$\phi:\Pi\rightarrow\mathbb H^3$ is given by
\[
   \phi(i)=(1,0,0,0)\qquad\qquad
   \phi_{*,i}(a\frac{\partial}{\partial x}+b\frac{\partial}{\partial y}+c\frac{\partial}{\partial}z)
   =(0,c,a,b)\,.
\]
With this choice,  the plane $P_k$ is  given by
$P_k\mathbb H^3\cap v_k^\perp$, where 
\[
   v_k=(3/2, 1/2, \sqrt{3} e^{i\frac{k\pi}{3}})
\]
is a unit vector.
Analogously, we obtain  $P=\mathbb H^3\cap v^\perp$ and $P'=\mathbb H^3\cap (v')^\perp$ where
\[
  v=(0,1,0,0)\qquad v'=(3/4,5/4,0,0)\,.
\]
The associated transformations in $O^+(3,1)$ then take the form
\[
 \gamma_k(x)=x-2\langle x,v_k\rangle v_k,\qquad\gamma(x)=x-2\langle x, v\rangle v
\qquad\gamma'(x)=x-2\langle x,v'\rangle v'\,.
\]

\subsubsection{The construction of the cocycle}
\label{ssc:cocycle}

\begin{figure}
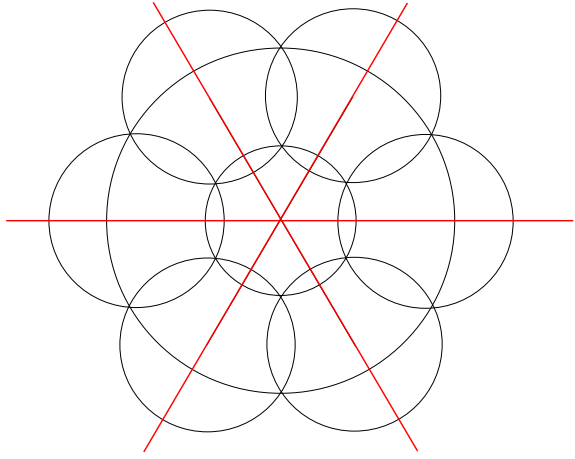

\begin{center}
\input group2.pstex_t
\caption{Construction of the translation cocycle}
\end{center}
\end{figure}

In $S^2_\infty=\mathbb C\cup\{\infty\}$ we consider the three lines through $0$ passing though the centers
of the circles $C_k$ and denote by  $W$ be the union of the  planes in $\mathbb H^3$ bounding these lines.
Clearly,  $W$ is the union of $6$ half-planes which meet  along the geodesic $l_0$ joining $0$ to $\infty$.
We denote these half-planes by  $W_0,W_1, W_2, W_3, W_4, W_5$ where the indices correspond to the ones of the circles in the obvious way.
Let $w_k\in\mathbb R^{3,1}$ be the vector orthogonal to $W_k$ and pointing towards $W_{k+1}$.
A direct computation then shows that $w_k$ is given by
\[
   w_k=(0,0,ie^{i\frac{k\pi}{3}})\,.
\]
Note that $w_{k+3}=-w_k$,  where the index $k$ is considered $\textrm{mod} 6$.

Now the $\Gamma$-orbit of $W$ is a  branched-surface in $\mathbb H^3$.
In particular, the sets  $\hat W=\Gamma\cdot W$ and $\hat W_0=\Gamma\cdot l_0$  have the following properties:
\begin{itemize}
\item
$\hat W_0$ is a disjoint union of geodesics.
\item Every connected component of $\hat W\setminus\hat W_0$ is a convex 
polygon (with infinitely many edges) and every edge is an element of $\hat W_0$.
\item There are exactly six faces up to the action of $\Gamma$. Indeed,  let $F_k$ be the face of $\hat W$
bounding $l_0$ and contained in $W_k$. Then the orbits of $F_0,\ldots, F_5$ are disjoint and
cover $\hat W$.
\item
Every connected component of $\mathbb H^3\setminus\hat W$ is a convex polyhedron.
\end{itemize}

These properties can be proved by considering the  intersection of $W$ with $K$ and using the fact that  $W$ is orthogonal
to the faces of $K$.

We can then  use a general construction explained in \cite{bonsante} to obtain non-trivial cocycles.
Given six  numbers $a_i$ such that
\begin{equation}\label{eq1}
   \sum_{i=0}^5  a_i w_i=0,
\end{equation}
we  obtain a cocycle via the following prescription.
We associate to the face $F_i$ the number $a_i$. In this way a number $a(F)$ is associated to
every face $F$ by requiring that  $a(\alpha(F))=a(F)$ for every $\alpha\in\Gamma$. 

Then we fix a basepoint $x_0$ in $\mathbb H^3$ that does not lie in $\hat W$.
Given a transformation $\alpha\in\Gamma$ we construct a vector in $\mathbb R^{3,1}$ 
in the following way:
we take any path $c$ joining $x_0$ to $\alpha(x_0)$ and avoiding $\hat W_0$. 
The path $c$ intersects some faces $F^1,\ldots, F^n$. We consider the unit vector $w^j\in\mathbb R^{3,1}$
orthogonal to $F^j$ and pointing  towards $\alpha(x_0)$ and set
\[
     \tau(\alpha)=\sum_{j=1}^na(F^j)w^j\,.
\]

It can be easily checked that
\begin{itemize}
\item $\tau(\alpha)$ does not depend on the path $c$ (this essentially follows from (\ref{eq1})).
\item $\tau$ is a $\mathbb R^{3,1}$-valued cocycle 
\item changing the basepoint changes $\tau$ by a coboundary.
\end{itemize}
Let $H\subset \mathbb R^6$ be the subspace  of solutions of (\ref{eq1}), which is of dimension  $\dim H=4$.
From \cite{bonsante} we have that the map
\[
   H\rightarrow H^1(\Gamma,\mathbb R^{3,1})
\]
is injective.
This map can be computed explicitly.
More precisely,  we fix the base point $x_0$ in the region of $K$ between  $W_0$ and $W_1$.

Given a set of  numbers $a_0,a_1,a_2,a_3, a_4, a_5$ that satisfy (\ref{eq1}) we compute
 the corresponding cocycle $\tau$ evaluated  on the generators.
This yields
\[
\begin{array}{l}
\tau(\gamma)=\tau(\gamma')=\tau(\gamma_0)=\tau(\gamma_1)=0\\
\tau(\gamma_2)=w_1-\gamma_2 w_1=3a_1v_2\\
\tau(\gamma_3)=w_1-\gamma_3 w_1+ w_2-\gamma_3w_2=3(a_1+a_2)v_3\\
\tau(\gamma_4)=-w_0+\gamma_4 w_0-w_5+\gamma_4 w_5=-3(a_0+a_5)v_4\\
\tau(\gamma_5)=-w_0+\gamma_5 w_0=-3a_0v_5\,.
\end{array}
\]

\subsection{Computations}

The computations  of the frequency functions were limited by the speed of the available computers, and more complete computations could provide better results. As in 
dimension $2+1$,  we constructed a domain of dependence in $\R^{3,1}$, invariant under
the group actions described above via  the construction in Section \ref{ssc:reconstructing_holo}.
However,  the computations are much heavier in dimension $3+1$, so we only considered the 
elements of the fundamental group in a ball of radius $4$.

Although we only considered one linear part of the holonomy --- the one in Section \ref{ssc:group} ---
we worked with two deformation cocycles, one corresponding to weights $(1,0,0,1,0,0)$ as described
in Section \ref{ssc:cocycle}, the other to the weights $(1,1/2,0,1,1/2,0)$. In both cases,  the 
observer was located at the point of coordinates $(50,0,0,0)$. This is a somewhat arbitrary choice,
made after trying different possibilities, which leads to interesting pictures.

The frequency function measured by the observer for the first choice of cocycle is presented in Figure \ref{fig:intensity1}, with
different colors encoding different values of the frequency. 

\begin{figure}[ht]
  \begin{center}
  \includegraphics[width=8cm]{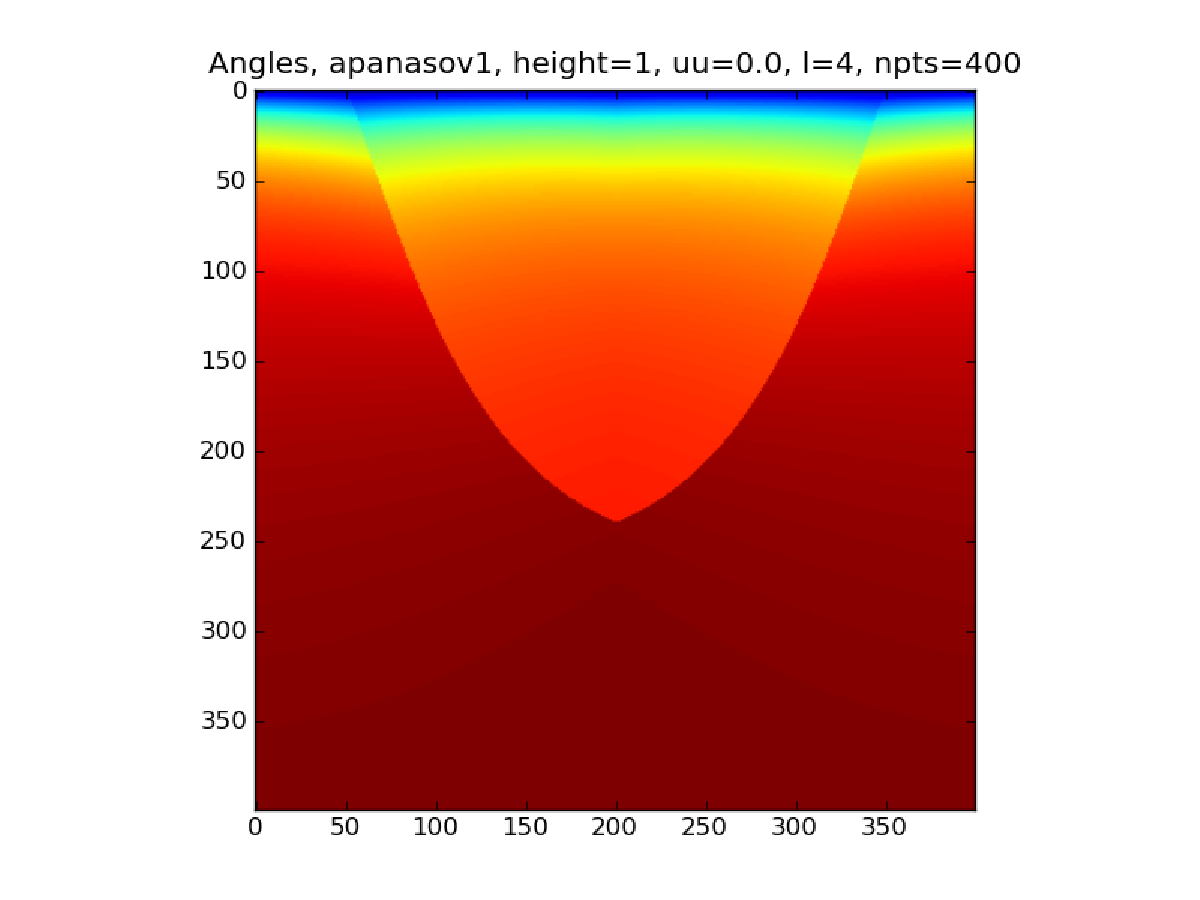}
  \includegraphics[width=8cm]{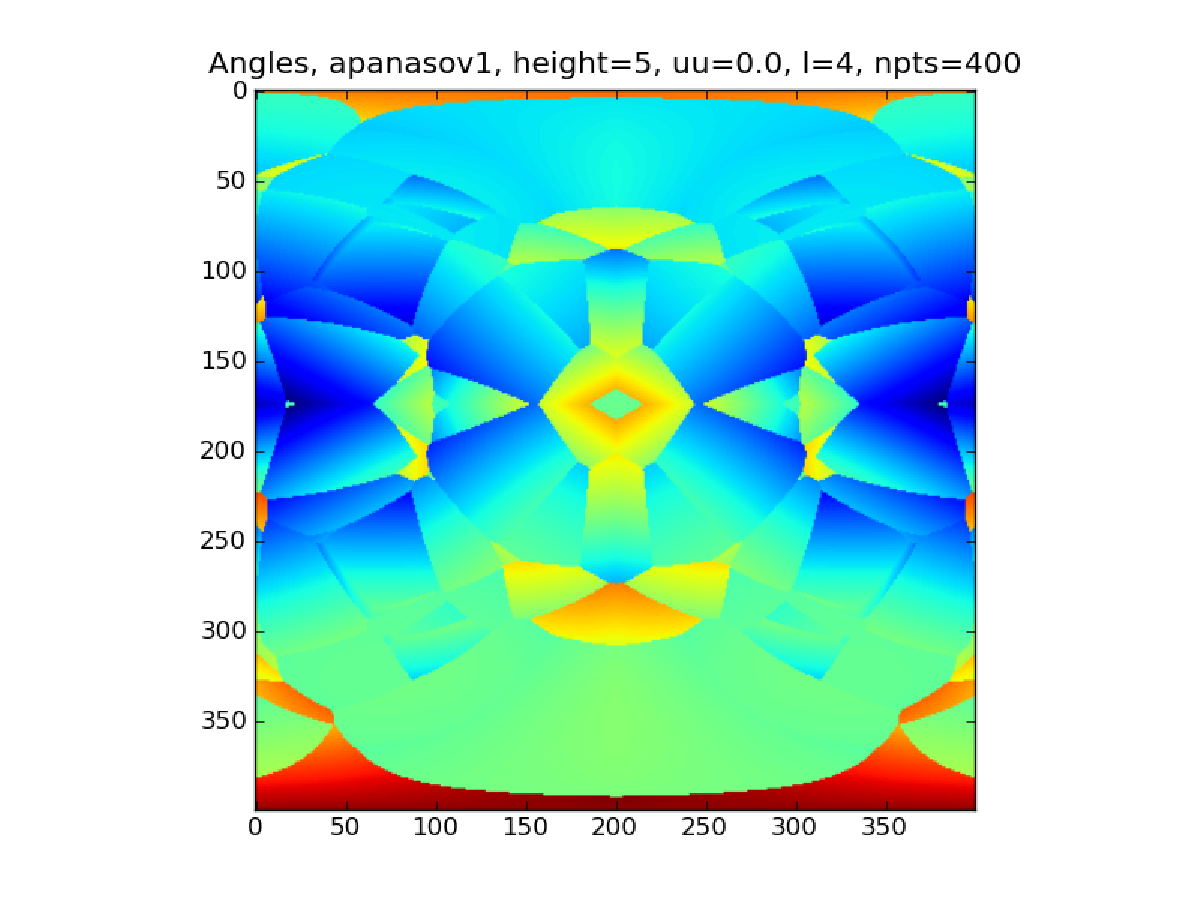}
  \includegraphics[width=8cm]{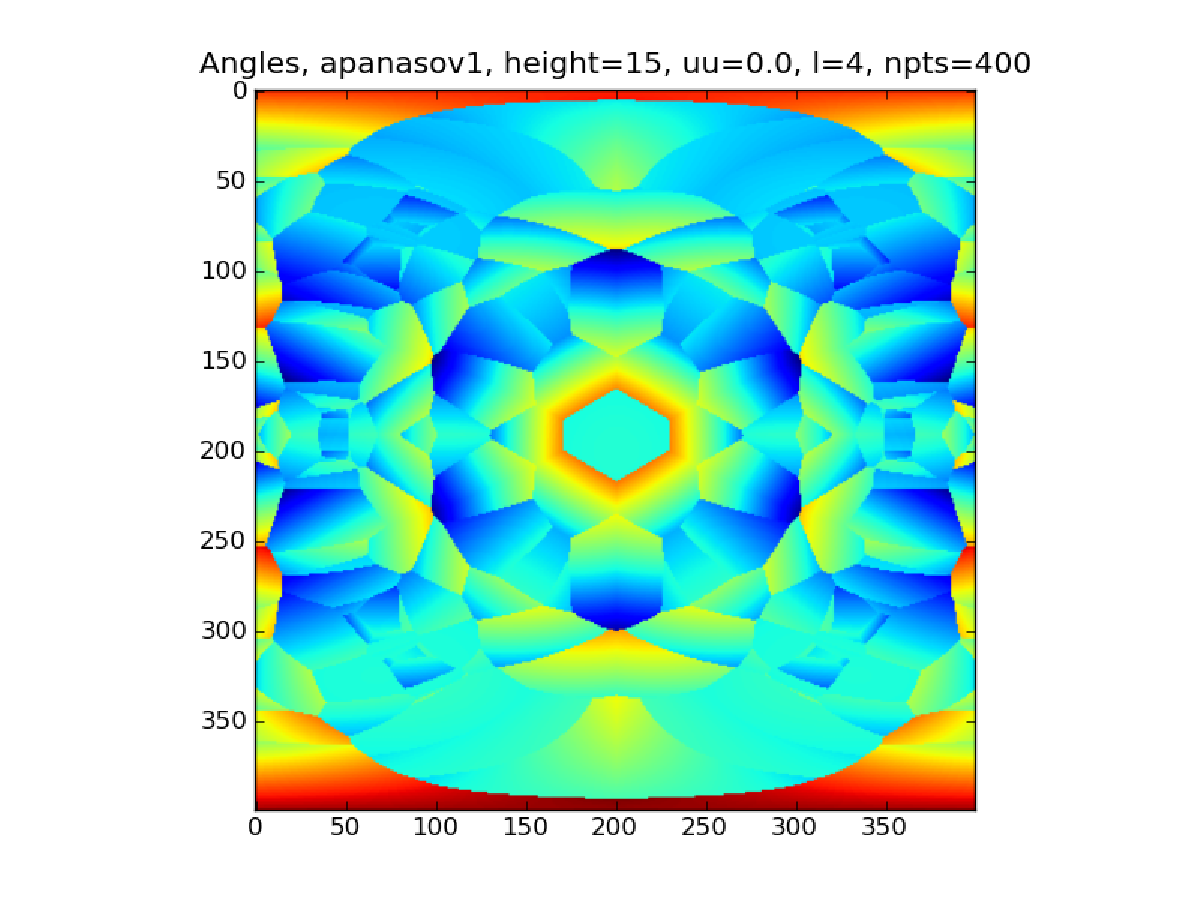}
  \includegraphics[width=8cm]{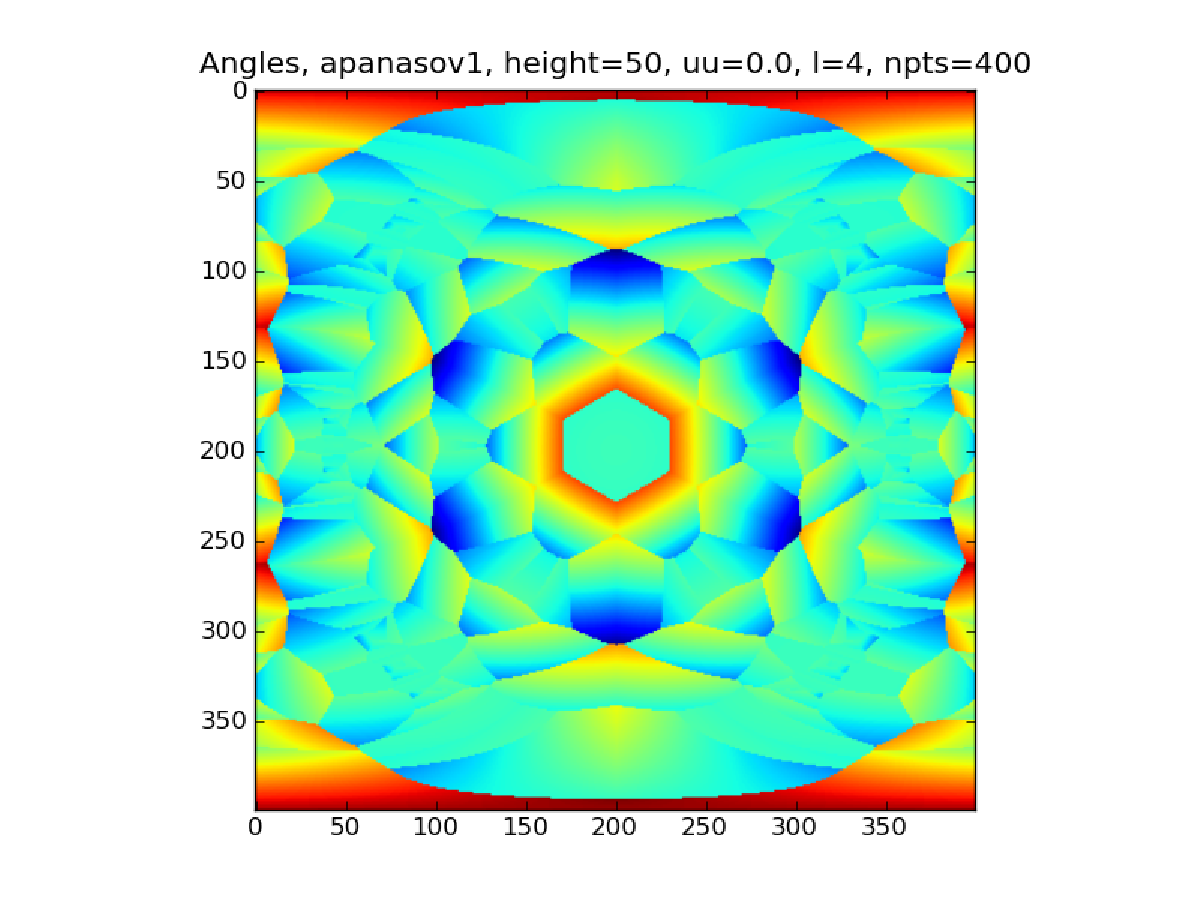}
  \caption{Computed frequency function, with translation coefficients $(1,0,0,1,0,0)$, 
for observers at an increasing distance from the initial singularity}
  \label{fig:intensity1}
  \end{center}
\end{figure}
 
\begin{figure}[ht]
  \begin{center}
  \includegraphics[width=8cm]{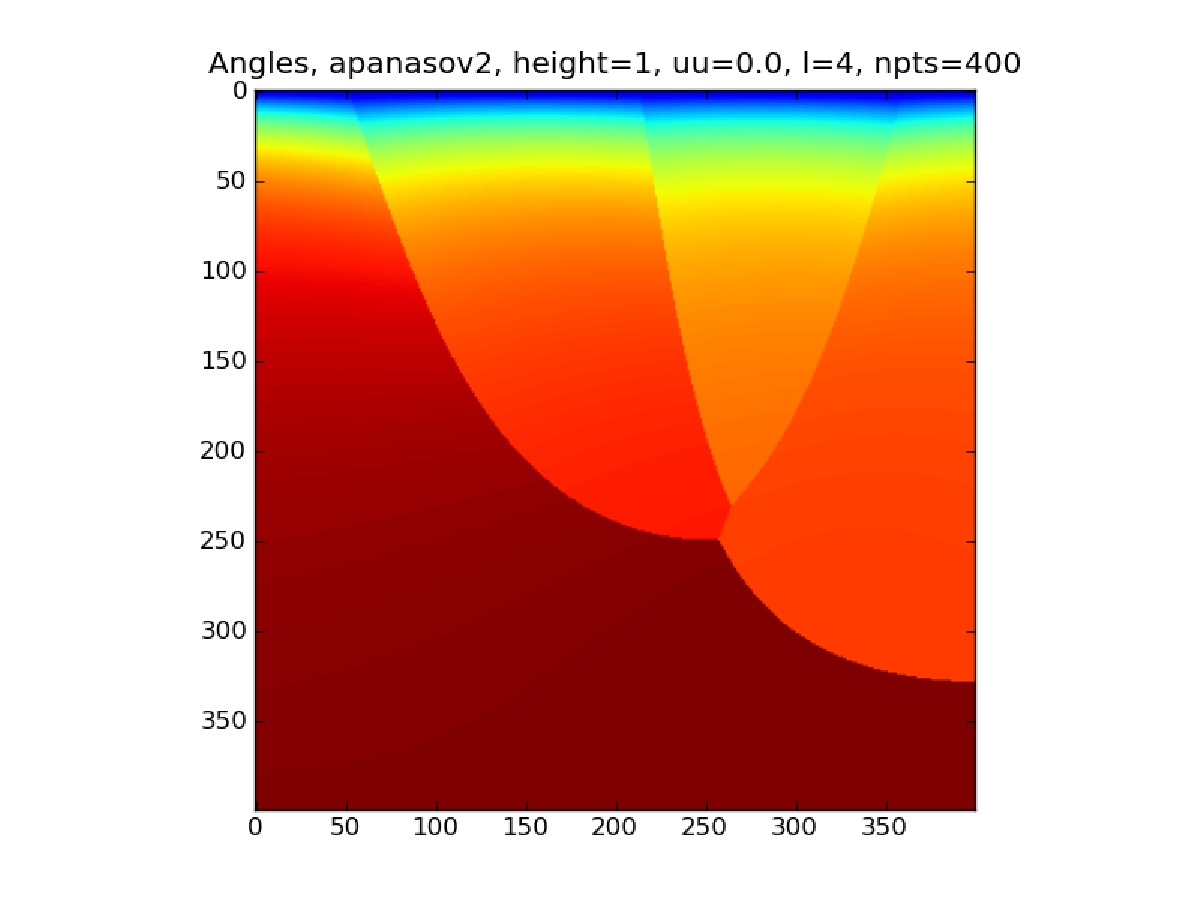}
  \includegraphics[width=8cm]{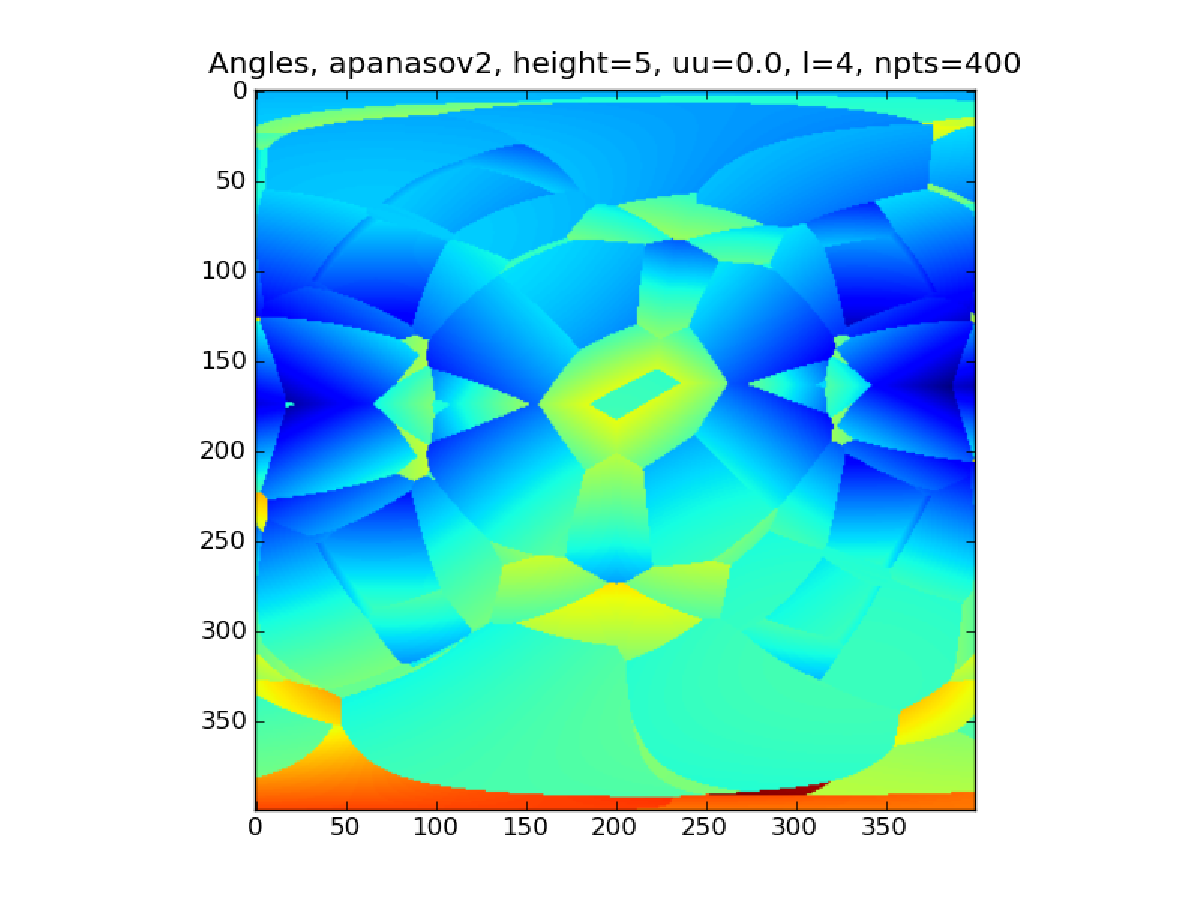}
  \includegraphics[width=8cm]{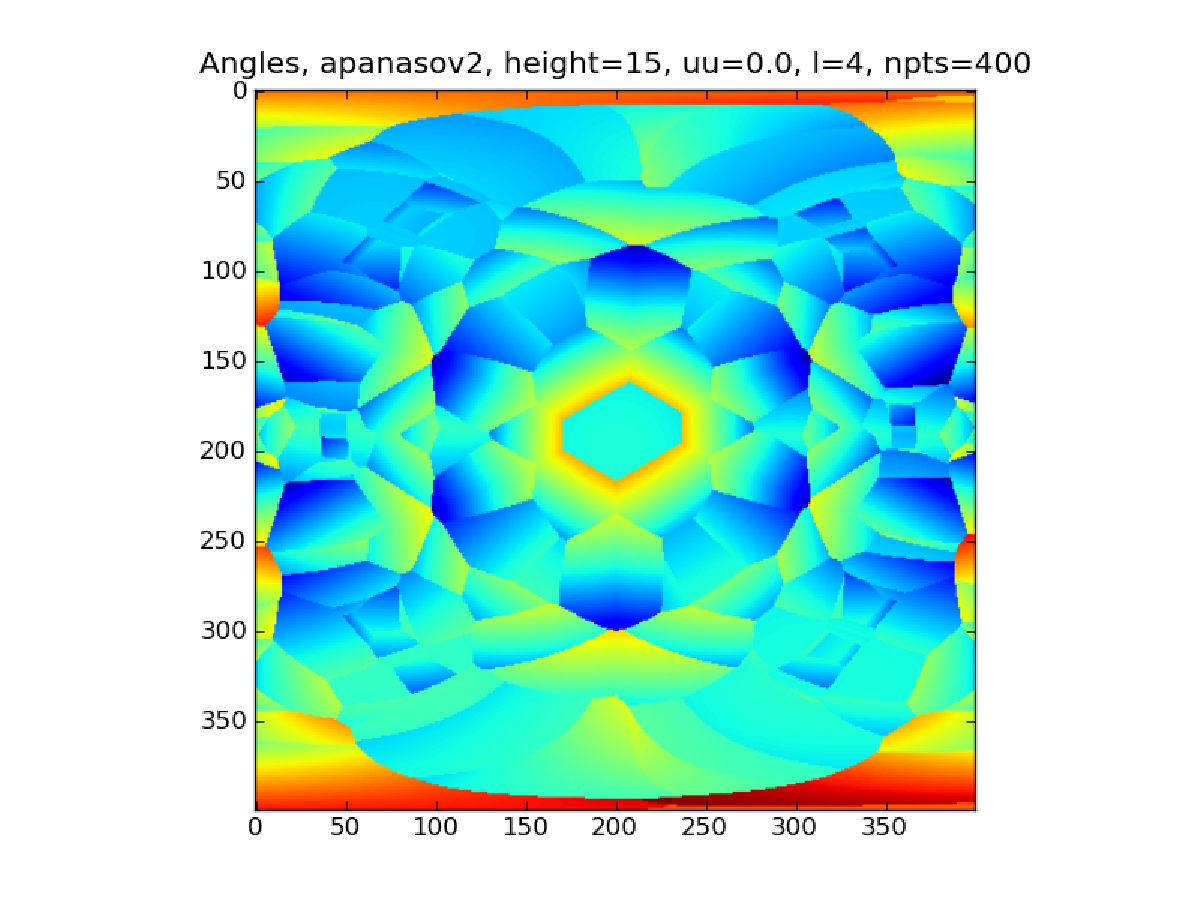}
  \includegraphics[width=8cm]{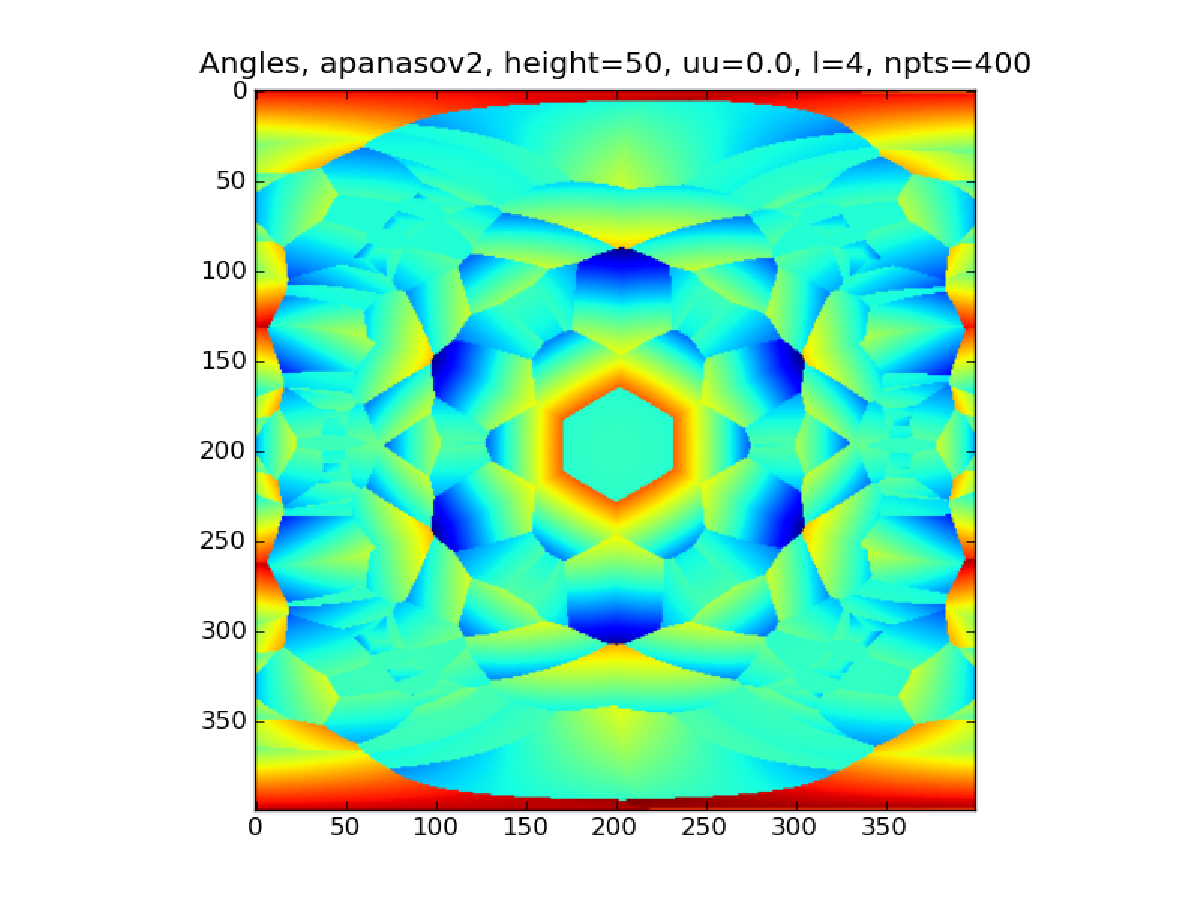}
  \caption{Computed frequency function, with translation coefficients $(1,1/2,0,1,1/2,0)$, 
for  observers at an increasing distance from the initial singularity}
  \label{fig:intensity2}
  \end{center}
\end{figure}

It should be noted that those results are  less certain that those obtained in dimension
$2+1$. This is due to the fact that we compute  the limit frequency function for a decreasing sequence
of finite domains of dependence approximating the domain under examination. In dimension $2+1$,
Theorem \ref{stability:thrm} and Proposition \ref{pr:2+1flat} indicate that the limit frequency function
is the frequency function of the limit, if the limit is the universal cover of a MGHFC spacetime. 
However, in dimension $3+1$, we only know by Theorem \ref{tm:non-flat} that the frequency function
of the limit is at least equal to the limit frequency function, and at most equal to three time the
limit frequency function. %(We also mentioned that we believe that this last bound might be improved to twice the limit frequency.) 
So the frequency functions computed here are  the limit frequency
(which is a well-defined notion for any domain of dependence, see Theorem \ref{tm:non-flat})
which differs from the ``real'' frequency function by a factor at most three.

The limit frequency function computed for the second choice of cocycle is depicted in Figure \ref{fig:intensity2}.
It is apparent  how the less symmetric cocycle leads to a distortion in the picture. The 
symmetry of degree six,  which is present in the linear part of the holonomy,  is readily apparent  in
Figure \ref{fig:intensity1}. In Figure \ref{fig:intensity2} it remains visible, but with differences in the size
of the corresponding parts of the picture. 

Even for this fairly simple example, it would be interesting to perform more powerful and complete computations,
for instance by computing the domain of dependence with all elements of the fundamental
group in a ball of radius larger than 4. It is conceivable that one would obtain somewhat
different pictures. Additionally, the picture should vary with the
position of the observer.  It should be simpler for an observer close to the initial
singularity, but become increasingly complex  as the observer moves away 
from it.

%\bibliography{bms_cm}{}
\bibliographystyle{amsplain}
\bibliography{/home/schlenker/papiers/outils/biblio}
\end{document}